\documentclass[reqno, 11pt]{amsart}

\usepackage{amsaddr}
\usepackage{a4wide}
   
\usepackage{amssymb, amsmath, amsthm}
\usepackage{graphicx, psfrag,color} 
\usepackage{url}      
\usepackage{upgreek}
\usepackage{hyperref}

\usepackage{algorithm, algorithmic}

\newcommand{\R}{\mathbb{R}}
\newcommand{\N}{\mathbb{N}}
\newcommand{\Rb}{\overline{\R}}
\renewcommand{\S}{\mathcal{S}} 

\newcommand{\Ls}[1]{\text{L}^{#1}}
\renewcommand{\L}[1]{\text{L}^{#1}(\Omega)}
\newcommand{\BV}{\text{BV}(\Omega)} 
 
\newcommand{\D}{\text{D}}
\newcommand{\id}{\textnormal{Id}}

\DeclareMathOperator*{\argmin}{\textnormal{argmin}}
\DeclareMathOperator*{\proj}{\text{proj}}

\newcommand{\E}[1]{\mathbf{E}\left(#1\right)}
\newcommand{\Prob}{\mathbb{P}}
\newcommand{\Var}[1]{\mathbf{Var}\left(#1\right)}

\newcommand{\eps}{\varepsilon}
\newcommand{\set}[1]{\left\{ #1 \right\}}
\newcommand{\abs}[1]{\left| #1 \right|}
\newcommand{\norm}[1]{\left\| #1 \right\|}
\newcommand{\inner}[2]{\left\langle #1, #2 \right\rangle}
\newcommand{\bigo}{\mathcal{O}}
\newcommand{\ra}{\rightarrow}

\theoremstyle{plain}

\newtheorem{theorem}{Theorem}[section]
\newtheorem{proposition}[theorem]{Proposition}

\theoremstyle{definition}
\newtheorem{remark}[theorem]{Remark}
\newtheorem{lemma}[theorem]{Lemma}

\title[Statistical Multiresolution Estimation for Variational
Imaging]{Statistical Multiresolution Estimation for Variational Imaging: With an
Application in Poisson-Biophotonics}

\begin{document}

\author{Klaus Frick}
\address{Institute for Mathematical Stochastics\\
University of G{\"o}ttingen\\
Goldschmidtstra{\ss}e 7, 37077 G{\"o}ttingen}
\email{frick@math.uni-goettingen.de}
\thanks{Correspondence to frick@math.uni-goettingen.de. Part of this work is
already published in \cite{FriMar12}}

\author{Philipp Marnitz}
\address{Institute for Mathematical Stochastics\\
University of G{\"o}ttingen\\
Goldschmidtstra{\ss}e 7, 37077 G{\"o}ttingen}
\email{stochastik@math.uni-goettingen.de}

\author{Axel Munk}
\address{Institute for Mathematical Stochastics\\
University of G{\"o}ttingen\\
Goldschmidtstra{\ss}e 7, 37077 G{\"o}ttingen\\ {\color{white} white space} \\
and} \address{Max Planck Institute for Biophysical Chemistry \\
Am Fa{\ss}berg 11, 37077 G{\"o}ttingen}
\email{munk@math.uni-goettingen.de}

\begin{abstract}
In this paper we present a spatially-adaptive method for image reconstruction
that is based on the concept of \emph{statistical multiresolution estimation} as
introduced in \cite{FriMarMun12}. It constitutes a variational
regularization technique that uses an $\ell_\infty$-type distance measure as
data-fidelity combined with a convex cost functional. The resulting convex
optimization problem is approached by a combination of an inexact alternating
direction method of multipliers and Dykstra's projection algorithm. We
describe a novel method for balancing data-fit and regularity that is fully automatic and
allows for a sound statistical interpretation. The performance of our estimation
approach is studied for various problems in imaging. Among others, this includes
deconvolution problems that arise in Poisson nanoscale fluorescence microscopy.
\end{abstract}


\maketitle 

\section{Introduction}

In this paper we are concerned with the reconstruction of an unknown
gray-valued image $u^0 \in \L{2}$ with $\Omega = [0,1]^2$ given the data
\begin{equation}\label{intro:data}
  Y_{ij} = (Ku^0)_{ij} + \eps_{ij},\quad 1\leq i \leq m, 1\leq j\leq n.
\end{equation}
For the moment, we assume that $\eps_{ij}$ are independent
and identically distributed Gaussian random variables with $\E{\eps_{11}}
= 0$ and $\E{\eps_{11}^2} = \sigma^2 > 0$ and that 
$K:\L{2}\ra \R^{m\times n}$ is a linear and bounded operator. $K$ is assumed to
model image acquisition and sampling at the same time, i.e. $(Ku)_{ij}$ is
assumed to be a sample at the pixel $(i\slash m, j\slash n)$ of a smoothed
version of $u$. Throughout the paper we will assume that $\sigma^2$ is known
(for reliable estimation techniques for $\sigma^2$ see e.g.
\cite{MunBisWagFre05} and references therein).  
 
A popular approach for computing a stable approximation of $u^0$ from the data
$Y$ given in the Gaussian model (\ref{intro:data}) consists in minimizing the
penalized least squares functional, i.e.
\begin{equation}\label{intro:rof}  
  \hat u(\lambda) \in \argmin_{u\in \L{2}}
  \frac{\lambda}{2}\sum_{i,j}\abs{(Ku)_{ij} - Y_{ij}}^2 + J(u)
\end{equation}
where $J:\L{2}\ra \Rb$ is a convex and lower-semicontinuous
regularization functional and $\lambda > 0$ a suitable multiplier. In the
seminal work
\cite{RudOshFat92}, for example, the authors proposed the \emph{total variation semi-norm} 
\begin{equation}\label{intro:tv}
  J(u) = \begin{cases} \abs{\D u}(\Omega) & \text{ if } u \in \BV \\
          + \infty & \text{ else}
         \end{cases}
\end{equation} 
as a penalization functional which has been a widely used model in imaging ever
since. Here, $\abs{\D u}(\Omega)$ denotes the total variation of the
(measure-valued) gradient of $u$ which coincides with $\int_\Omega \abs{\nabla
u}$ if $u$ is smooth. Numerous efficient solution methods for
\eqref{intro:optprob} \cite{Cha04,DobVog97,HinKun04} and
various modifications have been suggested so far (cf.
\cite{ChaLio97,FriSch07,OshBurGolXuWot05,ZhaBurBreOsh10}
to name but a few). In particular, in order to accelerate numerical algorithms
and to prevent oversmoothing the total variation semi-norm is often augmented to
\begin{equation}\label{intro:tvaug}
J(u) + \gamma \int_\Omega u^2
\end{equation}
with $\gamma \geq 0$.

The quadratic fidelity in \eqref{intro:rof} has an essential drawback:
The information in the residual is incorporated \emph{globally}, that is each
pixel value $(Ku)_{ij}-Y_{ij}$ contributes equally to the estimator $\hat
u(\lambda)$ \emph{independent of its spatial position}. In practical situations
this is clearly undesirable, since images usually contain features of different scales
and modality, i.e.\ constant and smooth portions as well as oscillating patterns
both of different spatial extent. A solution $\hat u(\lambda)$ of
\eqref{intro:rof} is hence likely to exhibit under- and oversmoothed regions at the same time.
  
Recently, \emph{spatially-adaptive} reconstruction approaches became popular
that are based on \eqref{intro:rof} with a locally varying regularization parameter, i.e.
\begin{equation}\label{intro:rofspat}
  \hat u(\lambda) \in \argmin_{u\in\L{2}} \frac{1}{2}\sum_{i,j}
  \lambda_{ij}\abs{(Ku)_{ij} - Y_{ij}}^2 + J(u).
\end{equation} 
However, the choice of the multiplier function $\lambda_{ij}$ is
subtle and different approaches have been suggested. See for instance
\cite{GilSocZee06,Gra09,DonHinRin11,HotMarStiDavKabMun12}.

In this paper we take a different route to achieve spatial adaption which allows
to ``localize'' any convex functional $J$ by minimizing it over a convex set
that is determined by the statistical extreme value behaviour of the residual
process. More precisely, we study estimators $\hat u$ of $u^0$ that are computed
as solutions of the convex optimization problem
\begin{equation}\label{intro:optprob}
  \inf_{u\in\L{2}} J(u)\quad\text{ s.t. }\quad \max_{S\in\S}
  \frac{c_S}{\sigma^2} \sum_{(i,j)\in S} \abs{(Ku)_{ij}- Y_{ij}}^2\leq 1.
\end{equation}
Here, $\S$ denotes a system of subsets of the grid
$G=\set{1,\ldots,m}\times\set{1,\ldots,n}$ and $\set{c_S~:~ S\in \S}$ is a set
of positive weights that govern the trade-off between data-fit and regularity
locally on each set $S\in\S$. Solutions of \eqref{intro:optprob} are special
instances of \emph{statistical multiresolution estimators (SMRE)} as studied in
\cite{FriMarMun12}. In this context the statistic $T:\R^{m\times n}\ra \R$
defined by \begin{equation}\label{intro:mre}
  T(v) = \max_{S\in\S} \frac{c_S}{\sigma^2} \sum_{(i,j)\in S}
  \abs{v_{ij}}^2,\quad v\in\R^{m\times n}
\end{equation}
is referred to as \emph{multiresolution (MR) statistic}. Summarizing, an SMRE
$\hat u$ of $u^0$ is an element with minimal $J$ among all candidate estimators
$u$ that satisfy the condition $T(Ku-Y)\leq 1$.

Special instances of \eqref{intro:optprob} have been studied recently: For the
case when $\S$ contains the entire domain $G$ only, it
has been shown in \cite{ChaLio97} that \eqref{intro:optprob} is equivalent to \eqref{intro:rof} if
$K$ satisfies certain conditions. As mentioned above, this approach is likely to
oversmooth small-scaled image features (such as texture) and/or underregularize
smooth parts of the image. An improved model was proposed in
\cite{BerCasRouSol03} where $\S$ is chosen to consist of a (data-dependent)
\emph{partition} of $G$ that is obtained in a preprocessing step (for the
numerical simulations in \cite{BerCasRouSol03}, Mumford-Shah segmentation is
considered). Under similar conditions on $K$ as in \cite{ChaLio97}, it was shown
in \cite{BerCasRouSol03} that \eqref{intro:optprob} is equivalent to
\eqref{intro:rofspat} where $\lambda_{ij}$ is constant on each $S\in\S$. This
approach was further developed in \cite{AlmBalCasHar08} where a subset $S\subset
G$ is fixed and afterwards $\S$ is defined to be the collection of all
translates of $S$ (in fact, the authors study the convolution of the
squared residuals with a discrete kernel). The authors propose a proximal point
method for the solution of \eqref{intro:optprob}. This approach of local constraints w.r.t. a
window (or kernel) of \emph{fixed size} was also studied in
\cite{FacAlmAujCas09} for irregular sampling and regularization functionals
other than the total variation were considered. In particular, it is observed
that the difference between results obtained by using the total
variation penalty \eqref{intro:tv} and the Dirichlet-energy (integrated squared
norm of the derivative) is not so big when using local constraints. This is in
accordance to findings in \cite{FriMarMun12a} for one-dimensional signals. In
\cite{DonHinRin11} the model of \cite{AlmBalCasHar08} was studied in the continuous function space setting. Moreover the authors in
\cite{DonHinRin11} provided a fast algorithm for the solution of the
constrained optimization problem based on the hierarchical decomposition scheme
\cite{TadNezVes04} combined with the unconstrained problem \eqref{intro:rofspat}.

In this paper, we propose a novel, automatic selection rule for the weights
$c_S$ based on a statistically sound method that is applicable for any
pre-specified, deterministic system of subsets $\S$. We are particularly
interested in the case when $\S$ constitutes a highly redundant collection of subsets of $G$ consisting of overlapping subsets of different
scales. This is a substantial extension to the approaches in
\cite{AlmBalCasHar08,FacAlmAujCas09,DonHinRin11} that only consider one fixed
(pre-defined) scale. Our approach will amount to select a single parameter
$\alpha \in [0,1]$ with the interpretation that the true signal $u^0$ satisfies
the constraint in \eqref{intro:optprob} with probability $\alpha$. From the
definition of \eqref{intro:optprob} it is then readily seen that
$\Prob\left(J(\hat u) \leq J(u^0)\right) \geq \alpha$
for any solution $\hat u$ of \eqref{intro:optprob}. In other words, our method
controls the probability that the reconstruction $\hat u$ is at least as smooth (in the
sense of $J$) as the true image $u^0$. To this aim, it will be necessary to gain
stochastic control on the null-distribution $T(\eps)$, where $\eps =
\set{\eps_{ij}}$ is a lattice of independent $\mathcal{N}(0,\sigma^2)$-
distributed random variabels. 

Moreover, for the efficient solution of \eqref{intro:optprob} we extend the
algorithmic ideas in \cite{FriMarMun12} and propose a combination of an
\emph{inexact} alternating direction method of multipliers (ADMM)
\cite{ElmGol94,ChaPoc11} with Dykstra's projection algorithm \cite{BoyDyk86}.
Finally, we indicate how our approach can be applied to image deblurring
problems in fluorescence microscopy where the observed data does not fit into
the white noise model \eqref{intro:data} but where one usually assumes that
independently
\begin{equation}\label{poisson:data}
Y_{ij} \sim \textnormal{Pois}\left((Ku)_{ij}\right),\quad 1\leq i\leq m,1\leq
j\leq n 
\end{equation}
Here, $\textnormal{Pois}(\beta)$ stands for the Poisson distribution with
parameter $\beta > 0$. We mention that similar models occur in positron emission
tomography (cf. \cite{VarSheKau85}) and large binocular telescopes (cf.
\cite{BerBoc00}) and we claim that our method can be useful there as well. We
apply Anscombe's transform to transform the Poisson data to normality.
Furthermore we present a modified version of the ADMM to solve the resulting
variant of \eqref{intro:optprob}. We finally illustrate the capability of our
approach by numerical examples: Image denoising, deblurring and inpainting and
deconvolution problems that arise in nanoscale fluorescence microscopy.

In the following, we denote by $\abs{S}$ the cardinality of $S\in \S$. We
often refer to $\abs{S}$ as the \emph{scale} of $S$. We assume that $m,n\in\N$
are fixed and denote by $\inner{\cdot}{\cdot}$ and $\norm{\cdot}$ the Euclidean
inner-product and norm on $\R^{m\times n}$ and by $\norm{u}_{\Ls{2}}$ the
$\Ls{2}$-norm of $u$. For a convex functional $J:\L{2}\ra \Rb$ the
subdifferential $\partial J(u)$ is the set of all $\xi\in\L{2}$ such that
$J(v)\geq J(u) + \int_\Omega \xi(v-u)$ for all $v\in\L{2}$. The
Bregman-distance between $v,u\in\L{2}$ w.r.t. $\xi \in \partial J(u)$ is
defined by
\begin{equation*}
D_J^\xi(v,u) = J(u) - J(v) - \int_\Omega \xi(u-v)\geq 0.
\end{equation*}
If additionally $\eta \in \partial J(v)$ we define the symmetric Bregman
distance by $D_J^{\text{sym}}(u,v) = D_J^\xi(v,u) + D_J^\eta(u,v)$. By $J^*$ we
denote the Legendre-Fenchel transform of $J$, i.e. $J^*(q) = \sup_{u\in\L{2}}
\int_\Omega u q - J(u)$.  We finally note that it would not be restricitve (yet
less intuitive) to replace $\L{2}$ by any other separable Hilbert space.

\section{Statistical Multiresolution Estimation}\label{smre}

We review sufficient conditions that guarantee existence of
SMREs, solutions of \eqref{intro:optprob} that is. To this end, we rewrite
\eqref{intro:optprob} into an equality constrained problem and study the
corresponding augmented Lagrangian function (Section \ref{smre:not}).
Moreover, we address the important question on how to choose the \emph{scale
weights} $c_S$ automatically in Section \ref{smre:choice}. Finally, we
discuss different choices for the system $\S$ that have proved feasible in
practice in Section \ref{smre:test}.

\subsection{Existence of SMRE}\label{smre:not}

For the time being, let $\set{c_S~:~ S\in \S}$ be a set of positive real numbers.
We rewrite \eqref{intro:optprob} to an equality constrained problem by
introducing the slack variable $v\in \R^{m\times n}$. To be more precise, we aim
for the solution of
\begin{equation}\label{smre:decomp}
  \inf_{u\in\L{2}, v\in \R^{m\times n}} J(u) + H(v) \quad\text{ s.t. }\quad Ku -
  v = 0
\end{equation}
where $H$ denotes the indicator function on the feasible set $\mathcal{C}$ of
\eqref{intro:optprob}, i.e.
\begin{equation}
\label{smre:feasset}
  \mathcal{C} = \set{v\in \R^{m\times n}~:~ T(v - Y) \leq 1 } \quad\text{
  and} \quad  H(v) = \begin{cases} 0 & \text{ if } v\in \mathcal{C} \\ +\infty &
  \text{ else } \end{cases}.
\end{equation}
Problems of type \eqref{smre:decomp} were studied extensively in
\cite[Chap. III]{ForGlo83}. There, Lagrangian multiplier methods are employed to
solve \eqref{smre:decomp}. Recall the definition of the \emph{augmented
Lagrangian} of \eqref{smre:decomp}:
\begin{equation}\label{smre:lagr}  
  L_\lambda (u,v;p) = \frac{1}{2\lambda} \norm{Ku - v}^2 + J(u) + H(v) +
  \inner{p}{Ku - v},\quad \lambda > 0.
\end{equation}
Here $p\in\R^{m\times n}$ denotes the Lagrange multiplier for the linear
constraint in \eqref{smre:decomp}. Note that $L_\lambda$ equals the ordinary
Lagrangian $L(u,v;p) = J(u) + H(v) + \inner{p}{Ku-v}$ augmented by the quadratic
term $\norm{Ku-v}^2\slash (2\lambda)$ that fosters the fulfillment of the linear
constraint in \eqref{smre:decomp}. 

It is well known that the saddle-points of $L$ and $L_\lambda$ coincide (cf.
\cite[Chap III Thm. 2.1]{ForGlo83}) and that existence of a saddle point of
$L_\lambda$ follows from existence of solutions of \eqref{smre:decomp} together
with constraint qualifications of the MR-statistic $T$. One typical example for
the latter is given in Proposition \ref{smre:exist}. The result is rather
standard and can be deduced e.g. from \cite[Chap III, Prop 3.1 and Thm.
4.2]{EkeTem76} (cf. also \cite[Chap III]{ForGlo83})

\begin{proposition}\label{smre:exist}
Assume that  \eqref{smre:decomp} has a solution $(\hat u, \hat v)\in
\L{2}\times \R^{m\times n}$ and that there exists $\bar u \in \L{2}$ such that $J(\bar u) < \infty$ and
  $T(K\bar u-Y) < 1$ \emph{(Slater's constraint qualification)}. Then, there
  exists $\hat p\in \R^{m\times n}$ such that $(\hat u, \hat v, \hat p)$ is a saddle point of $L_\lambda$, i.e.
 \begin{equation*}
      L_\lambda(\hat u,\hat v; p) \leq L_\lambda(\hat u, \hat v; \hat p) \leq
  L_\lambda( u,  v; \hat p),\quad \forall\left( u\in \L{2}, v,p \in \R^{m\times
  n}\right).
  \end{equation*}
\end{proposition}

\begin{remark}
  \begin{enumerate}
    \item If $\hat u\in \L{2}$ and $\hat v,\hat p\in\R^{m\times n}$ are as in
    Proposition \ref{smre:exist}, then $\hat u$ is an SMRE, i.e. it solves
    \eqref{intro:optprob}. Moreover, the following extremality relations hold:
    \begin{equation*}
    -K^*\hat p\in \partial J(\hat u),\quad \hat p \in \partial
    H(\hat v)\quad\text{ and }\quad K\hat u = \hat v.
    \end{equation*}
    \item Slater's constraint qualification is for instance satisfied if
    the set 
    \begin{equation*}
    \set{Ku~:~ u\in \L{2}\text{ and } J(u)< \infty}
    \end{equation*} is dense in $\R^{m\times n}$.
    \item If $J$ is chosen to be the total variation semi-norm \eqref{intro:tv},
    then a sufficient condition for the existence of solutions of \eqref{smre:decomp}
     will be that there exists $(i,j)\in S$ for some $S\in\S$ such that
    $(K\mathbf{1})_{ij}\not = 0$, where $\mathbf{1}\in\L{2}$ is the constant $1$-function. This is immediate
    from Poincar\'e's inequality for functions in $\BV$ (cf. \cite[Thm.5.11.1]{Zie89}). 
  \end{enumerate}
\end{remark}

\subsection{An a priori parameter selection method}\label{smre:choice}

The choice of the \emph{scale weights} $c_S$ in \eqref{intro:optprob} is of
utmost importance for they determine the trade-off between smoothing and
data-fit (and hence play the role of spatially local regularization parameters).
We propose a statistical method that is based on quantile values of extremes of transformed
$\chi^2$ distributions.

We pursue the strategy of controlling the probability that the true signal $u^0$
satisfies the constraint in \eqref{intro:optprob}. To this end, observe that
for $S\in \S$ the random variable
\begin{equation*}
  t_S(\eps) = \sigma^{-2} \sum_{(i,j) \in S}\eps_{ij}^2
\end{equation*}
is $\chi^2$-distributed with $\abs{S}$ degrees of freedom (d.o.f.). With this
notation, it follows from \eqref{intro:data} that $u^0$ satisfies the
constraints in \eqref{intro:optprob} if 
\begin{equation*} 
  T(Ku^0 -Y) = T(\eps) = \max_{S\in\S} c_S t_S(\eps)\leq 1. 
\end{equation*}
Additionally, we require that the maximum above is \emph{balanced} in the sense
that the probability for $c_S t_S(\eps) > 1$ is equal for each $S\in \S$.  

To achieve this, we first aim for transforming $t_S(\eps)$ to normality. It was
shown in \cite{HawWix86} that the \emph{fourth root transform}
$\sqrt[4]{t_S(\eps)}$ is approximately normal with mean and variance
\begin{equation*}
  \mu_S = \sqrt[4]{\abs{S}-0.5} \quad \text{ and }\quad \sigma^2_S =
  \left(8\sqrt{\abs{S}}\right)^{-1},
\end{equation*}
respectively. The fourth root transform outperforms other power transforms in
the sense that the Kullback-Leibler distance to the normal distribution is
minimized, see \cite{HawWix86}. In particular, it was stressed in
\cite{HawWix86} that the approximation works well for small d.o.f. Next,
we consider the extreme value statistic
\begin{equation}\label{smre:trans}
  \max_{S\in\S} \frac{\sqrt[4]{t_S(\eps)} - \mu_S}{\sigma_S}.
\end{equation}
We note that due to the transformation of the random variable $t_S(\eps)$ to
normality each scale contributes equally to the supremum in \eqref{smre:trans}. Hence a parameter
choice strategy based on quantile values of the statistic \eqref{smre:trans} is
likely to balance the different scales occurring in $\S$. We make this precise
in the following 

\begin{proposition}\label{smre:conf}
For $\alpha\in (0,1)$ and $S\in \S$ let $q_\alpha$ be the $\alpha$-quantile of
the statistic \eqref{smre:trans} and set $c_S = (q_\alpha \sigma_S +
\mu_S)^{-4}$. Then, any solution $\hat u$ of \eqref{intro:optprob} satisfies
\begin{equation}\label{smre:alphareg}
  \Prob(J(\hat u) \leq J(u^0)) \geq \alpha.
\end{equation}
\end{proposition}

\begin{proof}
From \eqref{intro:data} and from the monotonicity of
  the mapping $x\mapsto \sqrt[4]{x}$ it follows that
  \begin{eqnarray*}
      \Prob\left(T(Ku^0 - Y) \leq 1 \right) & = & \Prob\left(c_S t_S(\eps) \leq
      1\;\forall S\in \S\right) \\
      & = &  \Prob\left(\sqrt[4]{t_S(\eps)} \leq q_\alpha \sigma_S +
      \mu_S\;\forall S\in \S\right) \\
      & = &  \Prob\left(\max_{S\in\S} \frac{\sqrt[4]{t_S(\eps)} -
      \mu_S}{\sigma_S} \leq q_\alpha \right)  =  \alpha.
  \end{eqnarray*}
  In other words, the constants $c_S$ are chosen such that the true signal
  $u^0$ satisfies the constraints with probability $\alpha$. By the fact that
  $\hat u$ is a solution of \eqref{intro:optprob} it follows that $\Prob(T(Ku^0
  - Y) \leq 1) \leq \Prob(J(\hat u)\leq J(u^0))$.
\end{proof}   

\begin{remark}\label{smre:alpharegrem}       
By the rule $c_S = (q_\alpha \sigma_S + \mu_S)^{-4}$ in Proposition
  \ref{smre:conf} the problem of selecting the \emph{set} of scale weights $c_S$
  is reduced to the question on how to choose the \emph{single} value
  $\alpha\in(0,1)$. The probability $\alpha$ plays the role of an universal
  regularization parameter and allows for a precise statistical interpretation:
  It constitutes a lower bound on the probability that the SMRE $\hat u$ is more
  regular (in the sense of $J$) than the true object $u^0$. Moreover, one has
  that
  \begin{equation*}
  c_S t_S(\eps)>1\quad\Leftrightarrow\quad \frac{\sqrt[4]{t_S(\eps)} -
  \mu_S}{\sigma_S} > q_\alpha,\quad\text{ for all }S\in\S.
  \end{equation*}
\end{remark}

Note, that the constraint in \eqref{intro:optprob} can be rewritten into
\begin{equation*}
\frac{1}{\abs{S}} \sum_{(i,j)\in S}\abs{(Ku)_{ij} - Y_{ij}}^2 \leq
\frac{\sigma^2}{c_s \abs{S}},\quad \text{ for all }S\in \S.
\end{equation*}
Since $\E{\abs{S}^{-1}\sum_{(i,j)\in S} \eps_{ij}^2} = \sigma^2$ and
$\Var{\abs{S}^{-1}\sum_{(i,j)\in S} \eps_{ij}^2} = 2\sigma^4 \abs{S}^{-1}$
the factor $1 \slash (c_S \abs{S})$ can be considered as a
\emph{relaxation} parameter, that takes into account the uncertainty of
estimating the variance of the residual on finite scales $\abs{S}$. Put
differently, it is expected to be large on small scales and approaches $1$ as $\abs{S}$ increases.
This is illustrated in Figure \ref{smre:quantiles}: Here the quantity
$1\slash (c_S \abs{S})$ is depicted for the system $\S_0$ of all squares with
sidelengths up to $20$ (left panel) and the system $\S_2$ of all
dyadic squares (middle panel) in an $341\times 512$ image for $\alpha = 0.2$
('+') and $\alpha = 0.9$ ('o'). It becomes clear, that only on the smallest
scales there are non-negligible differences between the scale weights for 
$\S_0$ and $\S_2$. Also our numerical experiments confirm,  that reconstruction
results do not differ very much for different choices of $\alpha$.

 \begin{figure}[h!]
  \begin{center}
   \includegraphics[width = 0.325\textwidth]{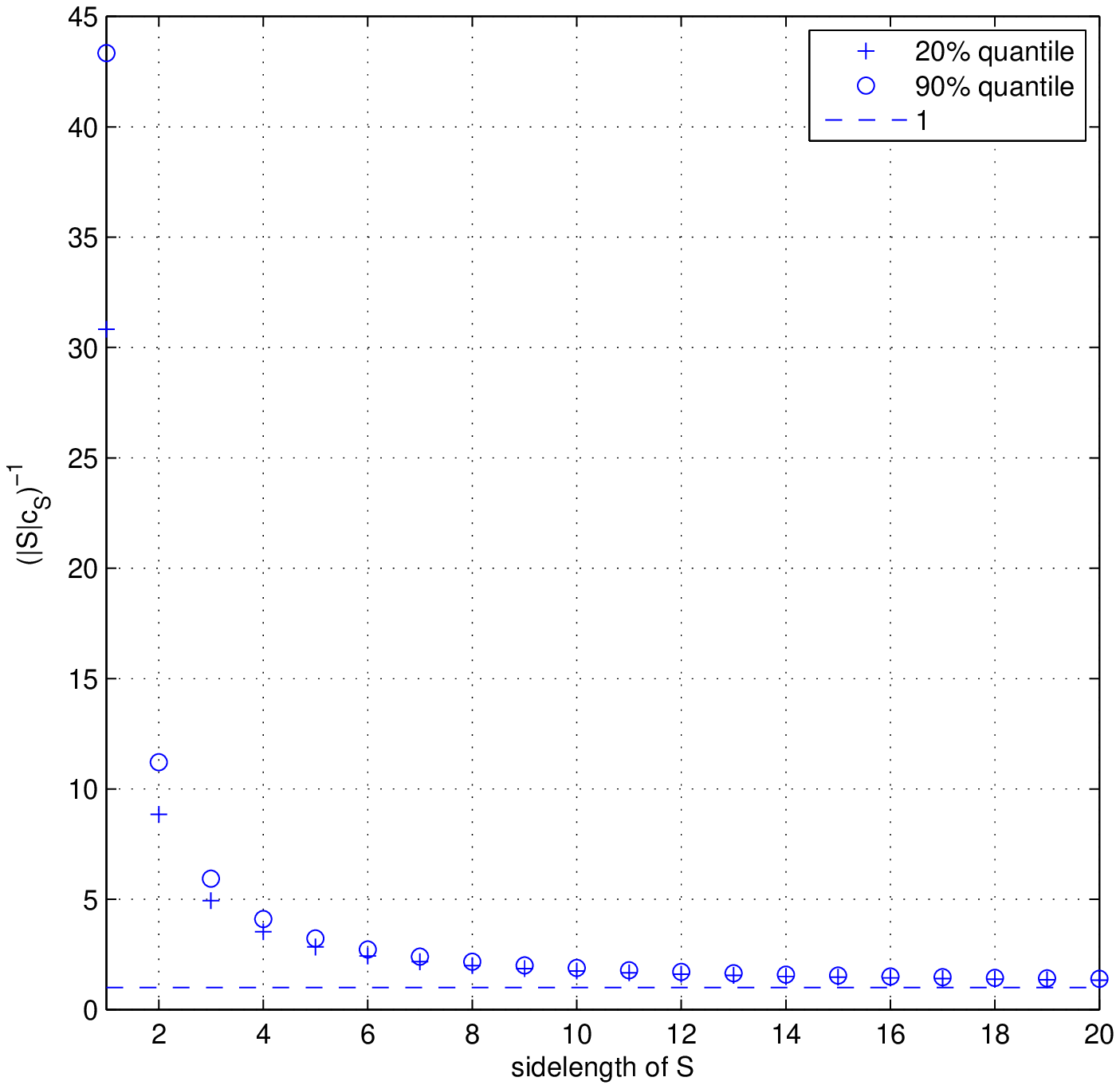} 
   \includegraphics[width = 0.325\textwidth]{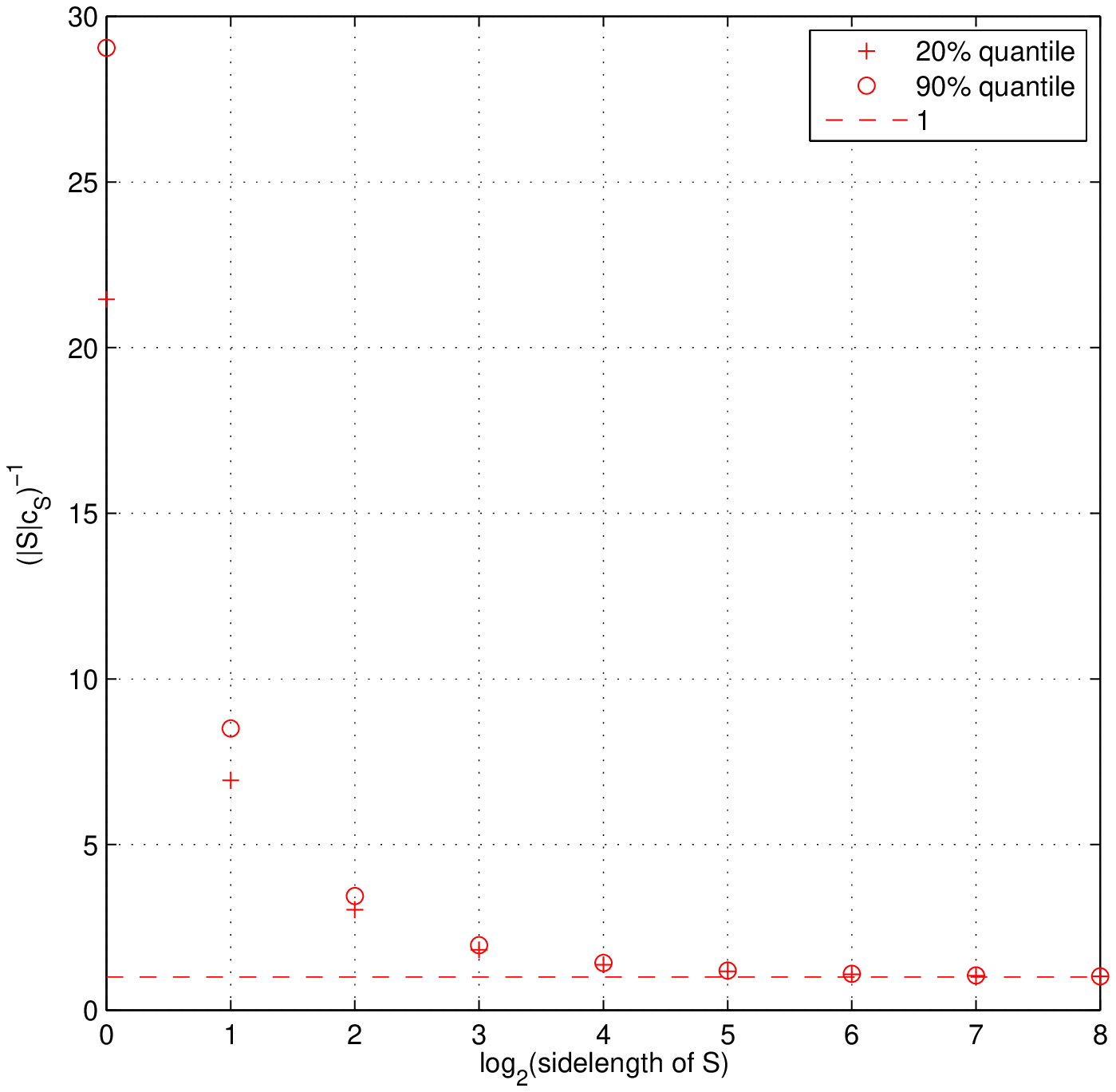} 
   \includegraphics[width = 0.325\textwidth]{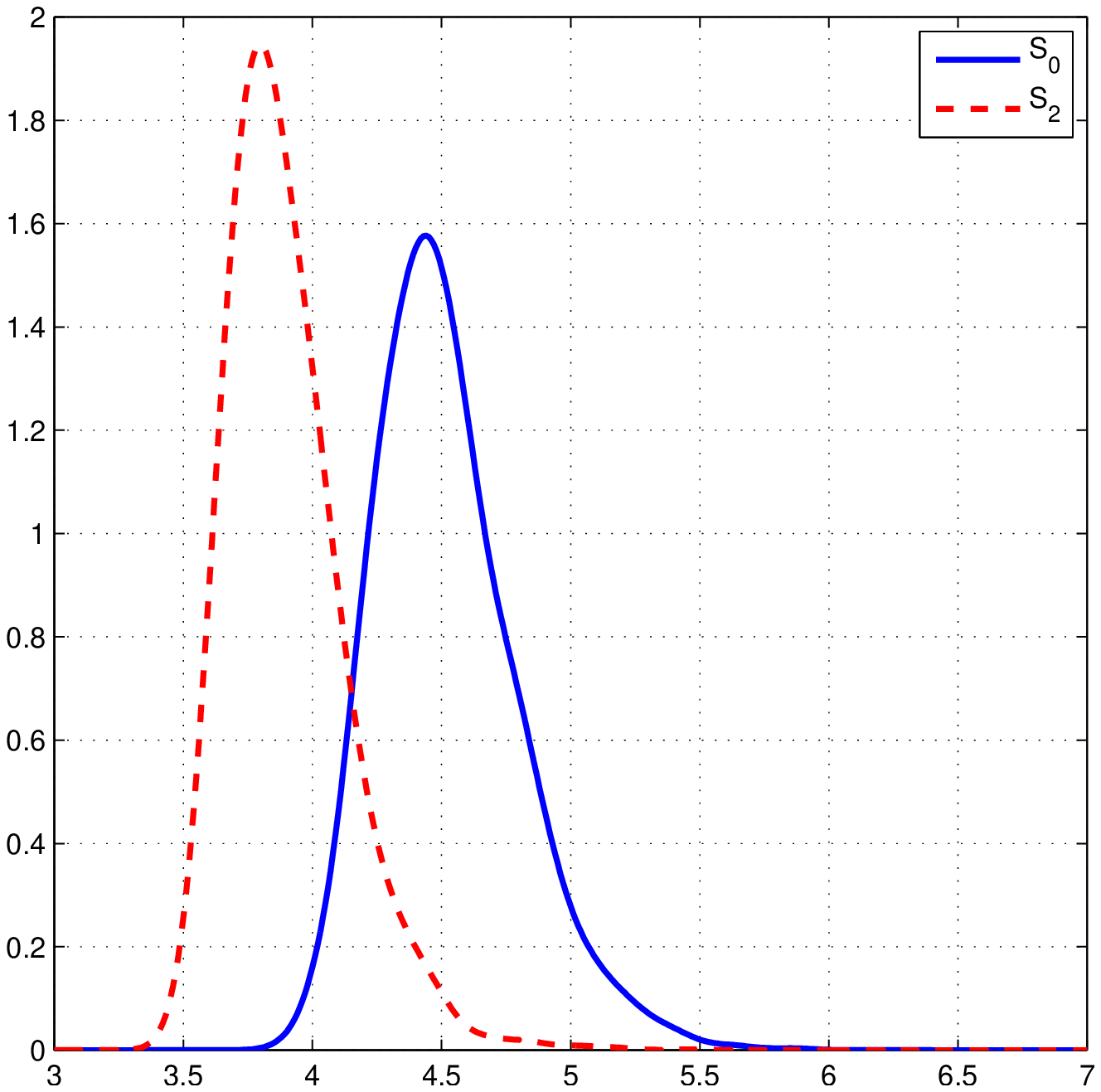}   
   \end{center}
   \caption{Left and middle: Dependence of relaxation factor 
   $(c_S\abs{S})^{-1}$ on $\alpha$ and scale $\abs{S}$. Right: Empirical density
   of the statistic \eqref{smre:trans}. All values are computed for a $341\times 512$
  image and w.r.t. to $\S_0$ (solid) and $\S_2$ (dashed).}\label{smre:quantiles}
\end{figure} 

We note that in \cite{AlmBalCasHar08} and \cite{DonHinRin11} the authors propose
relaxation parameters for the case when $\S$ consists of the translates of a
window of fixed size. In \cite{AlmBalCasHar08} the authors fix such a parameter,
$1.01$ say, and determine the corresponding window size by heurisitc reasoning.
In \cite{DonHinRin11} the authors give for a fixed window size $\abs{S}$ a
formula for a relaxation parameter that uses moments of the extreme value statistic of
independent $\chi^2$ random variables with $\abs{S}$ degrees of freedom. We note
that these methods can not be generalized to systems $\S$ that contains sets of different scales 
case in a straightforward manner. Our selection rule for the weights $c_S$ is
designed such that different scales are balanced appropriately. Hence our
approach is a \emph{multi-scale} extension of the (single-scale) methods in
\cite{AlmBalCasHar08, DonHinRin11}.

\begin{remark}\label{smre:remdependen}
It is important to note that the random variable $t_S(\eps)$ and $t_{S'}(\eps)$
are independent if and only if $S\cap S'=\emptyset$. As we do not
assume that $\S$ consists of pairwise disjoint sets, \eqref{smre:trans}
constitutes an extreme value statistic of \emph{dependent} random variables.
Except for special cases, little is known about the distribution of such
statistics (see e.g. \cite{Kab10,KabMun08} for asymptotic results).
It is an open and interesting problem to investigate the asymptotic properties
of the distribution of the statistic in \eqref{smre:trans}.

 In practice, the quantile values $q_\alpha$ in Proposition \ref{smre:conf} are
 derived from the empirical distribution of \eqref{smre:trans}. The right panel
 in Figure \ref{smre:quantiles} shows the empiricial density of the statistic
\eqref{smre:trans} for $m = 341$ and $n=512$ and the systems $\S_0$ (solid) and 
$\S_2$ (dashed) (for our simulations in Figure \ref{smre:quantiles} we used
$5000$ trials).
\end{remark}

\subsection{On the choice of $\S$}\label{smre:test}

In the previous section we addressed the question on how to select the scale
weights $\set{c_S}_{S\in\S}$ for a given system of subsets $\S$ of the grid $G$.
Altough it is not the primary aim of this paper to advocate a particular
systems $\S$, we will now comment on possible determinants for a rational
choice of $\S$.

\begin{figure}[t!]
\begin{center}
\includegraphics[width = 0.49\textwidth]{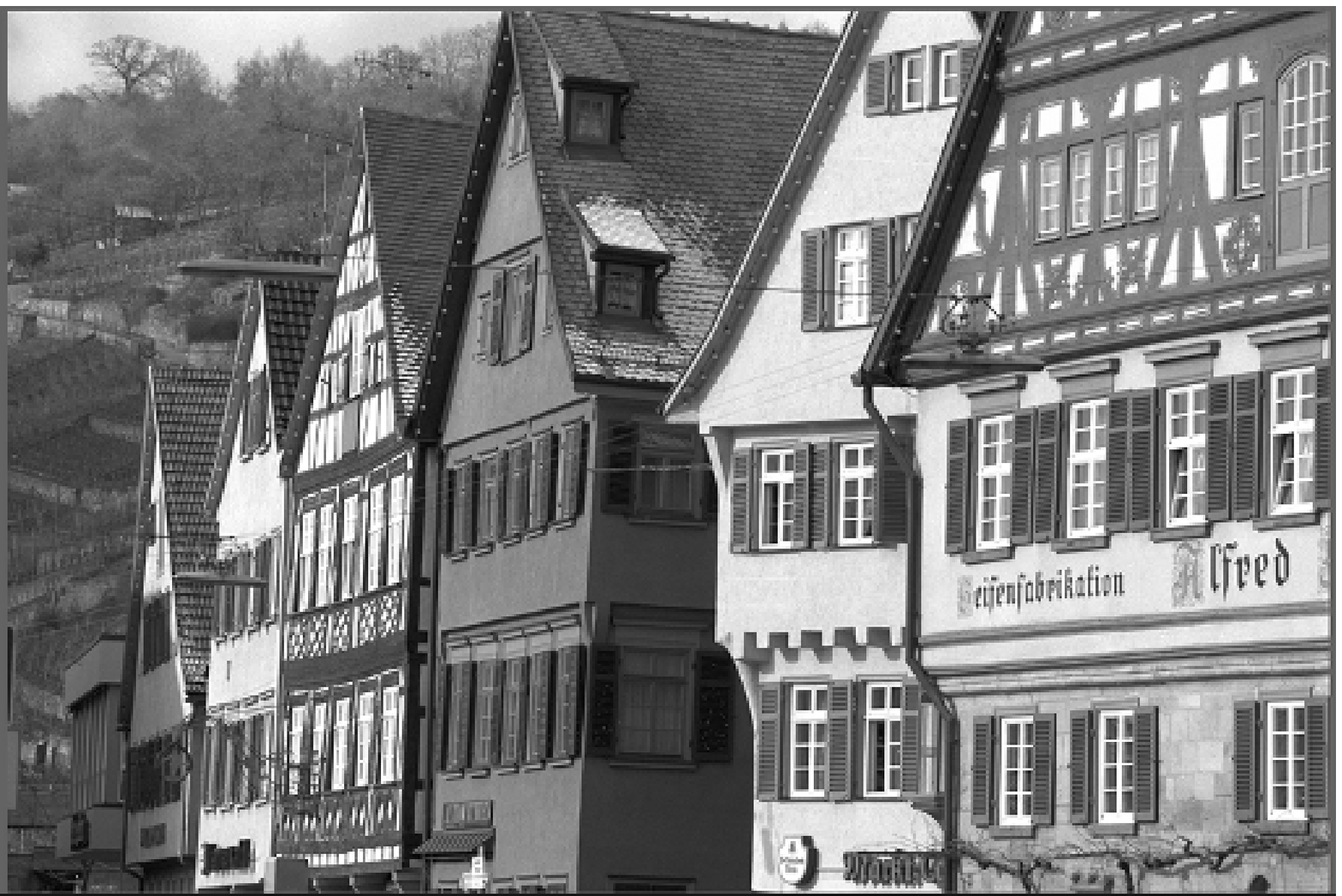}
\hspace{0.004\textwidth}
\includegraphics[width = 0.49\textwidth]{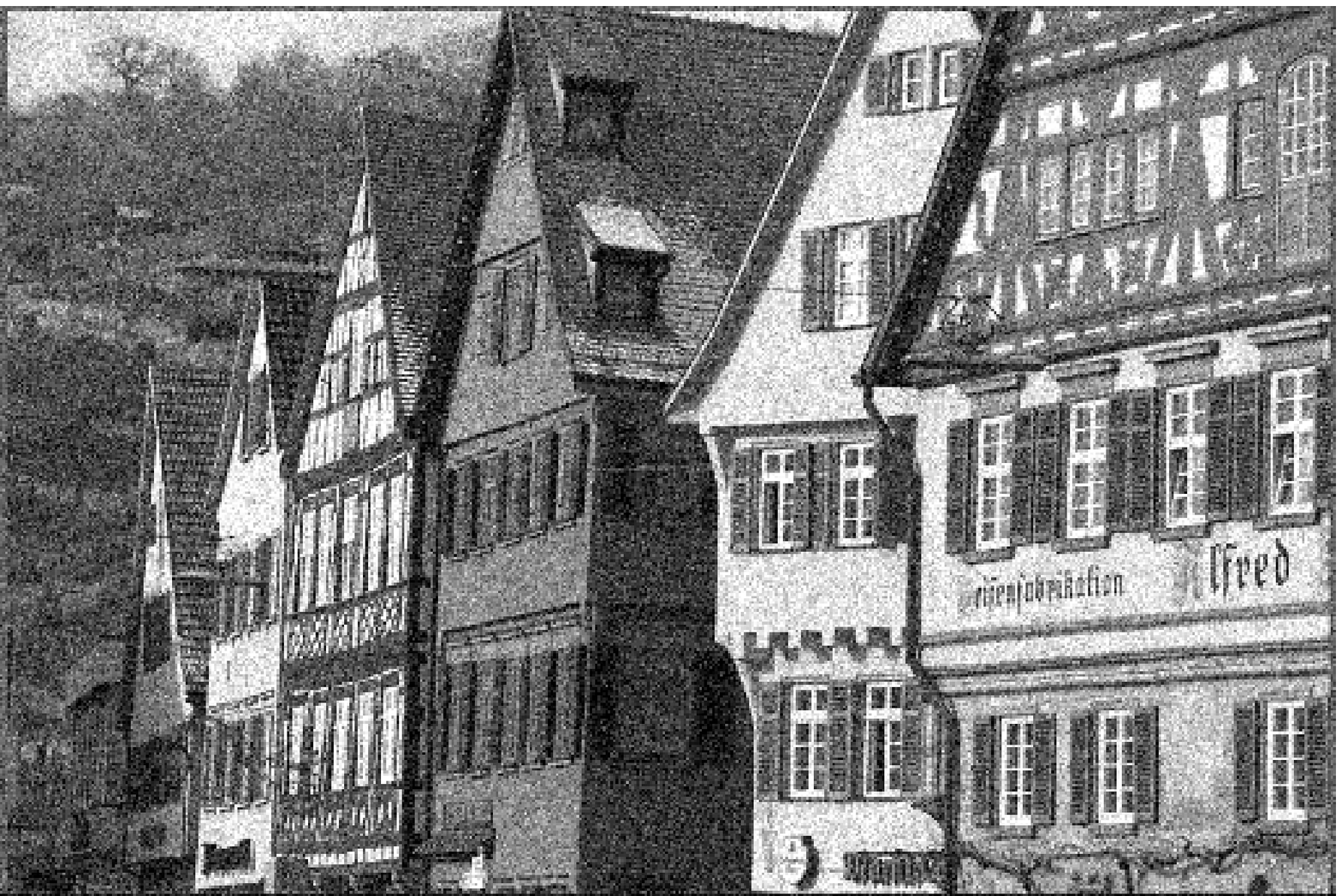}
\end{center}
\caption{Left: true signal $u^0$. Right: noiys data $Y$ with $\sigma =
0.1$}\label{smre:fig_data}
\end{figure}

\begin{figure}[h!]
\begin{center}
\includegraphics[width = 0.49\textwidth]{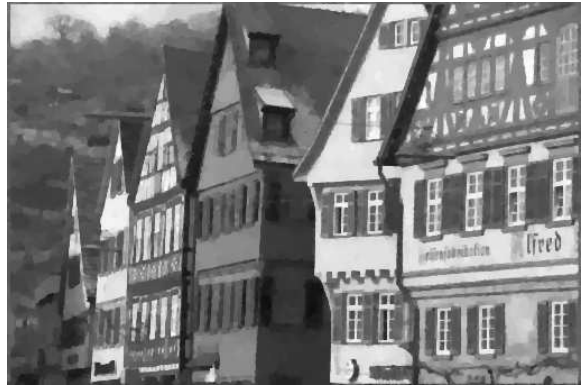}   
\caption{Solution of \eqref{intro:rof} $\hat u(\lambda)$ with $\lambda =
0.075$.}\label{smre:rofest}
\end{center}
\end{figure}

On the one hand, $\S$ should be chosen rich enough to resolve local features of
the image sufficiently well at various scales. On the other hand, it is
desirable to keep the cardinality of $\S$ small such that the optimization problem in
\eqref{intro:optprob} remains solvable within reasonable time. As a consequence
of this, \emph{a priori information} on the signal $u^0$ should be employed in
practice in order to delimit a suitable system $\S$ (e.g. the range of scales to
be used). Furthermore we note that for guaranteeing that the extreme
value statistic \eqref{smre:trans} does not degenerate (as $m,n$ and
the cardinality of $\S$ increase), $\S$ typically has to satisfy certain entropy
conditions (see e.g. \cite{FriMarMun12a}). We stress that it is a challenging
and interesting task to extend these results to \emph{random} (data-driven) systems
$\S$. It is well known that such methods can yield good results in practice (see
e.g. \cite{BerCasRouSol03}).

Here, we discuss two different choices of $\S$, namely: 
\begin{enumerate}
  \item the set of \emph{all discrete squares} in $G$: for computational
  reasons usually subsystems are considered. We found the subset consisting of
  all squares with sidelengths up to $20$ to be efficient. We denote
  this subset henceforth by $\S_0$.
  \item the set $\S_2$ of \emph{dyadic partitions} of $G$. For a quadratic grid
  $G$ with $m=n=2^r$ the system $\S_2$ is obtained by recursively splitting the grid into four equal
squares until the lowest level of single pixels is reached. To be more precise, 
\begin{equation*}
  \S_2 = \bigcup_{l=1}^{r}
  \set{\set{k2^{l},\ldots,(k+1)2^{l}}^2~:~ k =
  0,\ldots,2^{r-1}}.
\end{equation*}
For general grids $G$ the left and lower most squares are clipped accordingly.
\end{enumerate}

\begin{figure}[h!]
\begin{center}
\includegraphics[width = 0.49\textwidth]{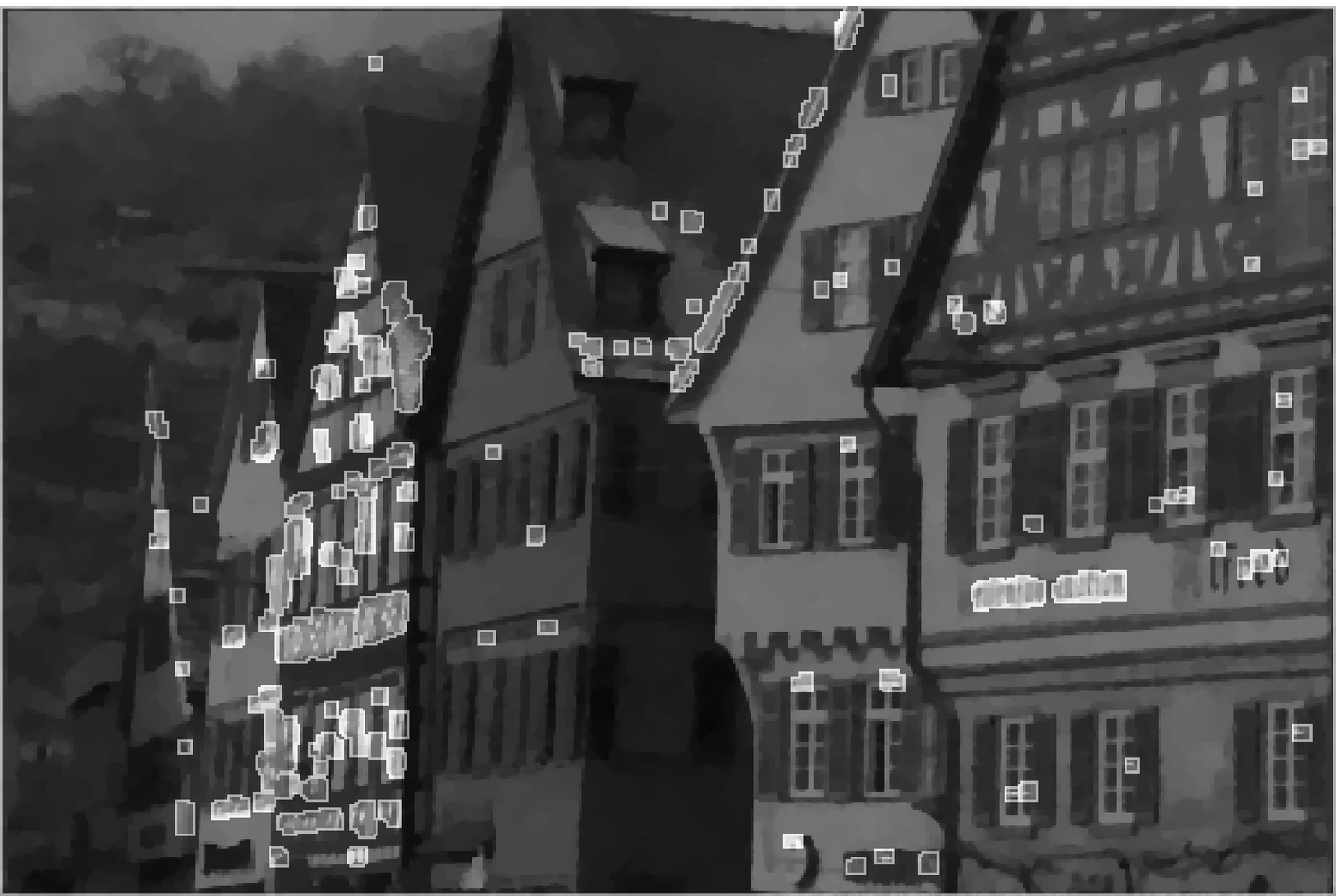}
\hspace{0.004\textwidth}
\includegraphics[width = 0.49\textwidth]{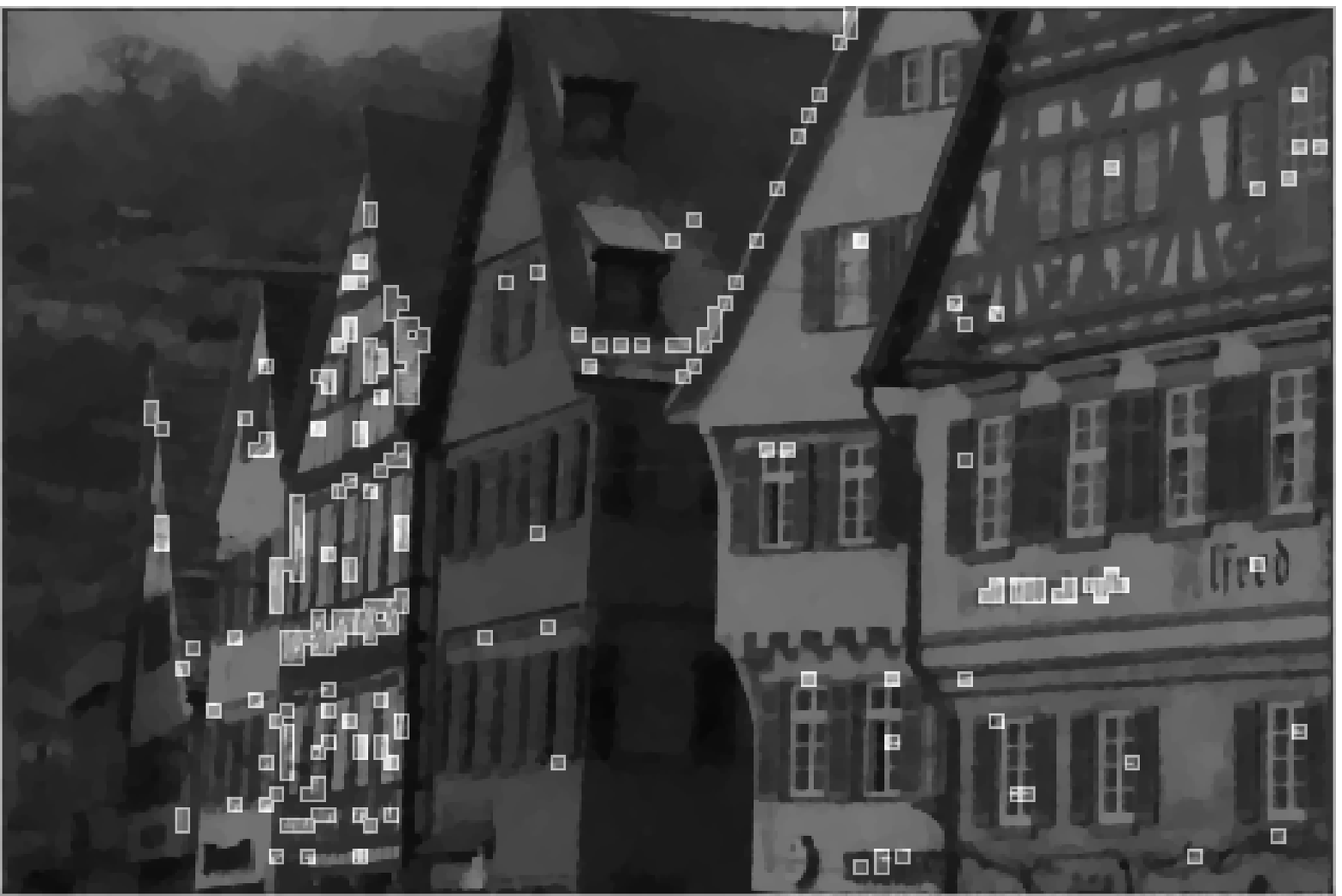}

\includegraphics[width = 0.49\textwidth]{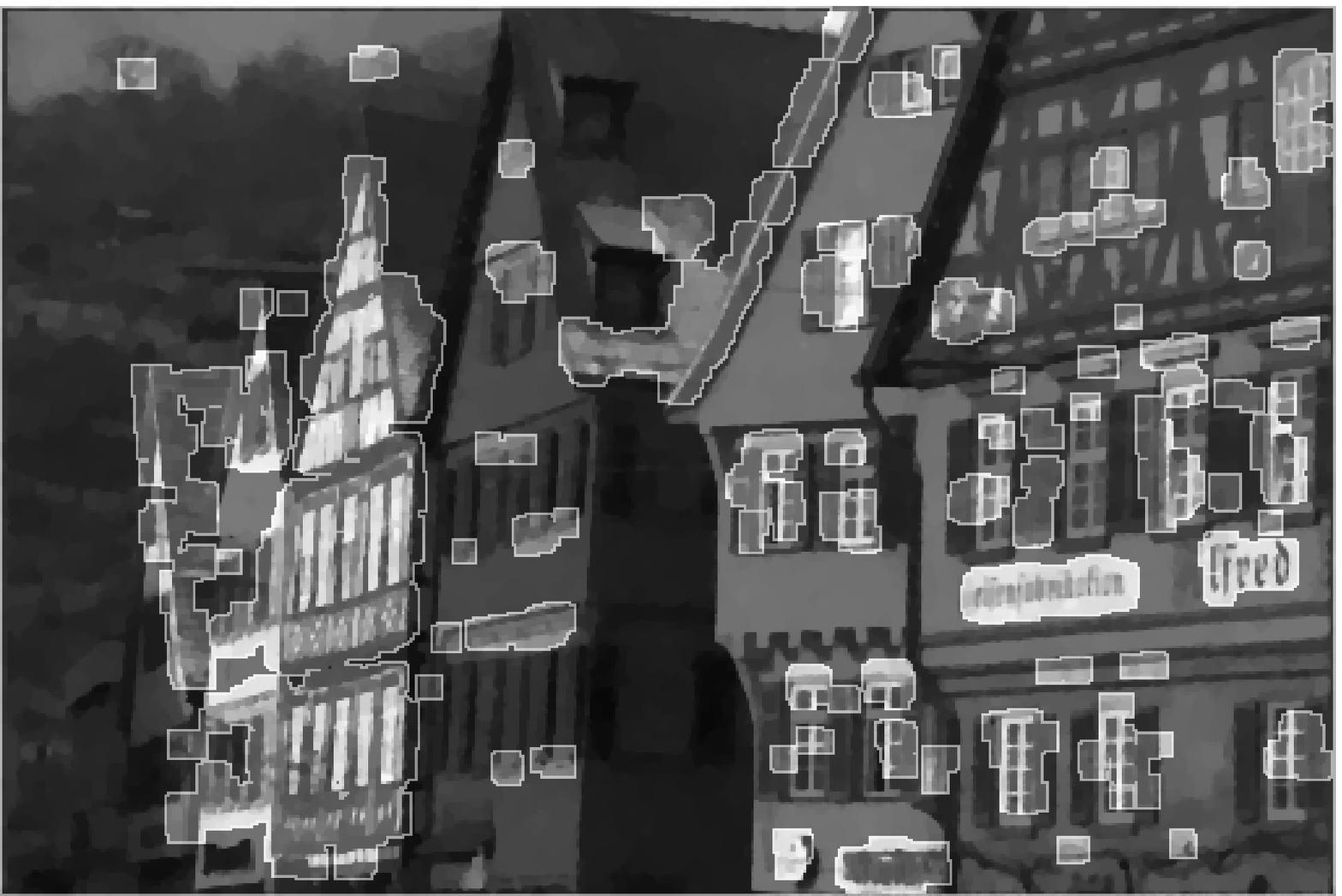}
\hspace{0.004\textwidth}
\includegraphics[width = 0.49\textwidth]{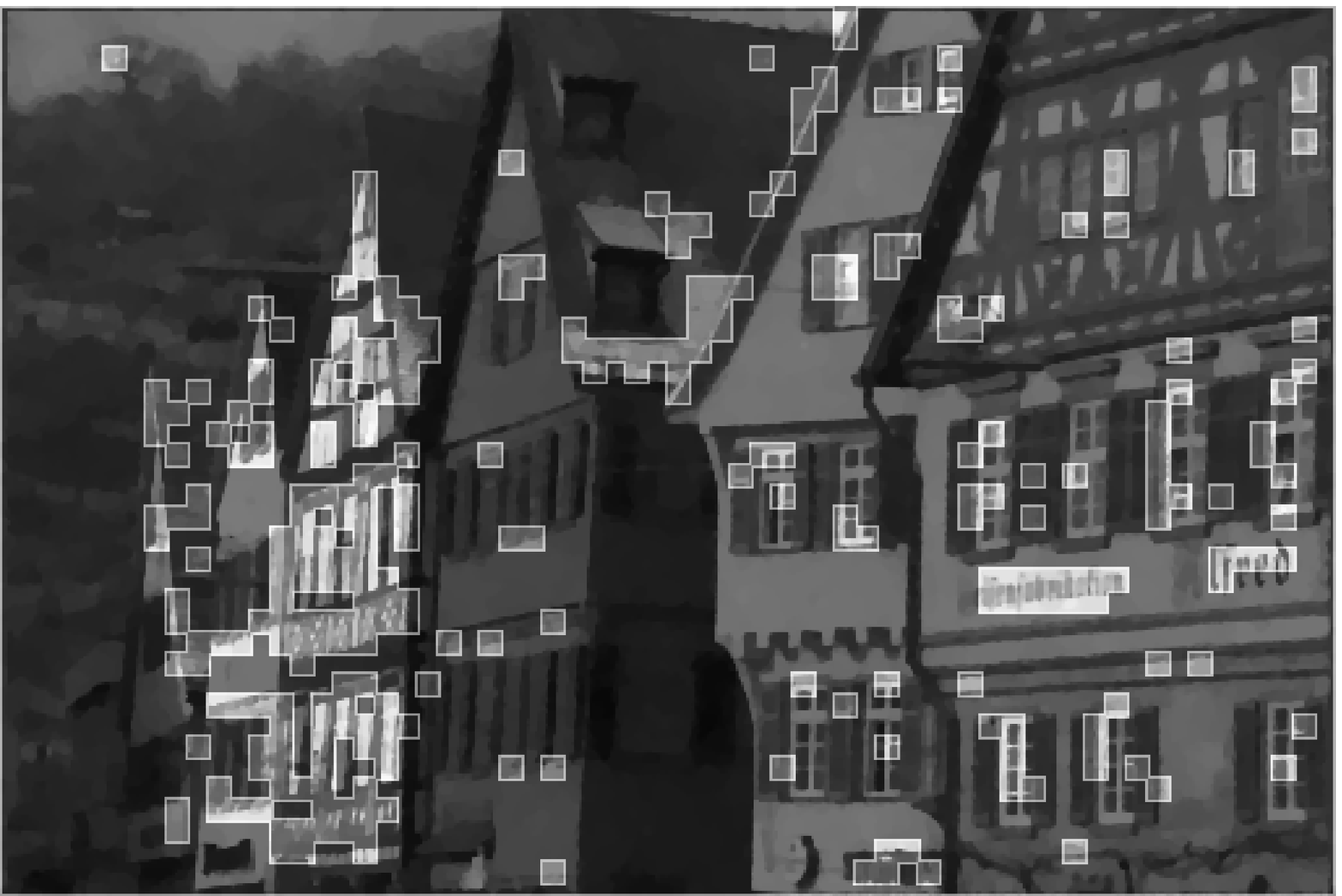}

\includegraphics[width = 0.49\textwidth]{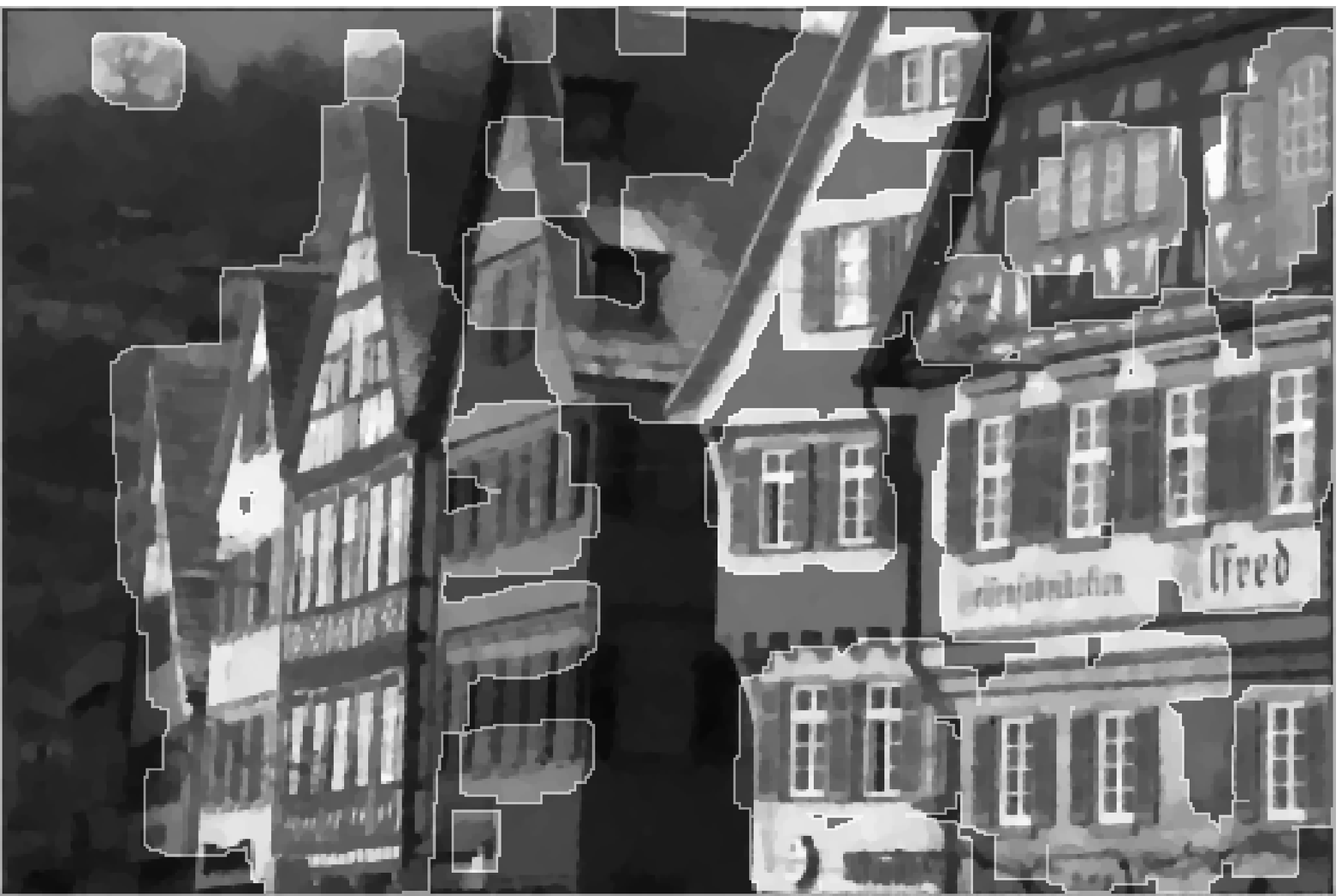}  
\hspace{0.004\textwidth}
\includegraphics[width = 0.49\textwidth]{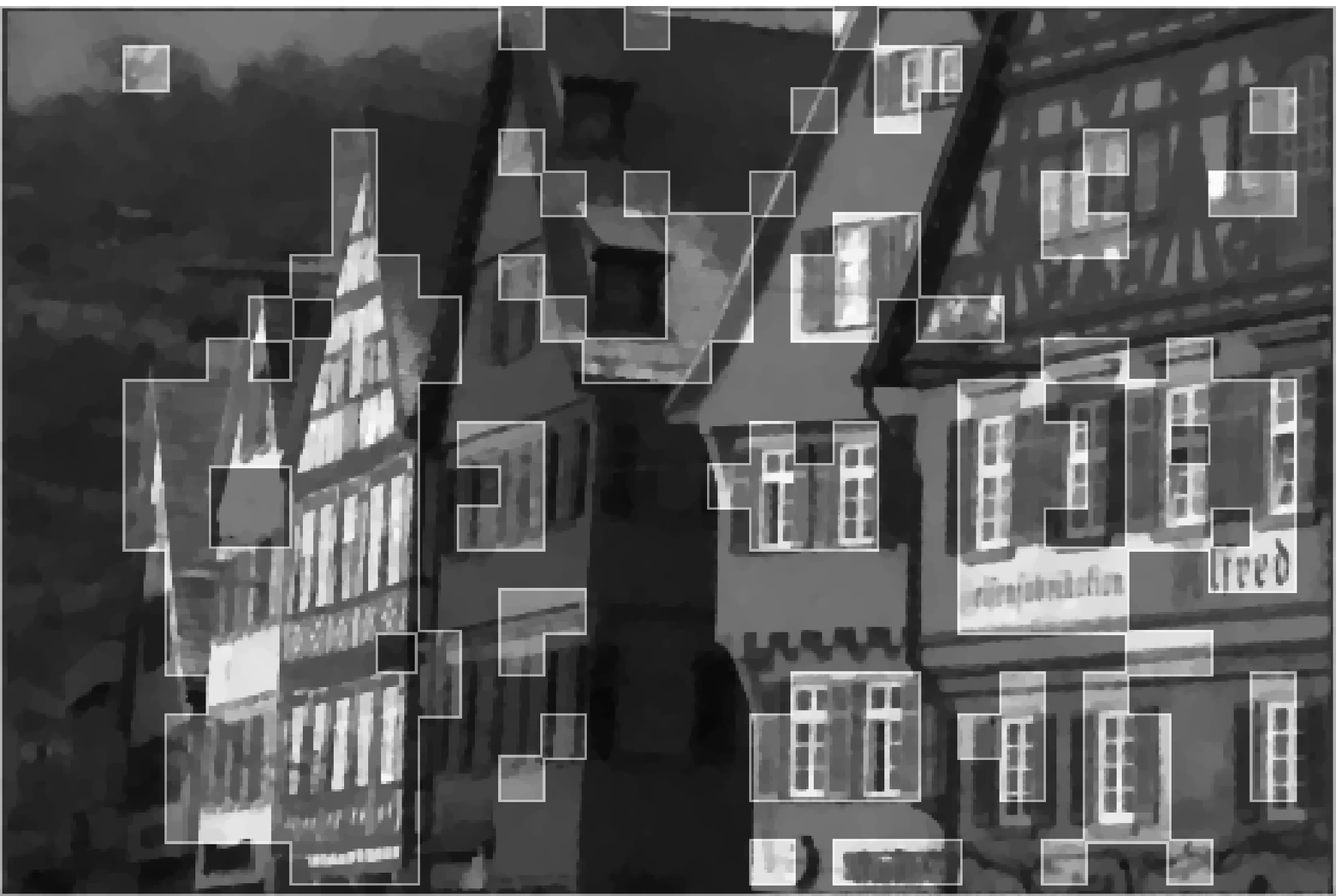}  
\end{center} 
\caption{Oversmoothed regions identified on the scales $\abs{S} = 4,8$ and $16$
(from left to right) for the system $\S_0$ (left column) and $\S_2$ (right
column).}\label{smre:viol}
\end{figure}

Obviously, $\S_0$ contains much more elements than $\S_2$ and is hence
likely to achieve a higher resolution. We indicate this by a numerical
simulation. Figure \ref{smre:fig_data} depicts the true signal $u^0$ (left, at a
resolution of $m\times n = 341\times 512$ pixels and gray values scaled in $[0,
1]$) and data $Y$ according to \eqref{intro:data} with $K=\id$ and $\sigma = 0.1$.

For illustrative purpose a global estimator $\hat u(\lambda)$ in
\eqref{intro:rof} for $u^0$ is computed that exhibits both over- and
undersmoothed regions (here, we set $\lambda = 0.075$). This estimator is
depicted in Figure \ref{smre:rofest}. The oversmoothed parts in $\hat
u(\lambda)$ can be identified via the MR-statistic $T$ in \eqref{intro:mre} by
marking those sets $S$ in $\S$ for which
\begin{equation*}
  \frac{c_S}{\sigma^2} \sum_{(i,j)\in S} \abs{Y_{ij} - (K\hat
  u(\lambda))_{ij}}^2 > 1.
\end{equation*}
The union of these sets for the systems $S_0$ (left column) and $S_2$ (right
column) are highlighted in Figure \ref{smre:viol} where we examine the
scales $\abs{S} = 4,8,16$. The parameters $c_S$ are chosen
as in Section \ref{smre:choice} with $\alpha=0.9$.

\section{Algorithmic Methodology}\label{impl}

In what follows, we present an algorithmic approach to the numerical computation
of SMRE in practice that extends the methodology in \cite{FriMarMun12}
where we proposed an alternating direction method of multipliers (ADMM).
Here, we use an \emph{inexact} version of the ADMM which decomposes the original
problem into a series of subproblems which are substantially easier to solve. In particular, an
inversion of the operator $K$ is no longer necessary. For this reason the
inexact ADMM has attracted much attention recently (see. e.g.
\cite{ChaPoc11,EssZhaCha10,ZhaBurBreOsh10}).

\subsection{Inexact ADMM}

In order to compute the desired saddle point of the augmented Lagrangian
function $L_\lambda$ in \eqref{smre:lagr}, we use the inexact ADMM that can be
considered as a modified version of the Uzawa algorithm (see e.g.\ \cite[Chap.
III]{ForGlo83}). Starting with some initial $p_0\in\R^{m\times n}$, the
original Uzawa algorithm consists in iteratively computing
\begin{enumerate}
  \item   $(u_k,v_k) \in \argmin_{u\in\L{2}, v\in\R^{m\times n}}
  L_\lambda(u,v;p_{k-1})$
  \item $p_k =  p_{k-1} +  \lambda(K u_k - v_k)$.
\end{enumerate}
Item $(1)$ amounts to an implicit minimization step w.r.t. to the variabels
$u$ and $v$ whereas $(2)$ constitutes an explicit maximization step for the
Lagrange multiplier $p$. The algorithm is usually stopped once the constraint in
\eqref{smre:decomp}  is fulfilled up to a certain tolerance. 

Applying this algorithm in a straightforward manner is known to be rather
impractical (mostly due to the difficult minimization problem in the first step)
and hence various modifications have been proposed in the optimization literature.
Firstly, we perform successive minimization w.r.t.\ $u$ and $v$ instead of
minimizing simultaneously, i.e. given $(u_{k-1}, v_{k-1}, p_{k-1})$ we compute
\begin{enumerate}
  \item   $u_k \in \argmin_{u\in \L{2}}
  L_\lambda(u,v_{k-1};p_{k-1})$
  \item $v_k \in \argmin_{v\in\R^{m\times n}}
  L_\lambda(u_k,v;p_{k-1})$
  \item $p_k =  p_{k-1} + \lambda(K u_k - v_k)$.
\end{enumerate}
This is the well-known \emph{alternating direction method of multipliers
(ADMM)} as proposed in \cite[Chap. III]{ForGlo83}). There, convergence of the
algorithm has been studied for the case when $J$ satisfies some regularity
assumptions. In \cite{FriMarMun12} we extended this result for general
functionals $J$ (as for example the total variation semi-norm \eqref{intro:tv}).
The resulting two minimization problems usually can be tackled much more
efficiently than the original problem.

Still, the first subproblem above requires the inversion of the (possibly
ill-posed) operator $K$. Thus, a second modification adds in the $k$-th loop of
the algorithm the following additional term to $L_\lambda(u,v_{k-1}; p_{k-1})$:
\begin{equation}\label{impl:modlagr}
  \frac{1}{2\lambda}\left(\zeta \norm{u - u_{k-1}}^2_{\Ls{2}} -
  \norm{K(u-u_{k-1})}^2\right).
\end{equation}
Here, $\zeta$ is chosen such that $\zeta>\norm{K}^2$. After some
rearrangements of the terms in $L_\lambda$ and \eqref{impl:modlagr} it can
easily be seen that $Ku$ cancels out and thus the undesirable inversion of $K$
is is replaced by a single evaluation of $K$ at the previous iterate $u_{k-1}$. 
However, by adding \eqref{impl:modlagr} the distance to the previous iterate $u_{k-1}$ is additionally penalized and
$L_\lambda(u,v_k;p_{k-1})$ is minimized only \emph{inexactly}. 

After the aforementioned rearrangements and by keeping in mind that  $H$ is the
indicator function of the convex set $\mathcal{C}$ in \eqref{smre:feasset}, the
inexact ADMM can be summarized as follows:

\begin{algorithm}[h!]\caption{Inexact ADMM}\label{impl:uzawa}
\begin{algorithmic}
\REQUIRE $Y\in \R^{m\times n}$ (data), $\lambda > 0$ (step size). 
\STATE $u_0\leftarrow \vec 0_{\Ls{2}}$ and $v_0 = p_0\leftarrow 0$.
\FOR{$k=1,2,\ldots$}
\STATE $k\leftarrow k+1$.
\STATE Minimize $L_\lambda(\cdot,v_{k-1};p_{k-1})+\frac{1}{2\lambda}\left(\zeta
\norm{\cdot - u_{k-1}}^2_{\Ls{2}} - \norm{K(\cdot-u_{k-1})}^2\right)$:
\begin{equation}\label{impl:primal}
  u_k \leftarrow \argmin_{u\in\L{2}} \frac{1}{2}\norm{ u - \left(u_{k-1}
  -\zeta^{-1}K^*(Ku_{k-1} - v_{k-1} + \lambda p_{k-1} \right) }_{\Ls{2}}^2 +
  \frac{\lambda}{\zeta} J(u).
\end{equation}
\STATE Minimize $L_\lambda(u_k,\cdot;p_{k-1})$:
\begin{equation}\label{impl:slack}
  v_k \leftarrow \proj_{\mathcal{C}}\left(Ku_k+\lambda p_{k-1}\right).
\end{equation}
\STATE Update dual variable:
\begin{equation}\label{impl:dual}
  p_k  \leftarrow  p_{k-1} +  \lambda^{-1}(K u_k - v_k).
\end{equation}
\ENDFOR
\end{algorithmic}
\end{algorithm}

In practice, Algorithm \ref{impl:uzawa} is very stable and straightforward to
implement, provided that efficient methods to solve \eqref{impl:primal} and
\eqref{impl:slack} are at hand (cf. Section \ref{impl:subprobs} below).
Moreover, it is equivalent to a general first-order primal-dual  algorithm as
studied in \cite{ChaPoc11} where the following convergence result was established. 

\begin{theorem}{\cite[Thm. 1]{ChaPoc11}}
Assume that $(\hat u, \hat v, \hat p)$ is a saddle point of $L_\lambda$.
Moreover, let $\set{u_k, v_k, p_k)}_{k\in\N}$ be generated by Algorithm
\ref{impl:uzawa} with $\zeta > \norm{K}^2$ and define the averaged sequences 
\begin{equation*}
\bar u_k = \frac{1}{k}\sum_{l=1}^k u_l\quad\text{ and }\quad \bar p_k =
\frac{1}{k}\sum_{l=1}^k p_l. 
\end{equation*}
Then, each weak cluster point of $\set{\bar u_k}_{k\in \N}$ is a solution of
\eqref{intro:optprob} and there exists a constant $C>0$ such that
\begin{equation}\label{impl:conveqn}
D_J^{-K^*\hat p}(\bar u_k, \hat u) + D_{H^*}^{\hat v}(\bar p_k, \hat p) \leq
C\slash k.
\end{equation}
\end{theorem}

The above result is rather general and in particular situations the assertions
may be quite weak. In particular if $J$ and $H^*$ have \emph{linear growth}, as
it is for instance the case for $J$ as in \eqref{intro:tv} and $H$ as in
\eqref{smre:feasset}, the Bregman distances appearing in
\eqref{impl:conveqn} may vanish although $(\bar u_k, \bar p_k) \not= (\hat u,
\hat p)$. If at least one of the functionals $J$ or $H^*$ is
uniformly convex, it is possible to come up with accelerated versions of
Algorithm \ref{impl:uzawa} that allow for stronger convergence results
(see \cite{ChaPoc11}). For the sake of simplicity we restrict our
consideration to the basic algorithm.

\subsection{Subproblems}\label{impl:subprobs}  

Closer inspection of Algorithm \ref{impl:uzawa} reveals that the original problem
- computing a saddle point of $L_\lambda$ - has been replaced by an iterative
series of subproblems \eqref{impl:primal} and \eqref{impl:slack}. We will now examine
these two subproblems and propose methods that are suited to solve them. Here we
proceed as in \cite{FriMarMun12}. 

We focus on \eqref{impl:slack} first. Note that the problem given there amounts
to computing the orthogonal projection of $v_k := K u_k + \lambda p_{k-1}$
onto the feasible region $\mathcal{C}$ as defined in \eqref{smre:feasset}. Due to the
supremum taken in the definition \eqref{intro:mre} of the statistic $T$, we can
decompose $\mathcal{C}$ into $\mathcal{C} = \bigcap_{S \in \S} \mathcal{C}_S$
where            

\begin{equation}\label{impl:indsets}
\mathcal{C}_S = \set{v\in \R^{m\times n}~:~ \frac{c_S}{\sigma^2}
\sum_{(i,j)\in S} \abs{v_{ij}-Y_{ij}}^2 \leq 1 },
\end{equation}
i.e.\ each $\mathcal{C}_S$ refers to the feasible region that would result if $\S$
contained $S$ only. Note that all $\mathcal{C}_S$ are closed and
convex sets (in fact, they are circular cylinders in $\R^{m\times n}$; see
the left panel in Figure \ref{fig:feas}). If we fix a $\mathcal{C}_S$ and
consider some $v\notin \mathcal{C}_S$, the projection of $v$ onto
$\mathcal{C}_S$ can be stated explicitly as

\begin{equation}
(P_{\mathcal{C}_S}(v))_{i,j} = \begin{cases}
                   v_{i,j} \quad &\text{if} \quad (i,j) \notin S\\
                   \left. Y_{ij} + (v_{ij} - Y_{ij})\sigma \middle/ \sqrt{c_S
                   \sum_{(k,l)\in S}(v_{kl} - Y_{kl})^2}\right.  \quad
                   &\text{if} \quad (i,j) \in S.\\
                   \end{cases}
\end{equation}

This insight leads us to the conclusion that any method which computes the
projection onto the intersection of closed and convex sets by merely using the
projections onto the individual sets only would be feasible to solve
\eqref{impl:slack}. Dykstra's Algorithm \cite{BoyDyk86}
works exactly in this way and is hence our method of choice to solve
\eqref{impl:slack}. For a detailed statement of the algorithm and how the total
number of sets that enter it may be decreased to speed up runtimes, see
\cite[Sec. 2.3]{FriMarMun12}. We note that despite these considerations,
the predominant part of the computation time of Algorithm \ref{impl:uzawa} is
spent for the projection step \eqref{impl:slack}. So far we did not take into
account parallelization of the projection algorithm. To some extent this is
possible in a straightforward manner, since the projections onto disjoint sets
in $\S$ can be carried out simultaneously (on GPUs for instance). But also
inherently parallel projection algorithms (including parallel versions of
Dykstra's method) received much attention recently and potentially would yield
a speed up of Algorithm \ref{impl:uzawa}. See for instance
\cite{ButCenRei01} for an overview.

We finally turn our attention to \eqref{impl:primal}. In contrast to the
standard version of the ADMM as proposed in
\cite{FriMarMun12}, the second subproblem in Algorithm \ref{impl:uzawa}
does not involve the inversion of the operator $K$. For this reason,
\eqref{impl:primal} here simply amounts to solving an unconstrained denoising
problem with a least-squares data-fit. Numerous methods for a wide range of
different choices of $J$ are available in order to cope with this problem. If
$J$ is chosen as the total variation seminorm, for example, the methods
introduced in \cite{Cha04,DobVog97,HinKun04} will be suited
(we will use the one in \cite{DobVog97}).

\section{Application in Fluorescence Microscopy}\label{poisson}

For image acquisition techniques that are based on single photon counts of a
light emitting sample, such as fluorescence microscopy, the Gaussian error
assumption \eqref{intro:data} is not realistic. Here, the non-additive model
\eqref{poisson:data} is to be preferred. Still, the estimation paradigm above
can be adapted to this scenario by means of \emph{variance stabilizing
transformations}. To this end we first recall \cite[Lem. 1]{BroCaiZhaZhaZhou10}
\begin{lemma}[Anscombe's Transform]\label{poisson:anscombe}
Let $Y\sim\textnormal{Pois}(\beta)$ with $\beta > 0$. Then, for all $c\geq 0$
\begin{align*}
\E{2\sqrt{Y+c}} & = 2\sqrt{\beta} + \frac{4c-1}{4\sqrt{\beta}} +
\bigo(\beta^{-3/2}) \\
\Var{2\sqrt{Y+c}} & = 1 + \frac{3 - 8c}{8\beta} + \bigo(\beta^{-2}).
\end{align*}
\end{lemma}   
Thus, the choice $c=3/8$ is likely to stabilize the variance of $2\sqrt{Y+c}$ 
at the constant value $1$ (in second order)  and its mean at $2\sqrt{\beta}$
(in first order) or in other words it approximately holds that
\begin{equation*}
2\sqrt{Y+3\slash 8} - 2\sqrt{\beta} \sim \mathcal{N}(0,1).
\end{equation*} 
It is obvious from Lemma \ref{poisson:anscombe} that the choice $c=1\slash 4$
will result in a better reduction of the bias, since the mean is then stabilized
at $2\sqrt{\beta}$ in second order (at the cost of a less stable variance).
Numerically, we found the difference to be negligible for our purposes.

 With the above considerations, it is straightforward to adapt the estimation
 scheme \eqref{intro:optprob} to the present case: Let $X_{ij} = 2\sqrt{Y_{ij} +
 3\slash 8}$. Then, we define a statistical multiresolution estimator $\hat u$ 
 for the model \eqref{poisson:data} to be any solution of
\begin{align}\label{poisson:optprob}
\inf_{u\in\L{2}} J(u) \quad\text{ s.t. }\quad & \sup_{S\in\S} c_S\sum_{(i,j)\in
S} \abs{ X_{ij} - 2\sqrt{(Ku)_{ij}}}^2\leq 1, \\
& (Ku)_{ij}\geq 0,\quad\text{ for all } (i,j)\in G.\nonumber
\end{align}
Note that for any $c>0$ the function $t\mapsto (\sqrt{t} - c)^2$ is convex on
$[0,\infty)$ and thus Problem \eqref{poisson:optprob} is again
a convex optimization problem. Similar as in Section \ref{smre:not} we can
rewrite \eqref{poisson:optprob} into
\begin{equation}\label{poisson:decomp}
\inf_{u\in\L{2}, v\in \R^{m\times n}} J(u) + \tilde H(v) \quad\text{
s.t. }\quad Ku - v = 0,
\end{equation}
where here $\tilde H$ denotes the indicator function on the feasibility
set $\tilde{\mathcal{C}}$ given by
\begin{equation}\label{poisson:feasset}
\tilde{\mathcal{C}} = \set{v\in\R^{m\times n}~:~ v_{ij}\geq 0 \text{ for all }
(i,j)\in G\text{ and } T(2\sqrt{v}-X)\leq 1}.
\end{equation}

\begin{figure}[h!] 
{\scriptsize
\psfrag{y1}{$Y_1$}
\psfrag{y2}{$Y_2$}
\psfrag{5}{$5$}   
\psfrag{10}{$10$}
\psfrag{v1}{$v_1$}  
\psfrag{v2}{$v_2$}
\begin{center}
\includegraphics[width =
0.45\textwidth]{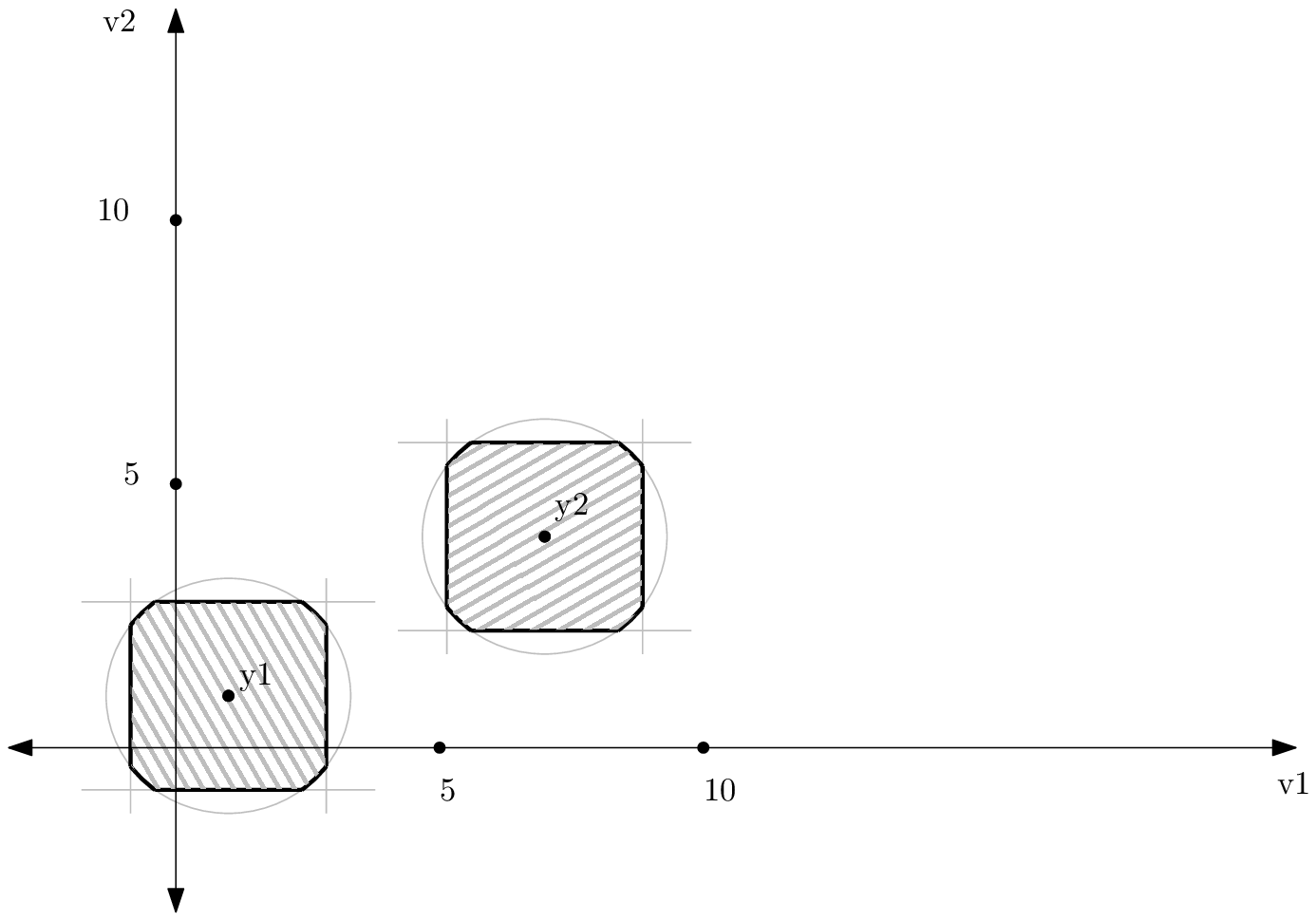}\hspace{0.05\textwidth}
\includegraphics[width = 0.45\textwidth]{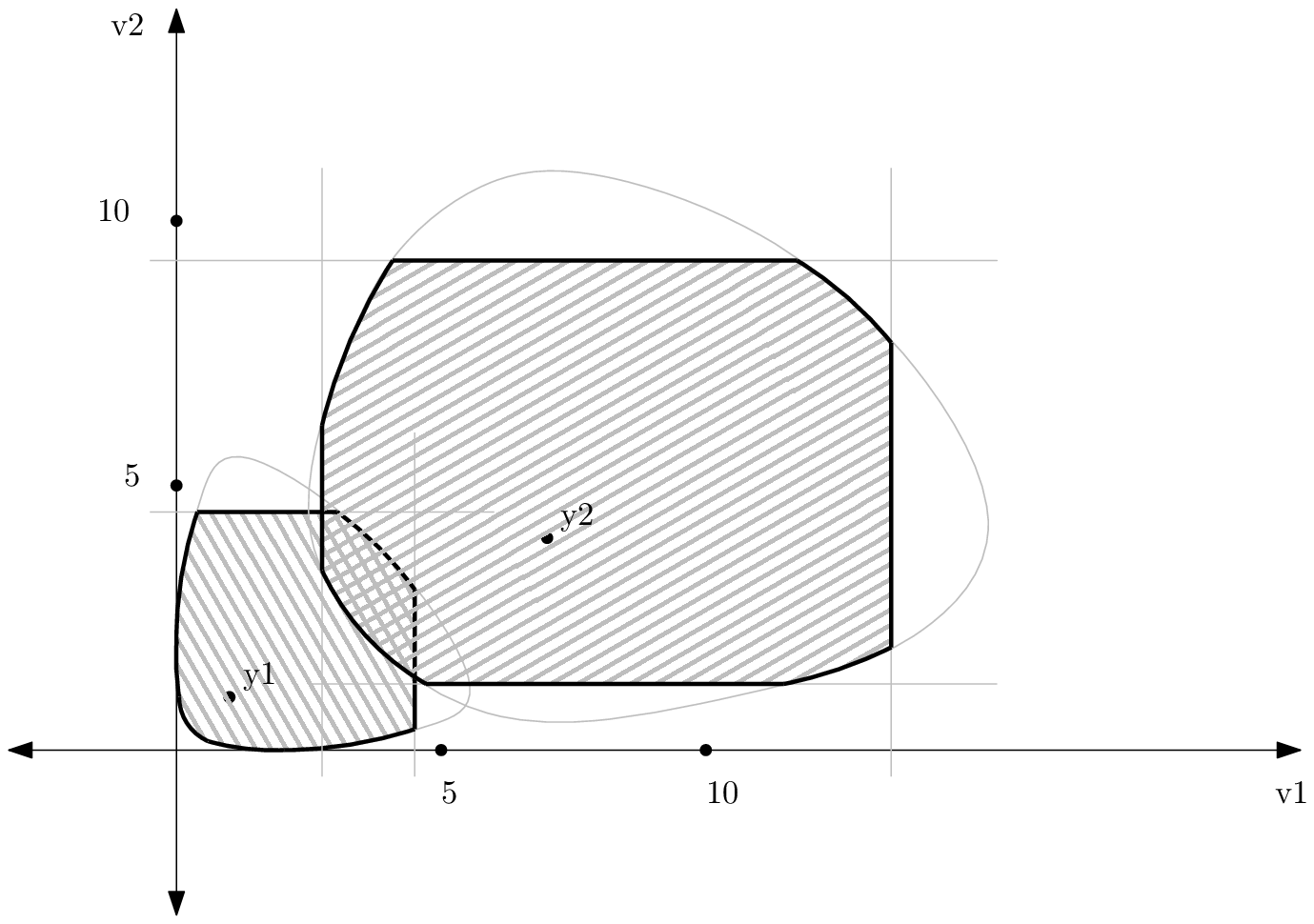}
\end{center}}
\caption{Admissible domains for $1\times 2$ ``images'' and data $Y_1 = (1,1)$
and $Y_2 = (7,4)$: Left: the sets $\mathcal{C}$ for the Gaussian case
with $\sigma^2 = 1$ (cf. \eqref{smre:feasset}). Right: the sets
$\tilde{\mathcal{C}}$ for the Poisson case (cf.
\eqref{poisson:feasset}). The scale weights $c_S$ are chosen
as proposed in Proposition \ref{smre:conf} with $\alpha = 0.9$.
The gray lines delimit the sets $\mathcal{C}_S$ and
$\tilde{\mathcal{C}}_S$ for $S\in \S =
\set{(1,0),(0,1),(1,1)}$.}\label{fig:feas}
\end{figure}

The right panel in Figure \ref{fig:feas} depicts the sets $\tilde{\mathcal{C}}$
for the simple case $m=2$ and $n=1$. It is important to note, that in contrast
to the Gaussian case (left panel), the feasibility sets $\tilde C$ are not
translation invariant, i.e. their shape and size depend on the data $Y$. In
particular, the size increases with $\norm{Y}$ which is due to the fact that the
variance of a Poisson random variable with law $\text{Pois}(\beta)$ increases
linearly with the parameter $\beta$.


In principle, Algorithm \ref{impl:uzawa} can be directly applied to solve
\eqref{poisson:optprob}: The set $\mathcal{C}$ in the projection step
\eqref{impl:dual} has to be replaced by the modified feasibility set
$\tilde{\mathcal{C}}$. Again, we observe that $\tilde{\mathcal{C}} =
\bigcap_{S\in\mathcal{S}} \tilde{\mathcal{C}}_S$, where
\begin{equation*}
\tilde{\mathcal{C}}_S = \set{v\in\R_{\geq 0}^{m\times n}~:~ c_S \sum_{(i,j)
\in S}\abs{2\sqrt{v_{ij}} - X_{ij}}^2 \leq 1}.
\end{equation*}
Using Dykstra's Algorithm amounts to compute the orthogonal projections onto the
sets $\tilde{\mathcal{C}}_S$ which, in contrast to the Gaussian case, is no
longer possible in closed form (except for the case when $\abs{S} = 1$) and
approximate solutions have to be used. Since during the runtime of Algorithm
\ref{impl:uzawa} these projections are to be computed a considerable amount of
time, this is clearly undesirable since inevitable numerical errors are likely
to accumulate.

As a way out, assume for the time being, that $\hat y \in\R^{m\times n}$  is
some estimator for $\sqrt{Ku^0}$ with $\hat y_{ij}>0$. Then, by Taylor
expansion, we find for $u\in\L{2}$ that
\begin{equation*}
2\sqrt{(Ku)_{ij}} = \hat y_{ij} + \frac{(Ku)_{ij}}{\hat
y_{ij}} + \bigo(\abs{\hat y_{ij}^2-(Ku)_{ij}}^2).
\end{equation*}
In place of \eqref{poisson:optprob}, we propose to solve the linearized problem
\begin{align}\label{poisson:optproblin}
\inf_{u\in\L{2}} J(u) \quad\text{ s.t. }\quad & \sup_{S\in\S} c_S\sum_{(i,j)\in
S} \abs{ X_{ij} - \hat y_{ij} - (Ku)_{ij}\slash \hat y_{ij}}^2\leq 1, \\
& (Ku)_{ij}\geq 0,\quad\text{ for all }(i,j)\in G.\nonumber
\end{align} 
Similar as above, we rewrite \eqref{poisson:optproblin}
into 
\begin{equation*}
\inf_{u\in\L{2}, w\in \R^{m\times n}} J(u) +  H_{\hat y}(w) \quad \text{ s.t.
}\quad \frac{Ku}{\hat y} - w = 0.
\end{equation*}
where $H_{\hat y}$ is the indicator function on the feasible set
$\mathcal{C}_{\hat y}$ given by
\begin{equation*}
\mathcal{C}_{\hat y} = \set{w\in \R^{m\times n}~:~ w_{ij}\geq 0 \text{
for all } (i,j)\in G\text{ and } T(w + \hat y - X)\leq 1}
\end{equation*}
Clearly, the orthogonal projection onto $\mathcal{C}_{\hat y}$ can again be
computed efficiently by Dykstra's algorithm as outlined in Section \ref{impl}.
The corresponding augmented Lagrangian functional reads as
\begin{equation}\label{poisson:auglarg}
L_{\lambda, \hat y}(u,w;p)  =  \frac{1}{2\lambda}\norm{Ku\slash \hat y - w}^2 +
J(u) + H_{\hat y}(w) + \inner{p}{Ku\slash \hat y - w}.
\end{equation}
For any ``good'' estimator $\hat y$ for $\sqrt{Ku^0}$ Algorithm \ref{impl:uzawa}
could be applied immediately to find a saddle point $(\hat u, \hat w, \hat p)$
of $L_{\lambda,\hat y}$, but usually such an estimator is not at hand.  We hence
propose to replace in the $k$-th loop of Algorithm \ref{impl:uzawa} the
estimator $\hat y$ by $\sqrt{Ku_k}$. To avoid instabilities we rather use
$\sqrt{\max(Ku_k,\delta)}$ for some small positive parameter $\delta >0$.
 We formalize these ideas in Algorithm \ref{poisson:uzawa}.

\begin{algorithm}[h!]\caption{ADMM for Poisson Noise}\label{poisson:uzawa}
\begin{algorithmic}
\REQUIRE $Y\in \R^{m\times n}$ (data), $\lambda > 0$ (step size), $\delta > 0$.
\STATE $w_0 = p_0\leftarrow 0$ and init $y_0\in\R_{>0}^{m\times n}$.
\FOR{$k=1,2,\ldots$}
\STATE $k\leftarrow k+1$.
\STATE Minimize $L_{\lambda, w_{k-1}}(\cdot,w_{k-1};p_{k-1})$:
\begin{equation}\label{poisson:primal}
  u_k \leftarrow \argmin_{u\in\L{2}} \frac{1}{2}\norm{Ku\slash y_{k-1} -
  (w_{k-1} - \lambda p_{k-1})}^2 + \lambda J(u).
\end{equation}
\STATE $y_k \leftarrow \sqrt{\max(Ku_k, \delta)}$.
\STATE Minimize $L_{\lambda, y_k}(u_k,\cdot;p_{k-1})$:
\begin{equation}\label{poisson:slack}
  w_k \leftarrow \proj_{\mathcal{C}_{y_k}}\left( Ku_k\slash y_k +\lambda
  p_{k-1} \right).
\end{equation}
\STATE Update dual variable:
\begin{equation}\label{poisson:dual}
  p_k  \leftarrow  p_{k-1} +  \lambda^{-1}(K u_k\slash y_k - w_k).
\end{equation}
\ENDFOR
\end{algorithmic}
\end{algorithm}  

In practice the Algorithm has proved to be very stable, however it seems to be
rather involved to prove numerical convergence of Algorithm \ref{poisson:uzawa}
which is beyond the scope of this paper. If $(\bar u, \bar w, \bar p)$ is a
limit of the sequence $(u_k, w_k, p_k)$ in Algorithm \ref{poisson:uzawa}, it is
quite obvious that $(\bar u, \bar w^2)$ is a solution of \eqref{poisson:decomp}
and hence that $\bar u$ solves \eqref{poisson:optprob}. Moreover, it is not
straightforward to incorporate a preconditioner similar to \eqref{impl:modlagr}
that renders the step \eqref{poisson:primal} semi-implicit.
  
\section{Numerical Results}\label{results}

We conclude this paper by demonstrating the performance of SMRE as computed by
our methodology introduced in the previous Sections. We will treat the
denoising problem  in Paragraph \ref{results:den}  as well as 
deconvolution and inpainting problems in Paragraph
\ref{results:inp}. Finally, we will study SMRE for the Poisson model
\eqref{poisson:data} computed by means of Algorithm \ref{poisson:uzawa} in
Paragraph \ref{results:fluor}. Here we will use real data from nanoscale
fluorescence microscopy provided by the Department of NanoBiophotonics at the
Max Planck Institute for Biophysical Chemistry in
G{\"o}ttingen\footnote{\url{http://www.mpibpc.mpg.de/groups/hell/}}.

When it comes down to computation, we think of an image $u$ as an $m\times n$
array of pixels rather than an element in $\L{2}$. Accordingly, the operator $K$
is realized as an $mn\times mn$ matrix and $\nabla$ denotes the discrete
(forward) gradient.  In all our experiments we use a step size $\lambda = 0.001$
for the ADMM method (Algorithm \ref{impl:uzawa}) and we stop the iteration if
the following criteria are satisfied
\begin{equation*}
{\norm{Ku_k - Ku_{k-1}}\over \norm{Y}}\leq 10^{-3},\; {\norm{Ku_k - v_k}\over
\norm{Y}}\leq 10^{-3}\;\text{ and }\; \max_{S\in\S}
\frac{\sqrt[4]{t_S(Ku_k - Y)} - \mu_S}{\sigma_S} \leq 1.01 q_\alpha.
\end{equation*}
Here, $\S$ is the system of subsets in use and $t_s,\mu_S$ and $\sigma_S$ are
defined as in Section \ref{smre:choice}. 

For the Poisson modification in Algorithm \ref{poisson:uzawa} we use the same
criteria, instead that $Ku_k - v_k$ is replaced by $\sqrt{Ku_k}-w_k$ and
$Ku_k-Y$ by $2\sqrt{Ku_k}-X$, where $X$ is as in Section \ref{poisson}.

\subsection{Denoising}\label{results:den}

In this paragraph we consider data $Y$ given by \eqref{intro:data} when $K$ is
the identity matrix and $u^0$ is the test image in Figure \ref{smre:fig_data}
($m=341$ and $n=512$), i.e.
\begin{equation*}
Y_{ij} = u^0_{ij} + \eps_{ij},\quad (i,j)\in G.
\end{equation*}
The study the scenarios when $\sigma = 0.1$ ($10\%$ Gaussian noise) and $\sigma
= 0.2$ ($20\%$ Gaussian noise). We compute SMRE based on the subsystems of
$S_0$ and $S_2$ as introduced in Paragraph \ref{smre:test} where we fixed $\alpha = 0.9$. To this end we utilize Algorithm
\ref{impl:uzawa} with $\zeta = 1$, i.e. the standard ADMM.

\begin{figure}[t!]
\begin{center}
\includegraphics[width =
0.4\textwidth]{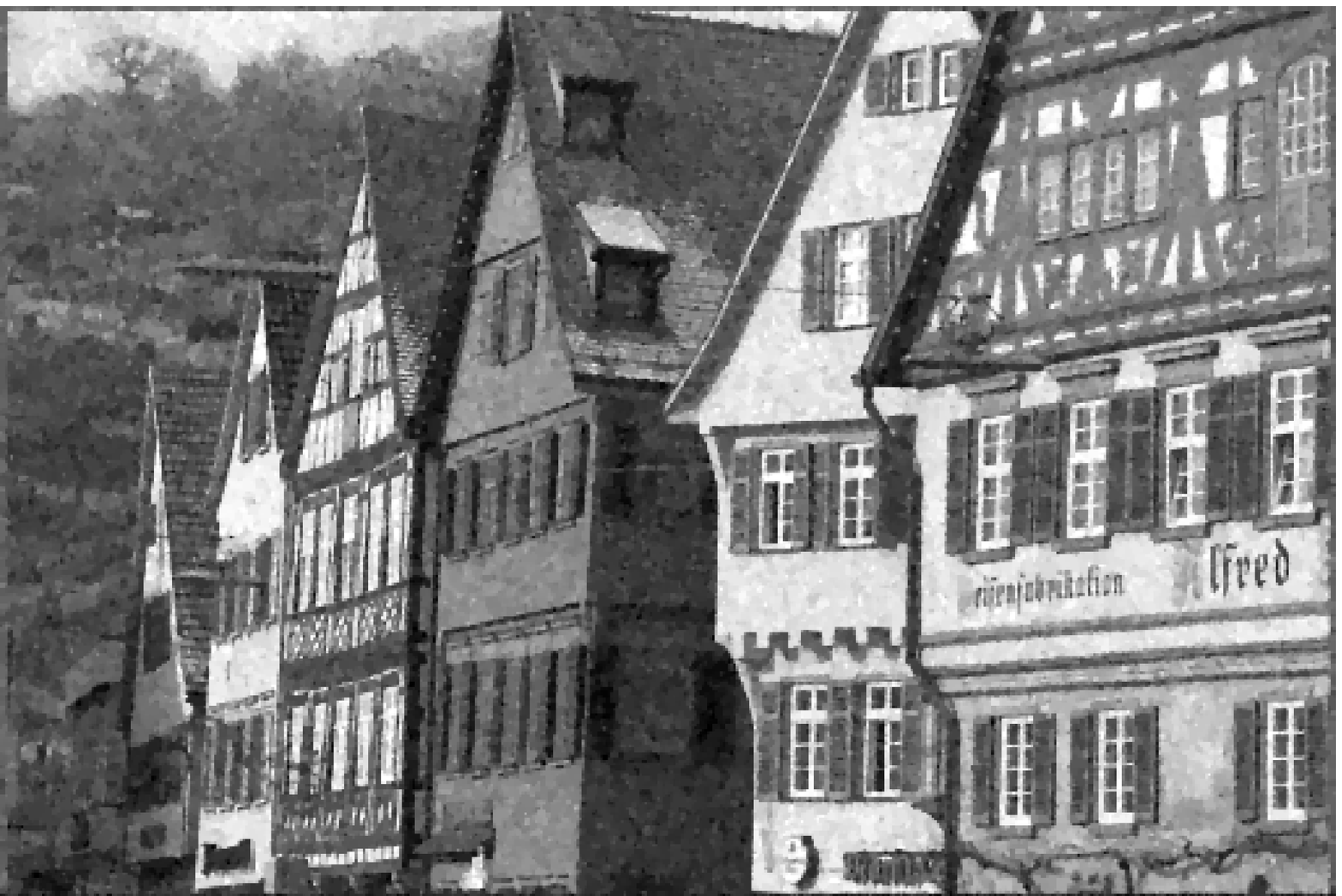}
\includegraphics[width =
0.4\textwidth]{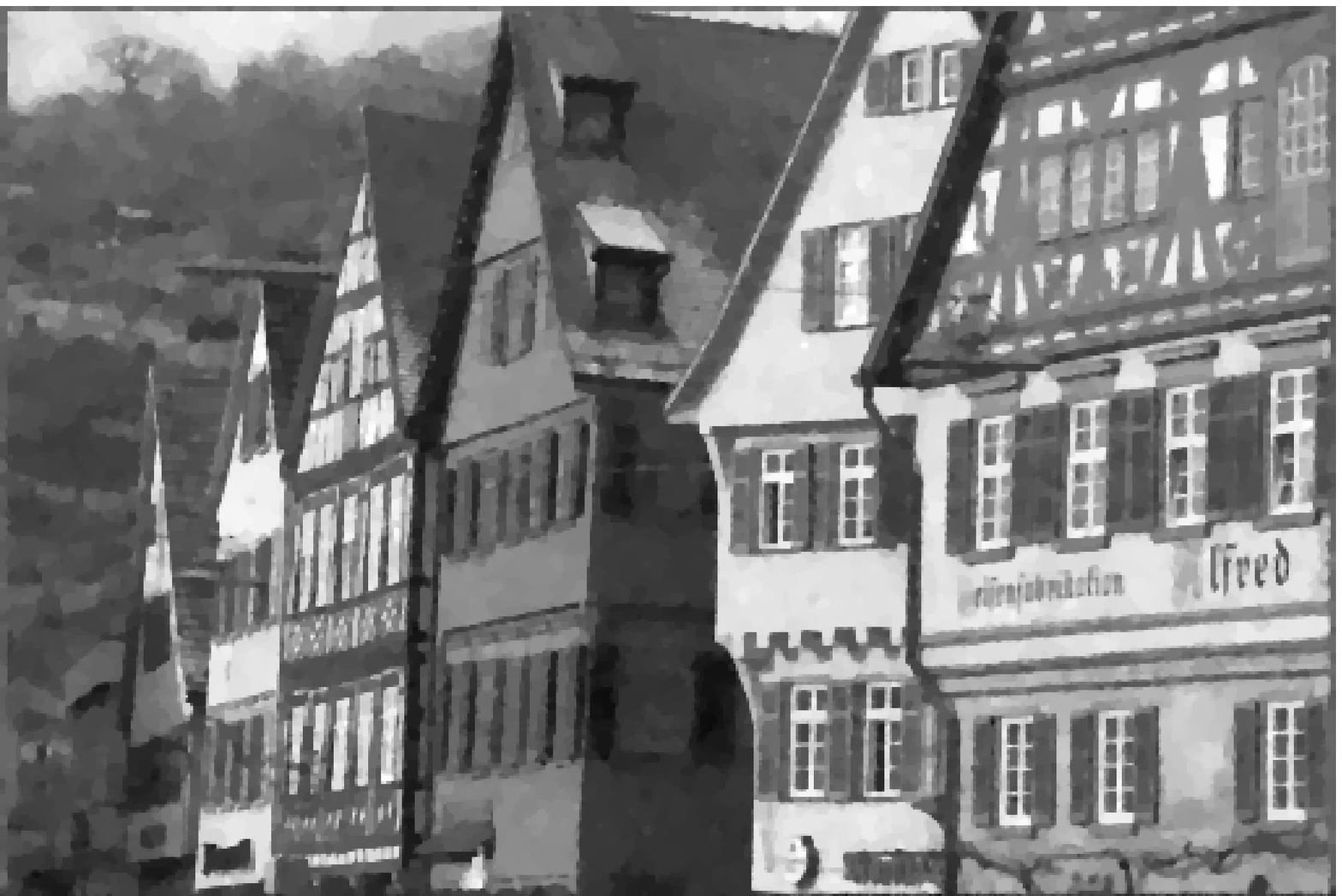}

\includegraphics[width =
0.4\textwidth]{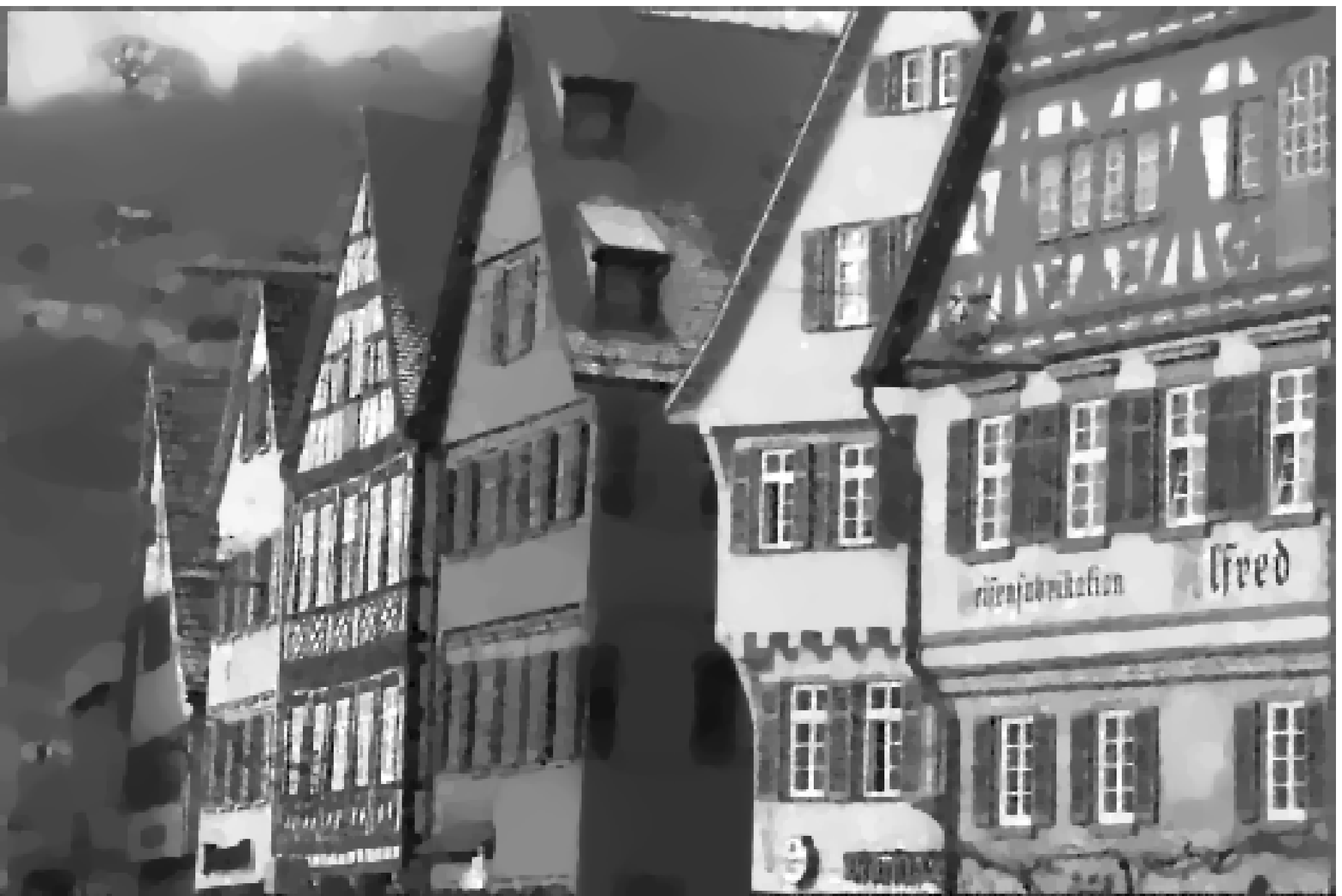}
\includegraphics[width =
0.4\textwidth]{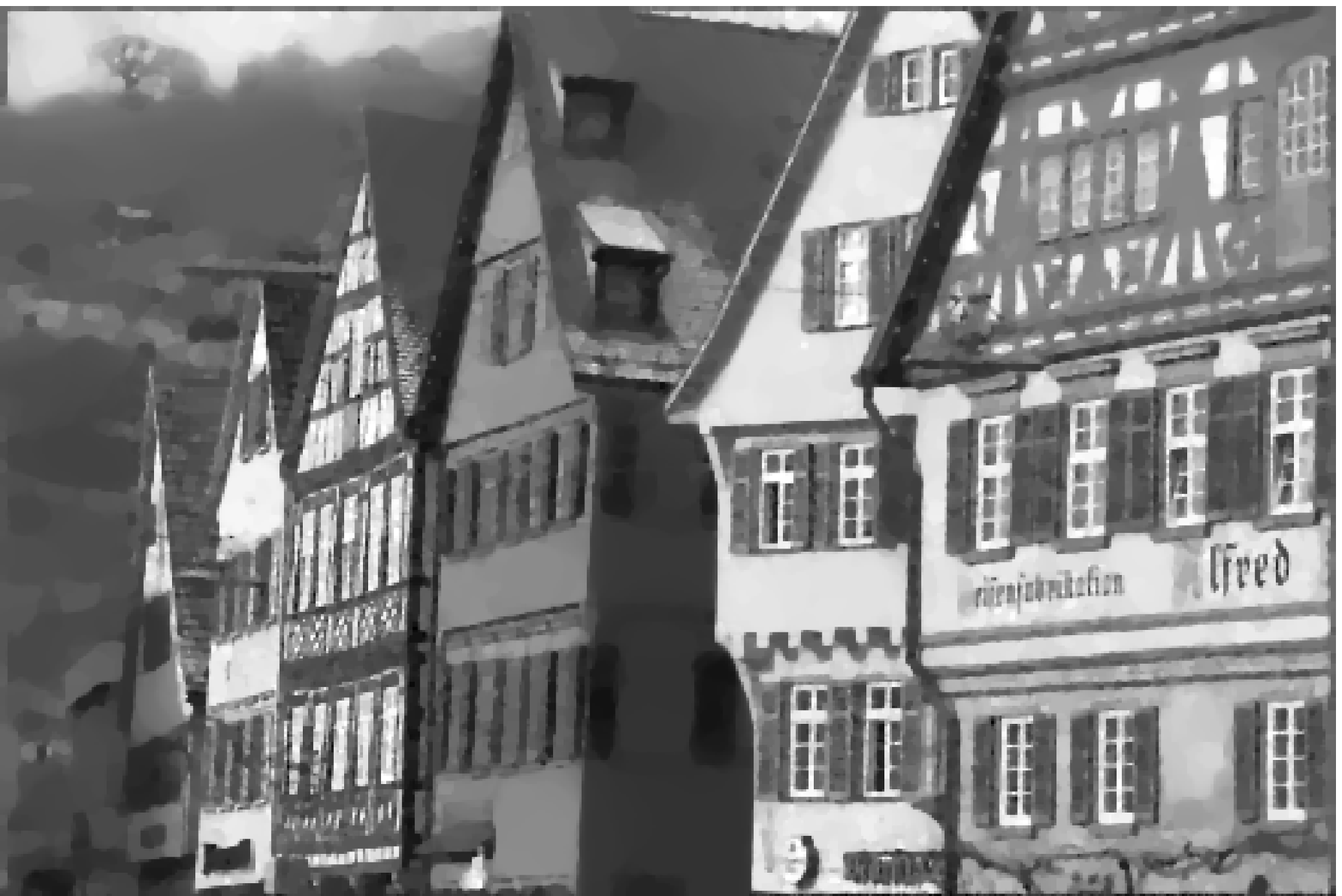}
 
\includegraphics[width = 
0.4\textwidth]{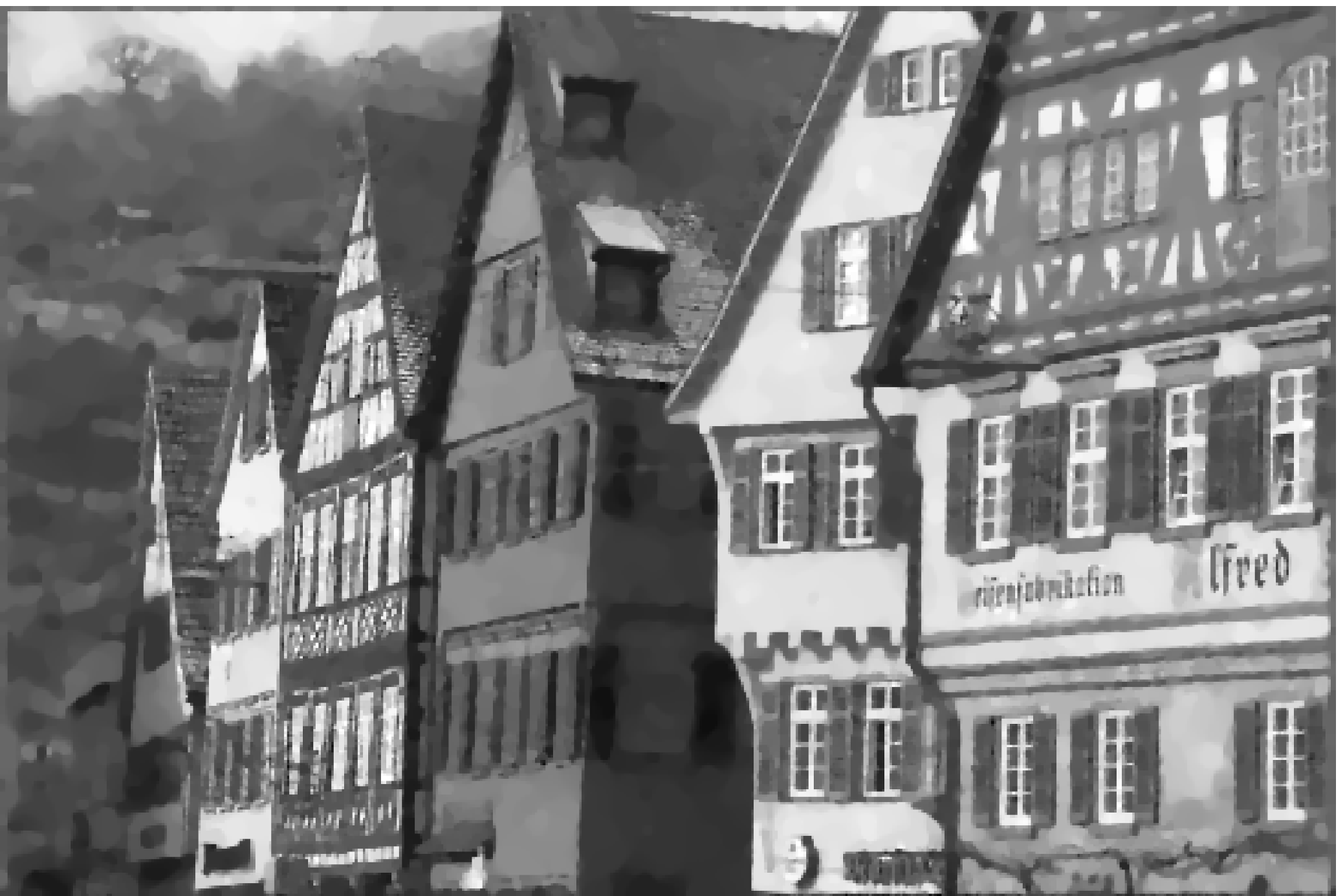}
\includegraphics[width =
0.4\textwidth]{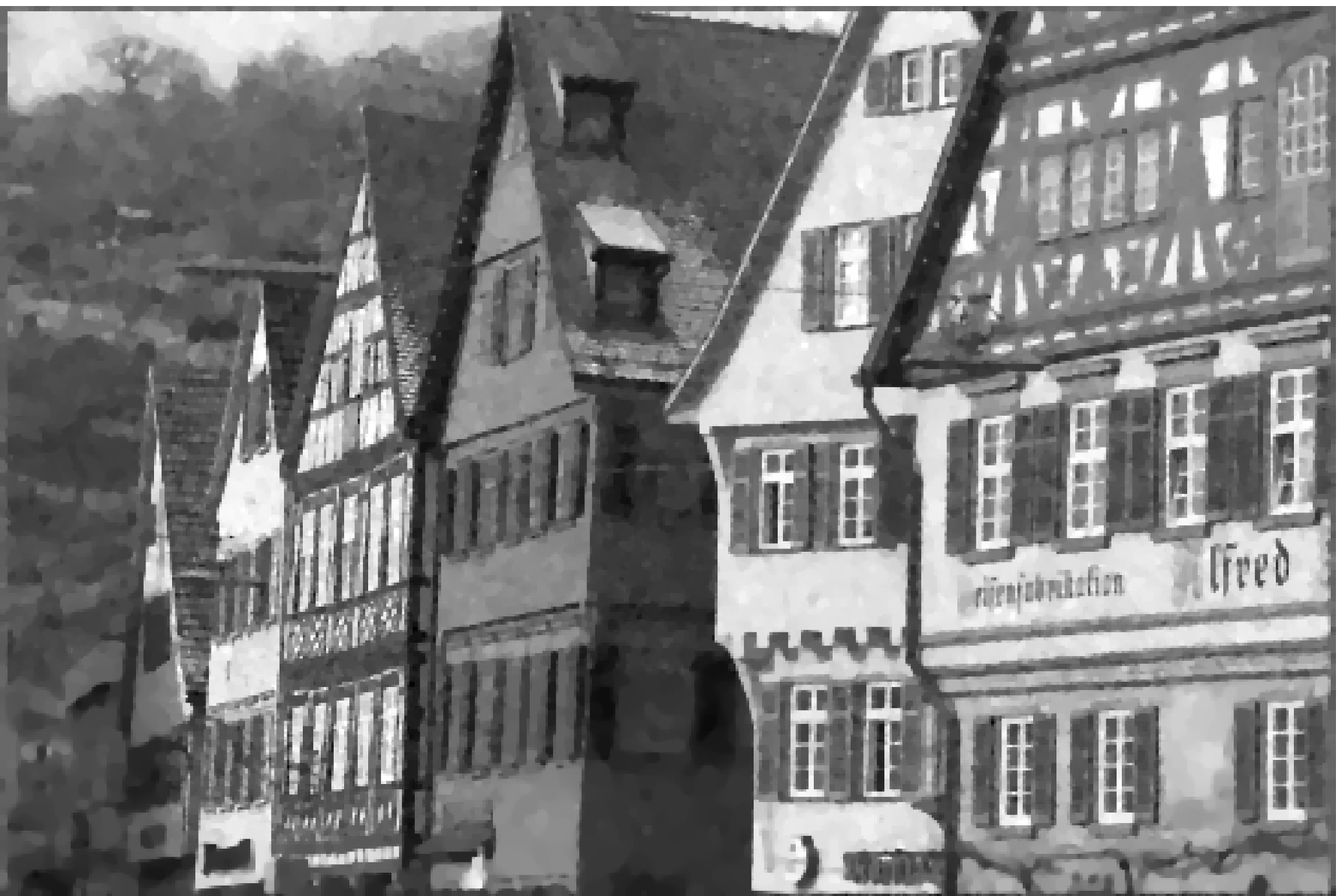}
\end{center} 
\caption{Denoising results for $10\%$ Gaussian noise. First row: $\Ls{2}$ - and
Bregman oracle. Middle row: SA-TV with window size $11$ and $19$. Last row:
SMREs w.r.t. $\S_0$ and $S_2$ (with $\alpha = 0.9$). } \label{results:figden10}
\end{figure}

\begin{figure}[t!]
\begin{center}
\includegraphics[width =
0.4\textwidth]{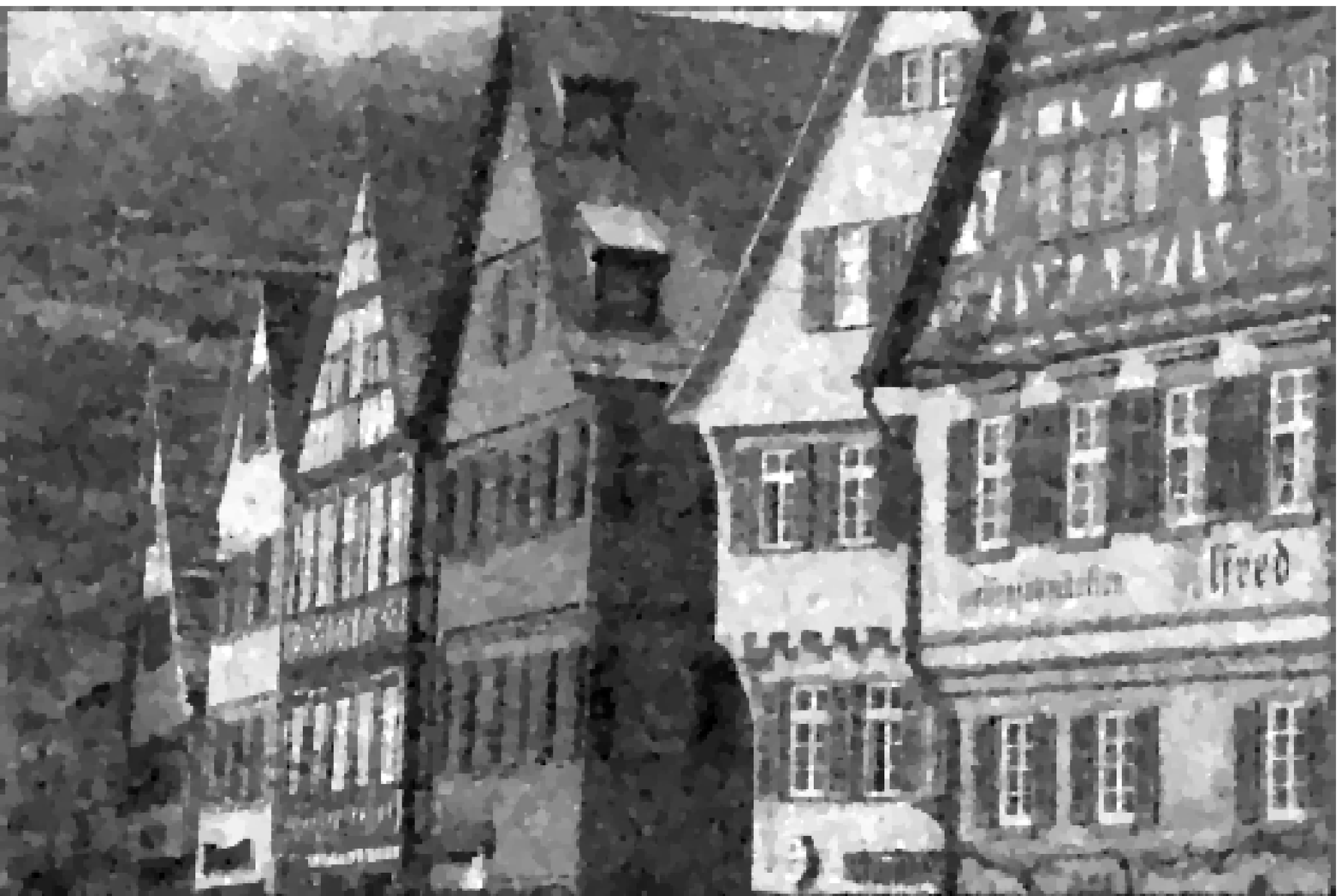}
\includegraphics[width =
0.4\textwidth]{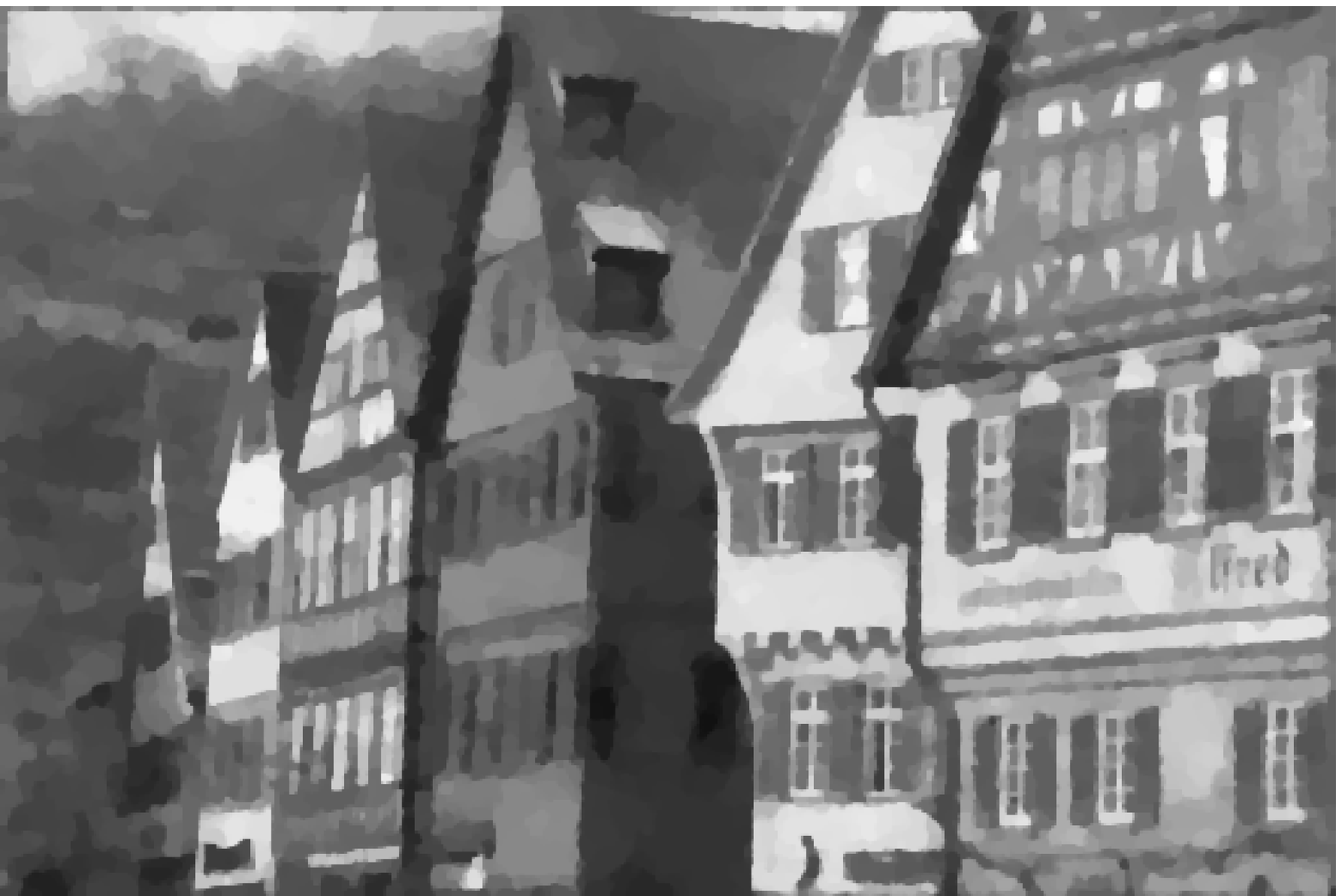}

\includegraphics[width =
0.4\textwidth]{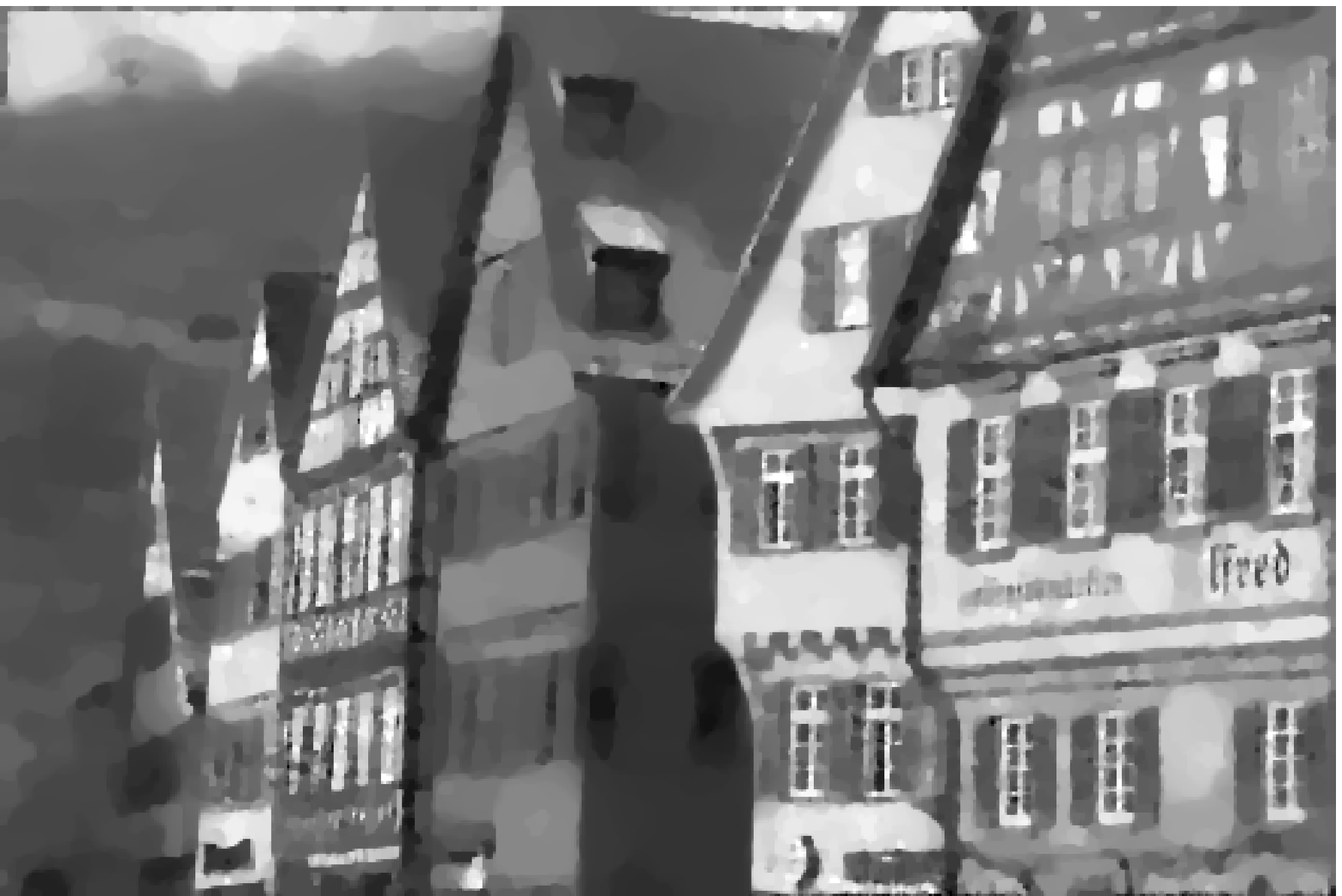}
\includegraphics[width =
0.4\textwidth]{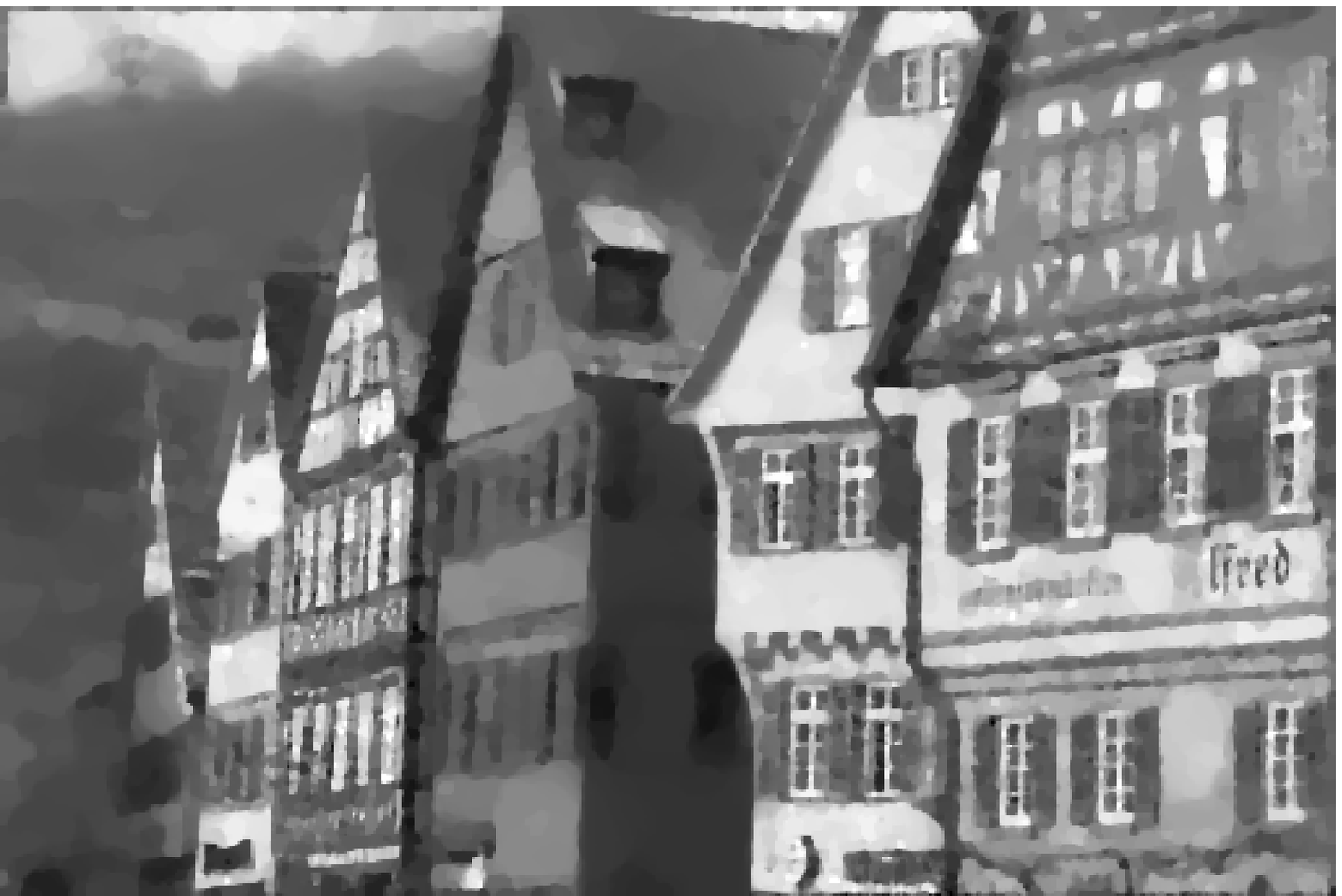}
 
\includegraphics[width = 
0.4\textwidth]{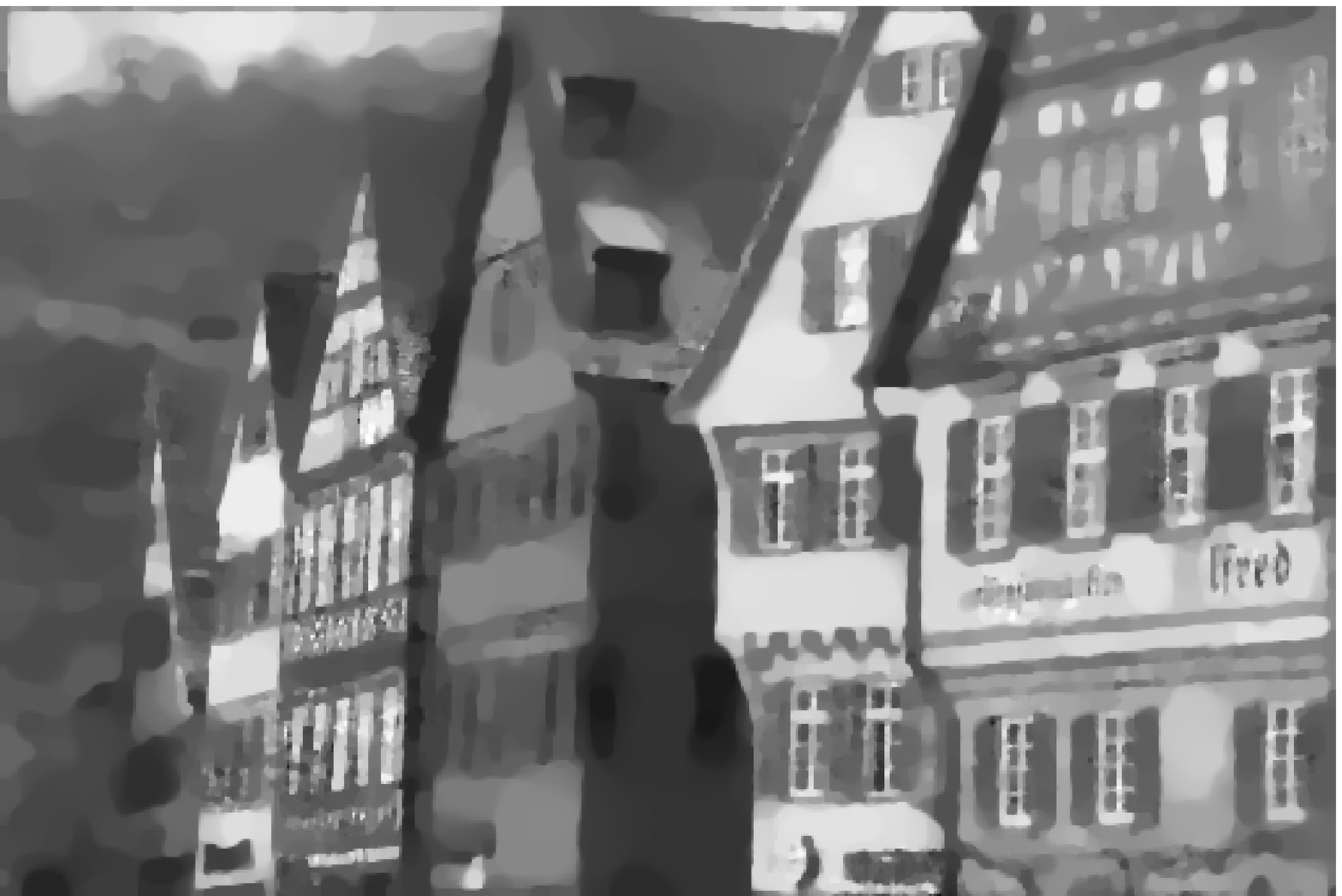}
\includegraphics[width =
0.4\textwidth]{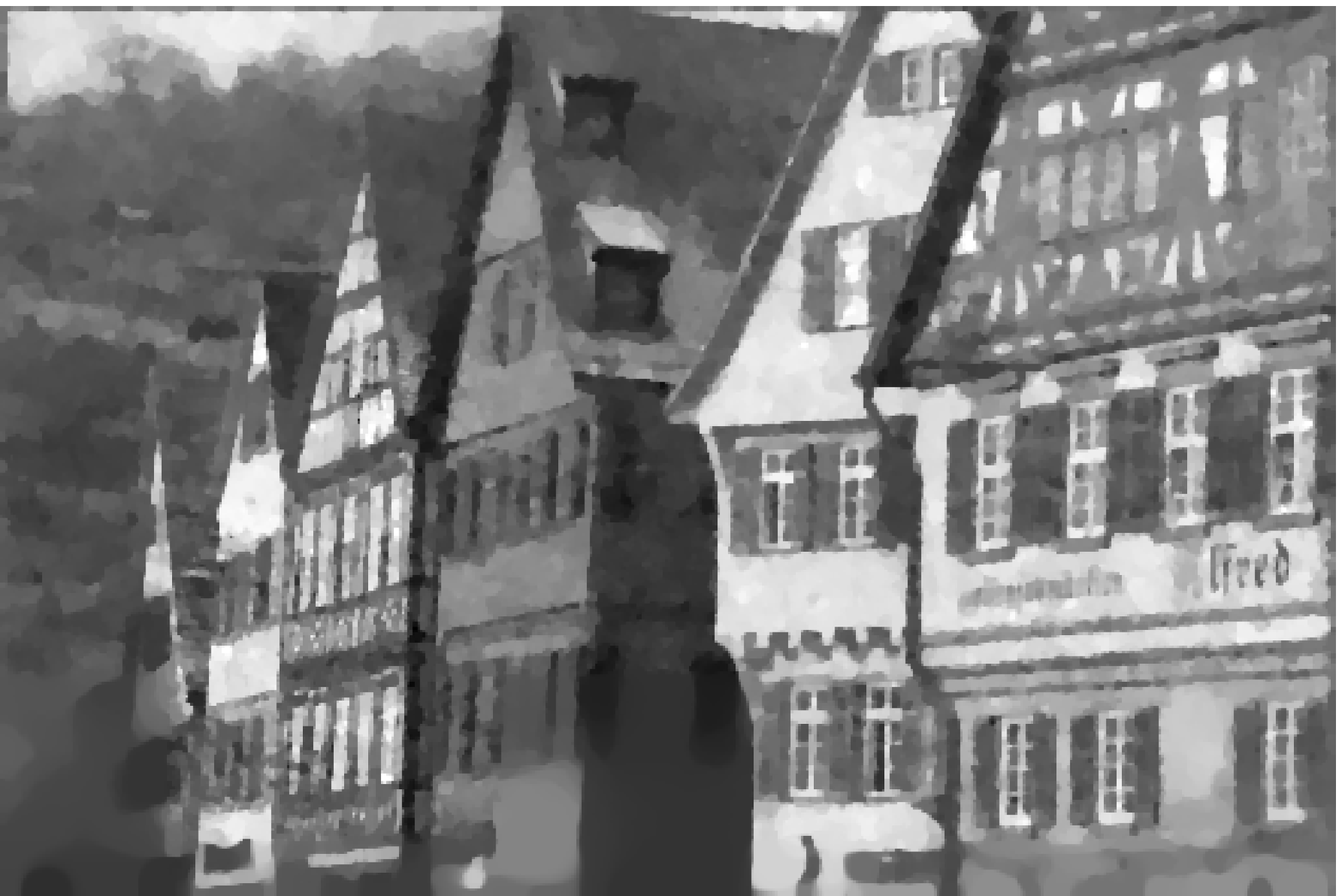}
\end{center} 
\caption{Denoising results for $20\%$ Gaussian noise. First row: $\Ls{2}$ - and
Bregman oracle. Middle row: SA-TV with window size $11$ and $19$. Last row:
SMREs w.r.t. $\S_0$ and $S_2$ (with $\alpha = 0.9$). } \label{results:figden20}
\end{figure} 

\begin{figure}
\fbox{\includegraphics[height=0.87\textheight]{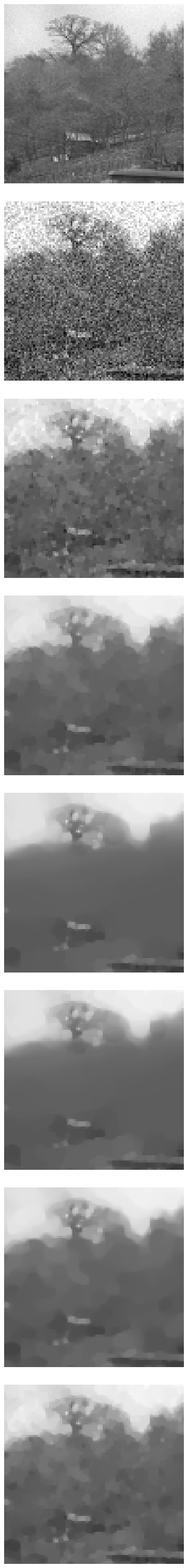}
\includegraphics[height=0.87\textheight]{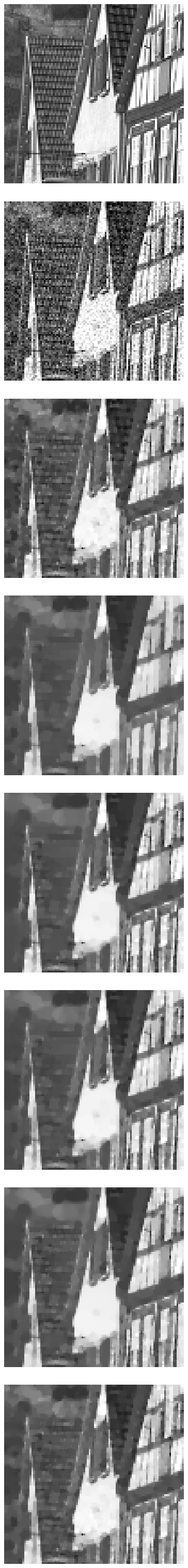}
\includegraphics[height=0.87\textheight]{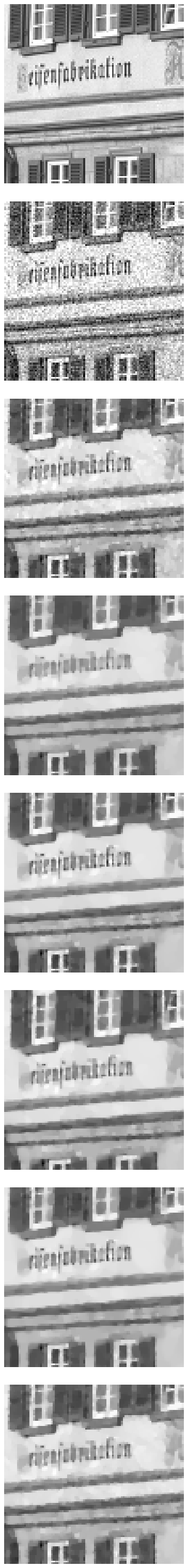}}
\fbox{\includegraphics[height=0.87\textheight]{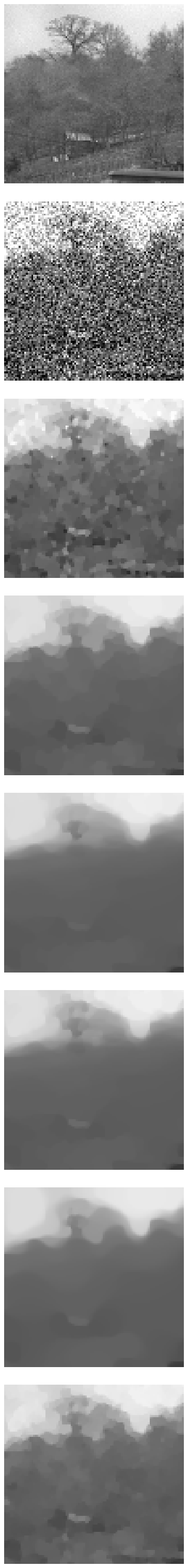}
\includegraphics[height=0.87\textheight]{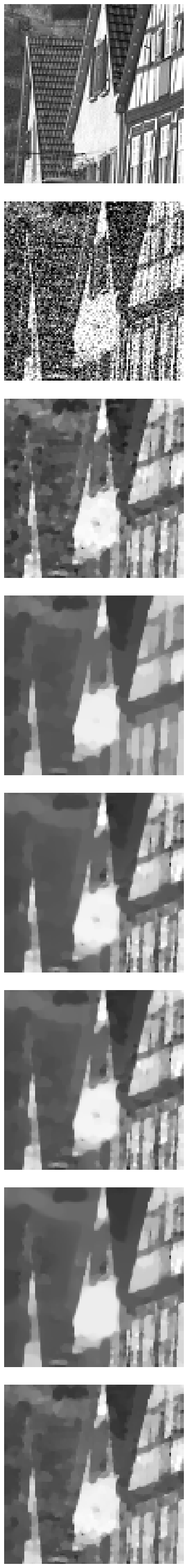}
\includegraphics[height=0.87\textheight]{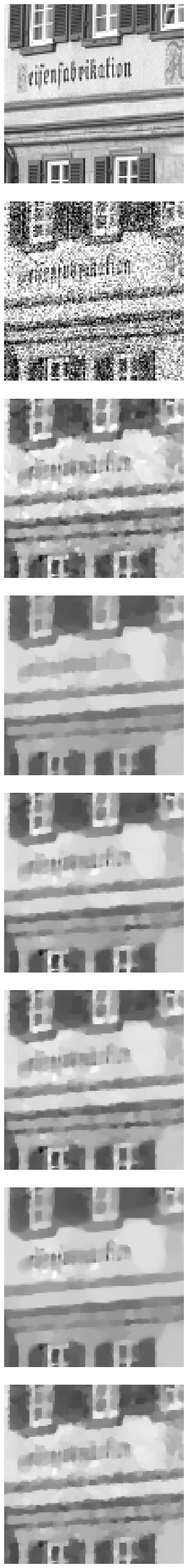}}
\caption{Reconstruction details for $10\%$ (left frame) and $20\%$ (right frame)
Gaussian noise. From top to bottom: True signal, data, $\Ls{2}$- and Bregman oracle,
SA-TV with window size $11$ and $19$, SMRE w.r.t. $\S_0$ and $\S_2$
and $\alpha=0.9$}\label{results:details}
\end{figure}

We compare our estimators to the global estimators $\hat u(\lambda)$ ($\lambda >
0$) as defined in \eqref{intro:rof}. We choose $\lambda = \lambda_2$ and
$\lambda = \lambda_{\text{B}}$ such that the mean squared distance and the mean
symmetric Bregman distance to the true signal $u^0$ is minimized, respectively.
To be more precise, we set
\begin{equation}\label{results:oracledef}
  \lambda_2 = \E{\argmin_{\lambda > 0} \norm{u^0 - \hat
  u(\lambda)}^2}\; \text{ and }\; \lambda_{\text{B}} =
  \E{\argmin_{\lambda > 0} D_J^{\text{sym}}(u^0, \hat u(\lambda))},
\end{equation}
where the symmetric Bregman distance for $J$ as in \eqref{intro:tv}
formally reads as
\begin{align*}
  D^{\text{sym}}_J(u,v) & = \sum_{(i,j)\in G} \left(\frac{\nabla
  u_{ij}}{\abs{\nabla u_{ij}}} - \frac{\nabla v_{ij}}{\abs{\nabla
  v_{ij}}}\right)\cdot \left(\nabla u_{ij} - \nabla v_{ij}\right) \\
  & =  \sum_{(i,j)\in G} (\abs{\nabla u_{ij}} + \abs{\nabla
  v_{ij}})\left( 1 - \frac{\nabla u_{ij} \cdot \nabla v_{ij}}{\abs{\nabla
  u_{ij}} \abs{\nabla v_{ij}}}\right).
\end{align*}
This means that $D^{\text{sym}}_J(u,v)$ is small if for sufficiently many pixels
$(i,j)\in G$ either both $u$ and $v$ are constant in a neighborhood of $(i,j)$
or the level lines of $u$ and $v$ at $(i,j)$ are locally parallel. In practice,
we rather use
\begin{equation*}
J(u) = \sum_{(i,j)\in G} \sqrt{|\nabla u_{ij}|_2^2 + \beta^2}
\end{equation*}
instead of $J$ in \eqref{intro:tv} for some small constant $\beta \approx
10^{-8}$. Then the above formulae are slightly more complicated. Since the
parameters $\lambda_2$ and $\lambda_{\text{B}}$ are not accessible
in practice as $u^0$ is unknown, we refer to $ \hat u(\lambda_2)$ and $\hat
u(\lambda_{\text{B}})$ as \emph{$\Ls{2}$}- and \emph{Bregman-oracle},
respectively. Simulations lead values $\lambda_2 = 0.026, 0.0789$ and
$\lambda_{\text{B}} = 0.0607, 0.1767$ for $\sigma = 0.1, 0.2$, respectively. 

In addition, we compare our approach to the \emph{spatially adaptive TV (SA-TV)}
method as introduced in \cite{DonHinRin11}. The SA-TV algorithm
approximates solutions of \eqref{intro:optprob} for the case where $\S$ constitutes the set
of all translates of a fixed window $S\subset G$ (cf. also
\cite{AlmBalCasHar08}) by computing a solution of \eqref{intro:rofspat} with a
suitable spatially dependent regularization parameter $\lambda$. Starting from a
(constant) initial parameter $\lambda \equiv \lambda_0$ the SA-TV algorithm
iteratively adjusts $\lambda$ by increasing it in regions that were poorly
reconstructed in the previous step. For our numerical comparisons, we used the
SA-TV-Algorithm considering square windows with side lengths $11$
(as suggested in \cite{DonHinRin11}) and $19$. All parameters
involved in the algorithm were chosen as suggested in \cite{DonHinRin11}. In
particular we set $\lambda_0 = 0.5$  and choose an upper bound for $\lambda$
of $L=1000$ in all our simulations. As a stopping condition, we used the
discrepancy principle which ended the reconstruction process after exactly four
iteration steps in all of our experiments.

The reconstructions are displayed in Figure \ref{results:figden10} ($\sigma =
0.1$) and  Figure \ref{results:figden20} ($\sigma = 0.2$). By visual inspection,
we find that the oracles are globally under- ($\Ls{2}$) and over-regularized
(Bregman), respectively. While the scalar parameter $\lambda$ was chosen
optimally w.r.t.\ the different distance measures, it still cannot cope with the
spatially varying smoothness of the true object $u^0$.

In contrast, SMRE and SA-TV reconstructions exhibit the desired locally adaptive
behaviour. Still, the SMRE as formulated in this paper has the advantage that
multiple scales are taken into account \emph{at once}, while SA-TV only adapts
the parameter on a single given scale. As a result, SA-TV reconstructions are of
varying quality for finer and coarser features of the object, while the SMRE is
capable of reconstructing such features equally well. This becomes particularly
obvious when zooming into the reconstructions (cf. Figure
\ref{results:details}).

\subsection{Deconvolution \& Inpainting}\label{results:inp}

We next investigate the performance of our approach if the operator $K$ in
\eqref{intro:data} is non-trivial. To be exact, we consider \emph{inpainting}
and \emph{deconvolution} problems. For the first we consider an inpainting
domain that occludes $15\%$ of the image with noise level $\sigma = 0.1$ (upper left
panel in Figure \ref{results:figinpdeconv}) and for the latter a Gaussian
convolution kernel with variance $2$ and noise level $\sigma = 0.02$ (lower left
panel in Figure \ref{results:figinpdeconv}).

 For all experiments we use the dyadic system $\S_2$ and $\alpha = 0.9$. Note
 that in both cases we have $K = K^*$ and $\norm{K} = 1$; we therefore set $\zeta =
 1.01$ in \eqref{impl:primal}. The results are depicted in the
 upper right and lower right images of Figure \ref{results:figinpdeconv}, respectively.

Again, the results indicate that a reasonable trade-off between data fit and
smoothing is found by the proposed a priori parameter choice rule and that the
amount of smoothing is adapted according to the image features.
 
\begin{figure}[h!]
\begin{center}
\includegraphics[width = 0.45\textwidth]{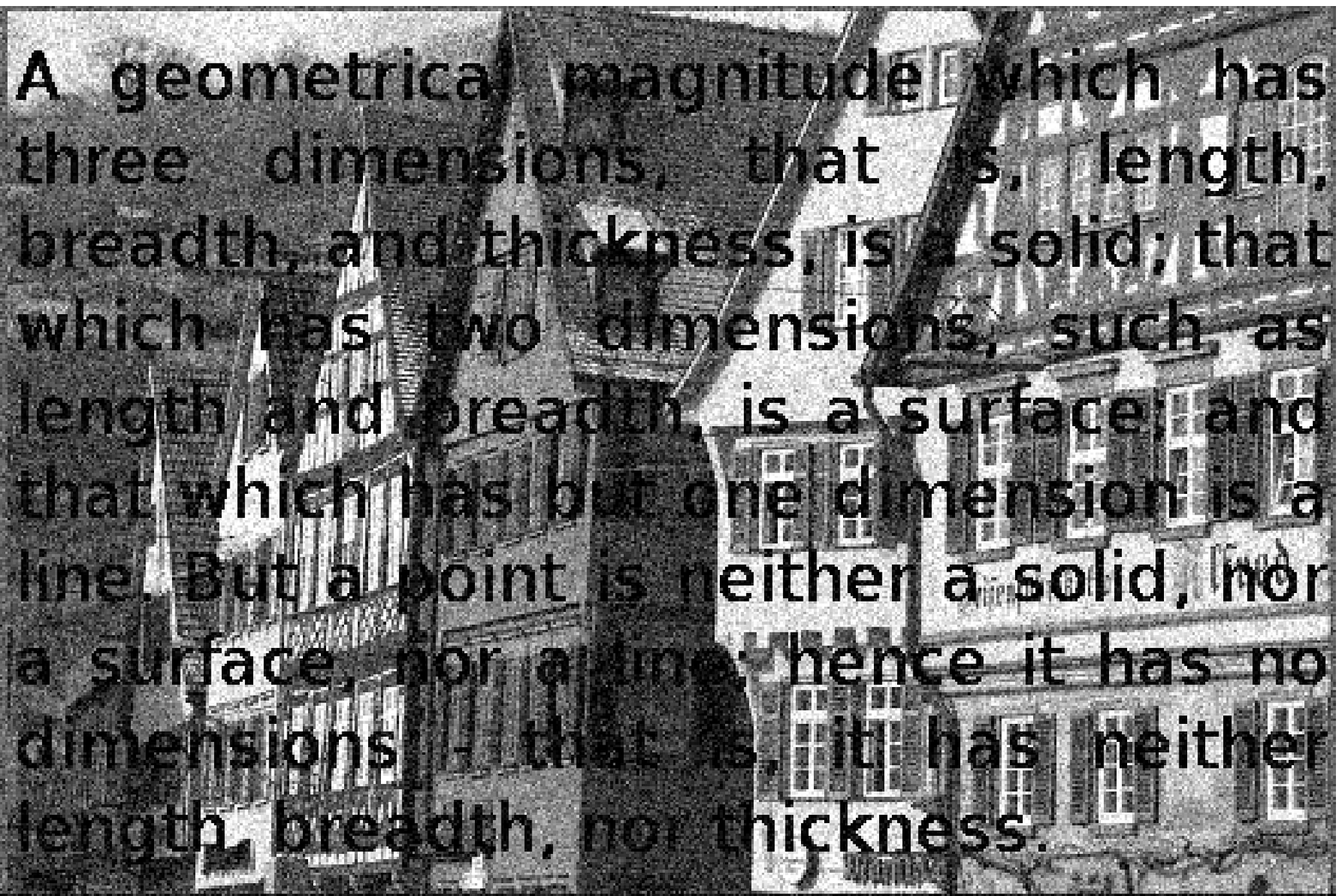}
\hspace{0.004\textwidth}
\includegraphics[width =0.45\textwidth]{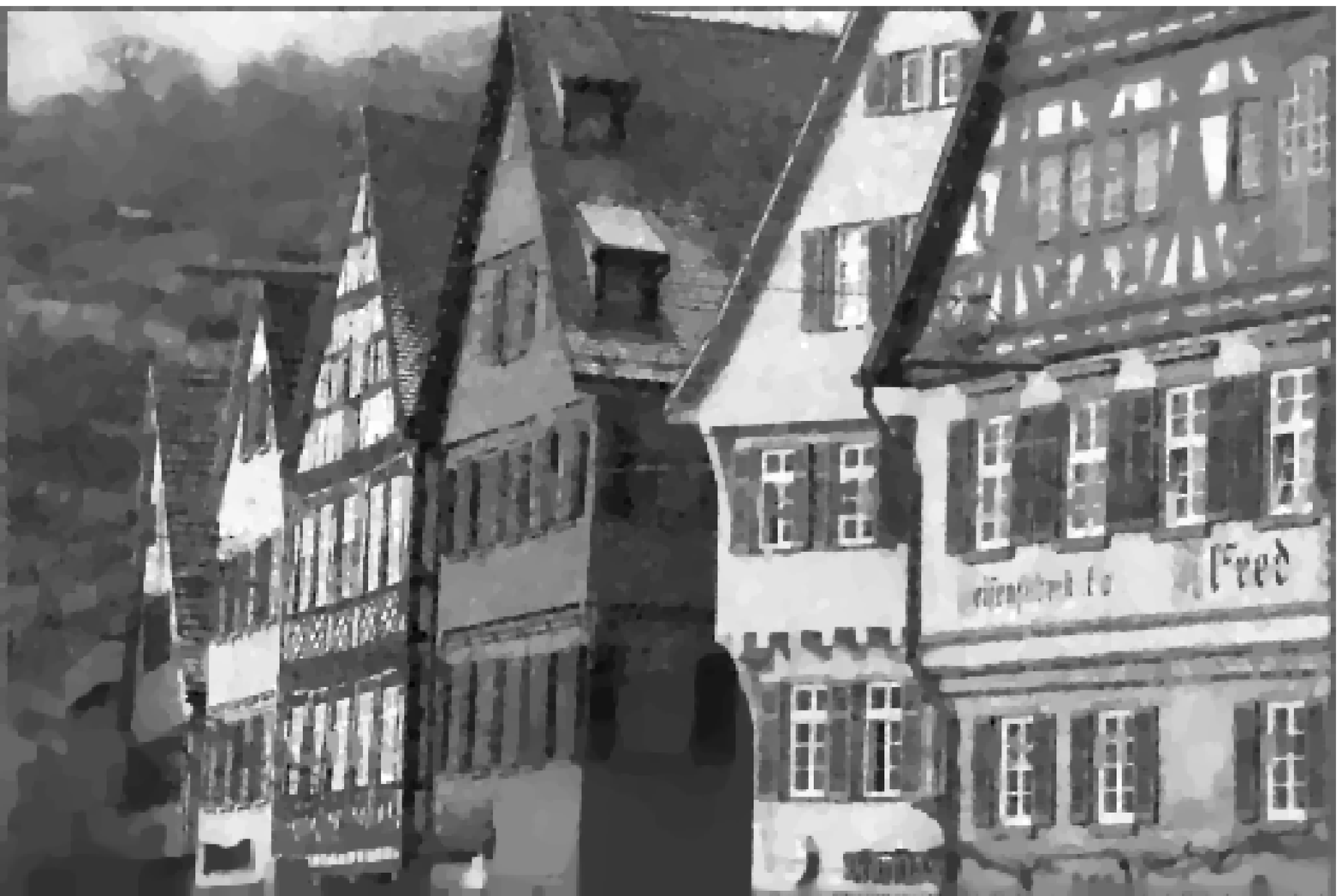}

\vspace{0.01\textwidth}
\includegraphics[width = 0.45\textwidth]{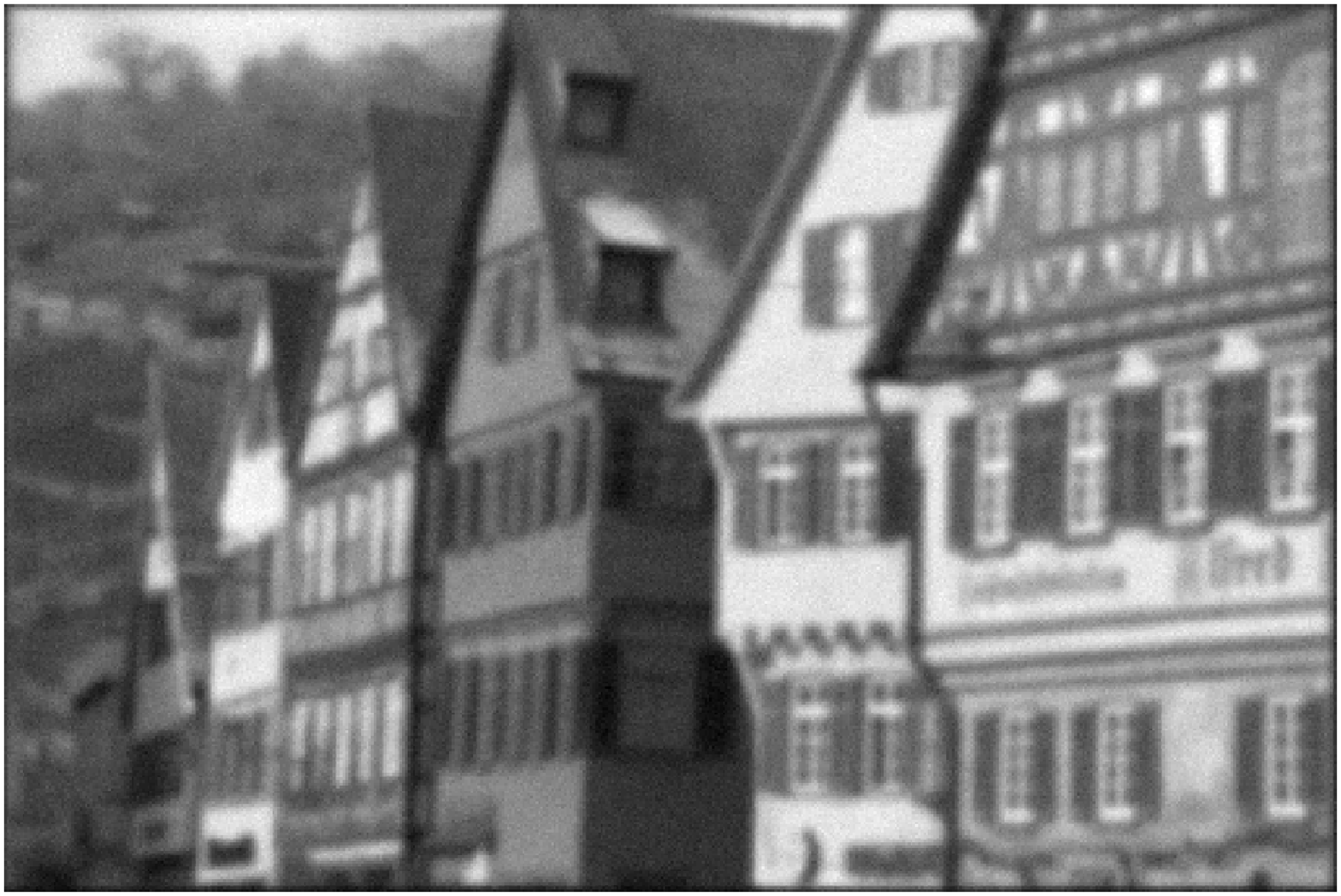}
\hspace{0.004\textwidth}  
\includegraphics[width 
=0.45\textwidth]{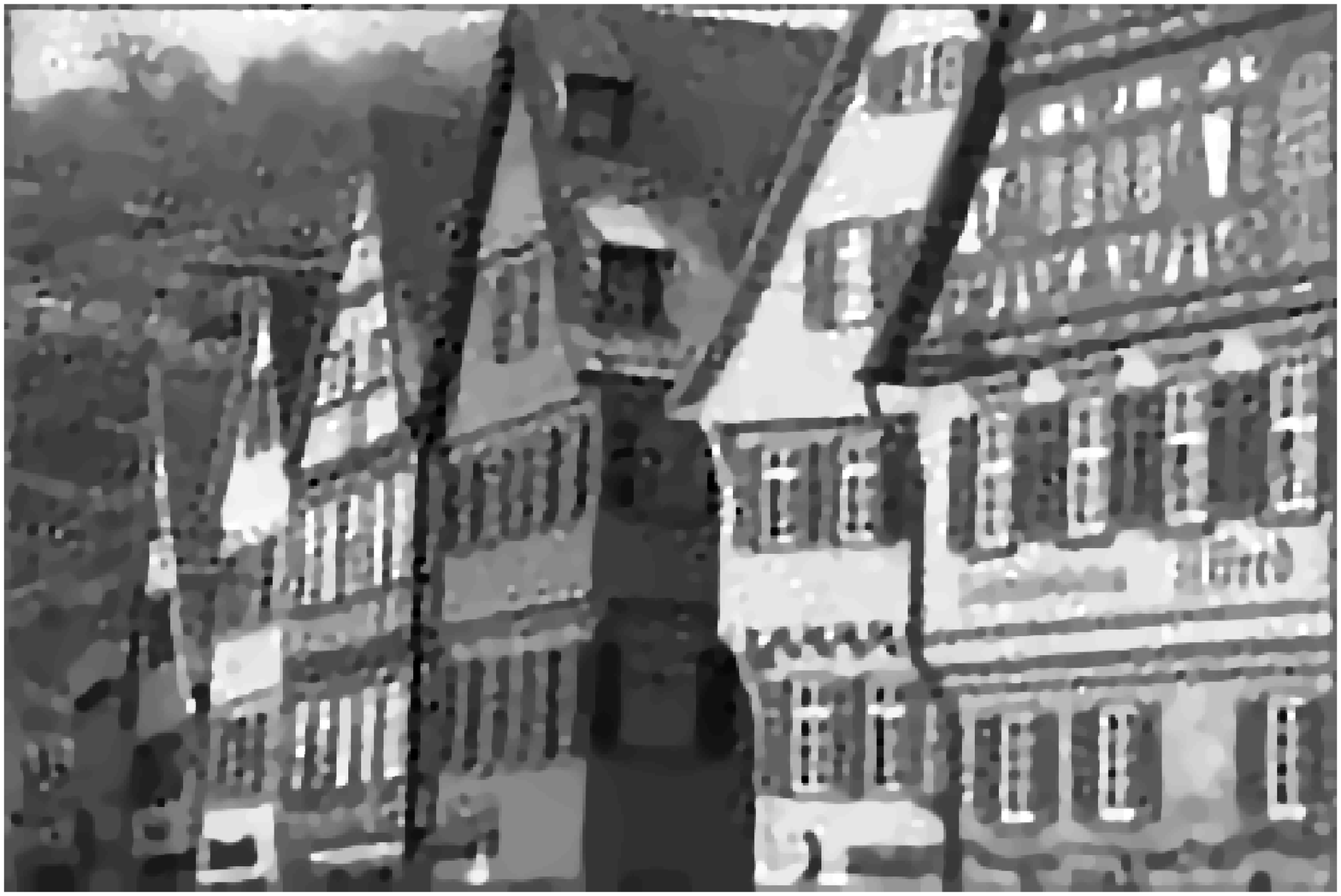} 
\end{center}  
\caption{Inpainting (top): data $Y$ with $\sigma = 0.1$ (left)
and SMRE (right). Deconvolution (bottom): data $Y$ with $\sigma = 0.02$
(left) and SMRE (right).}\label{results:figinpdeconv}
\end{figure} 

\subsection{Examples from Fluorescence Microscopy}\label{results:fluor}

We finally study the performance of our approach in a practical application,
namely fluorescence microscopy. To be more precise, we consider deconvolution
problems for standard confocal microscopy and STED (STimulated Emission
Depletion) microscopy. Both examples have in common, that the recorded data is a
realization of independent Poisson variables where the intensity at each pixel
is determined by a blurred version of true signal. In other words, Model
\eqref{poisson:data} applies. In both cases the blurring can be modelled (in
first order) as a convolution with a Gaussian kernel where the width of the
kernel for confocal microscopes is $3-4$ times larger than it is for STED. As
standard references we refer to \cite{Paw06} (confocal microscopy) and to
\cite{Hel94,Hel07} (STED).
 
In both cases, we will use sample images of  PtK2 cells taken from the kidney of
\emph{potorous tridactylus}, where beforehand the protein $\beta$-tubulin was
tagged with a fluorescent marker. What becomes visible is the microtubule part
of the cytosceleton of the cells. The left panels in Figures
\ref{results:confocal_full} and \ref{results:sted_full} show the confocal and
STED recordings, respectively. Both sample images show an area of $18\times
18$$\upmu$m$^2$ at a resolution of $798\times 798$ pixels. As a regularization
functional we use in both cases a combination of the total variation semi-norm
and the $\Ls{2}$-norm as in \eqref{intro:tvaug} with $\gamma = 1$.
  
\subsubsection{Confocal microscopy}\label{results:fluor:confocal}
 
Figure \ref{results:confocal_full} depicts a confocal recording of a PtK2 cell
(left) and the solution of \eqref{poisson:optprob} computed by Algorithm
\ref{poisson:uzawa} (right). We have used the subset of $\S_0$ with
maximal side length of $20$ pixel and the scale weights $c_S$ are chosen as in
Proposition \ref{smre:conf} with $\alpha = 0.9$. For the convolution kernel we
assume a full width at half maximum of $230$nm which corresponds to a standard
deviation of 4.3422 pixels. Due to this relatively large kernel, the impact of
the deconvolution is clearly visible. 

\begin{figure}[h!]
\begin{center}
\includegraphics[width = 0.48\textwidth]{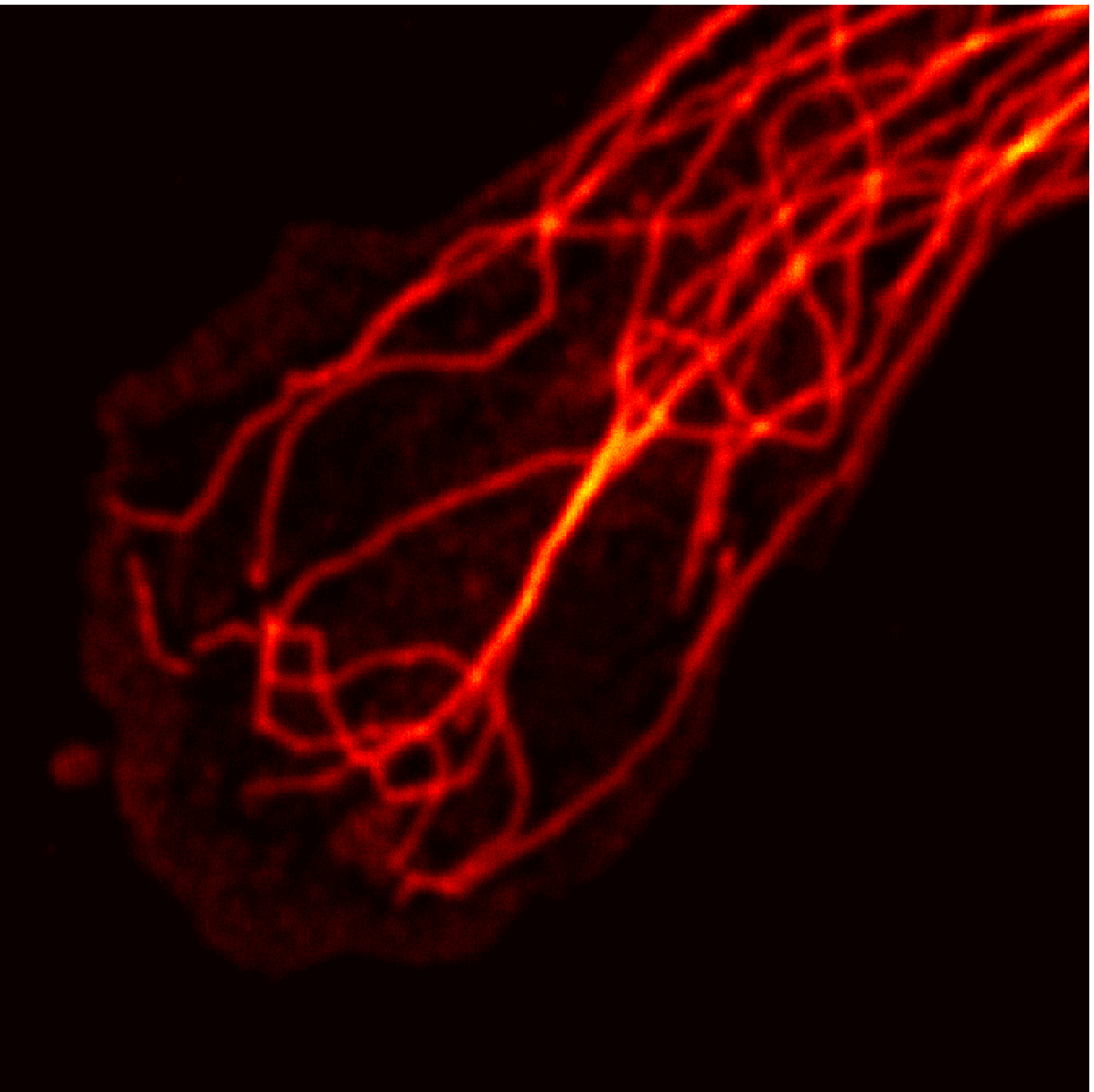}
\hspace{0.004\textwidth}
\includegraphics[width
=0.48\textwidth]{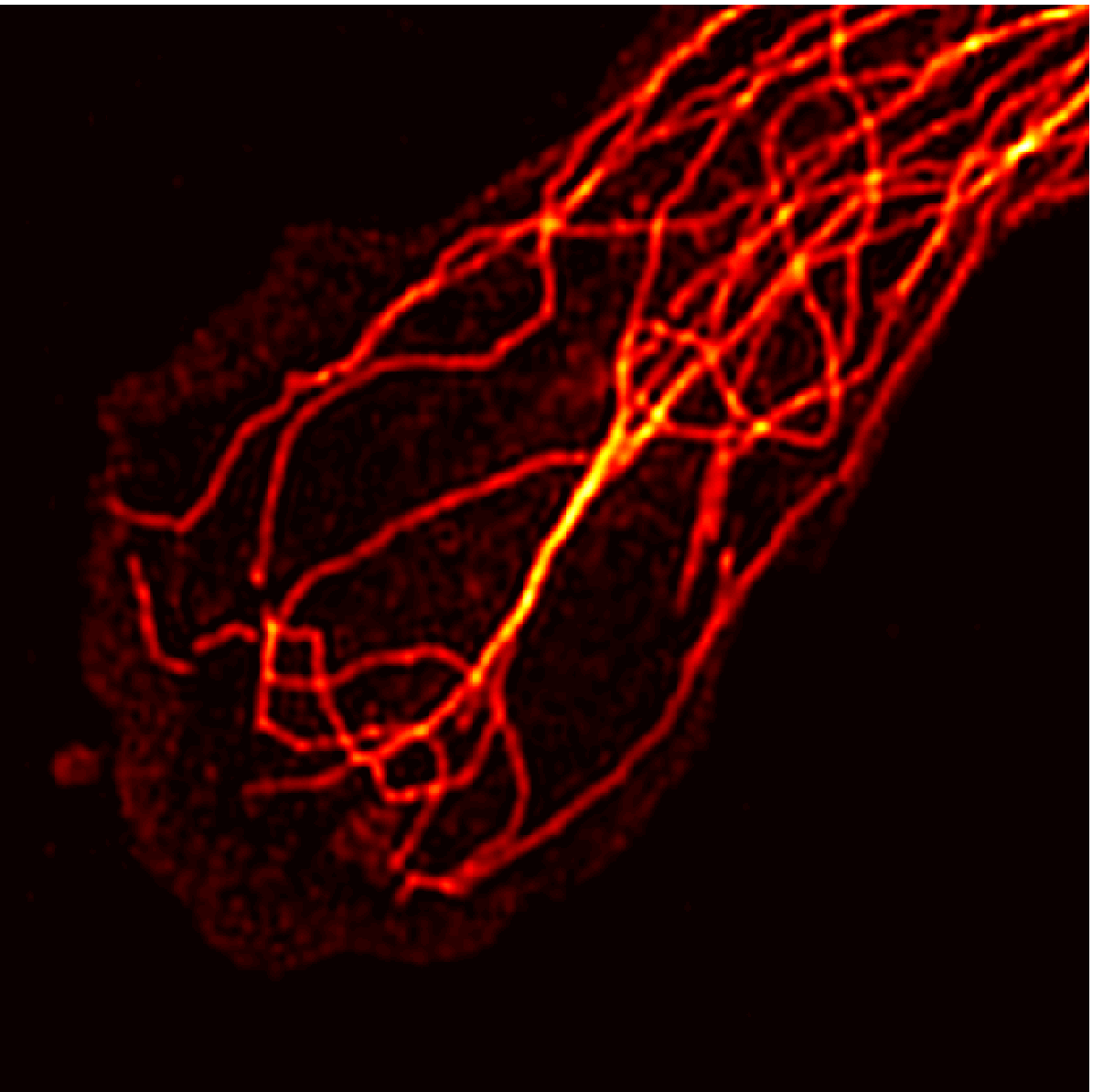}
\end{center}
\caption{Confocal deconvolution: data $Y$ (left) and SMRE
(right).}\label{results:confocal_full}
\end{figure}
 
This becomes remarkably apparent when one takes a closer look: In Figure
\ref{results:confocal_detail} a zoomed-in area of the data (left) and the SMRE
(right) is depicted. Clearly, the reconstruction reveals details that are
not visible (or at least not apparent) in the data. In particular, the
individual microtubule-filaments are separated properly. 

\begin{figure}[h!]
\begin{center}
\includegraphics[width = 0.48\textwidth]{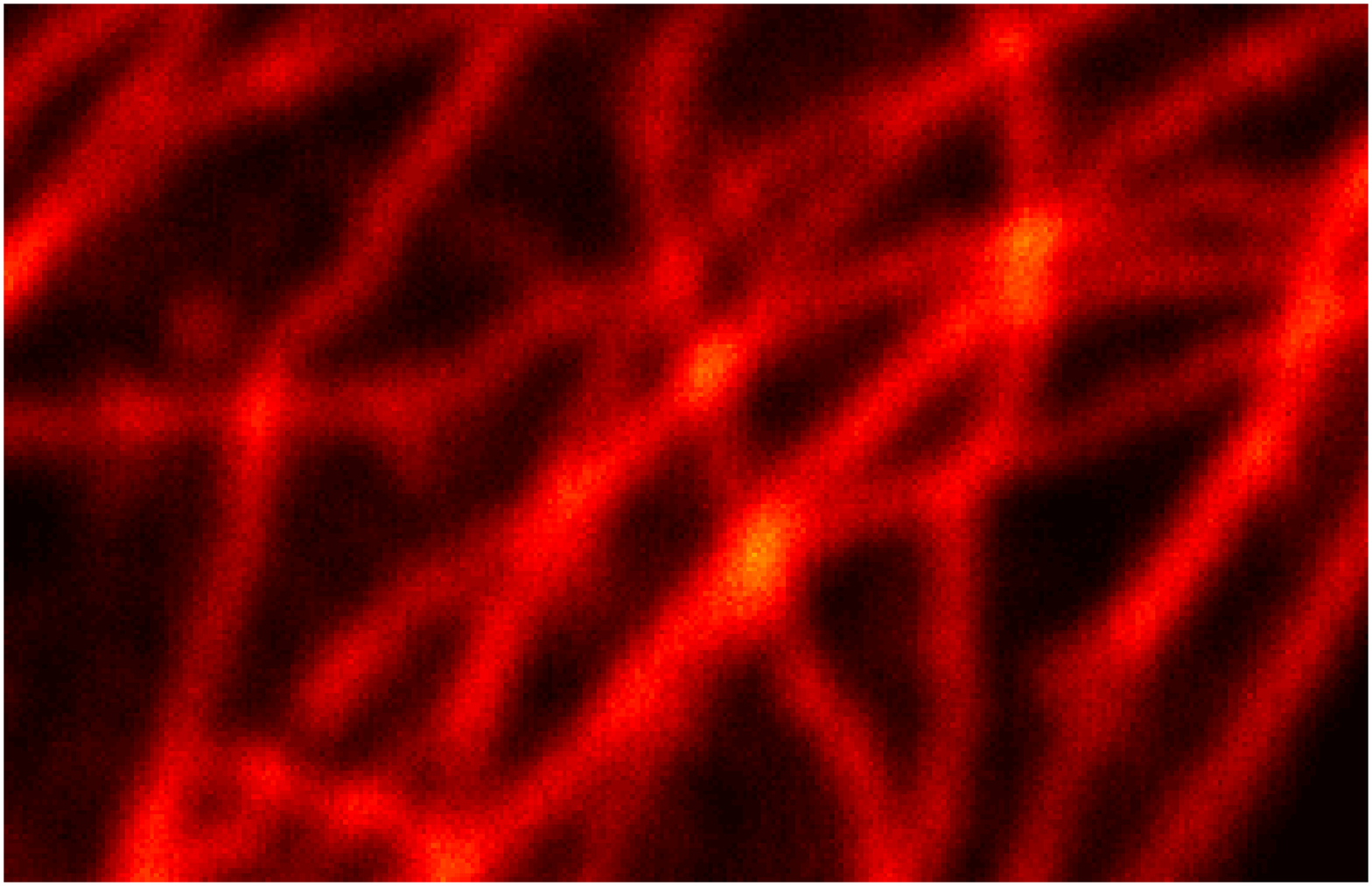}
\hspace{0.004\textwidth}
\includegraphics[width
=0.48\textwidth]{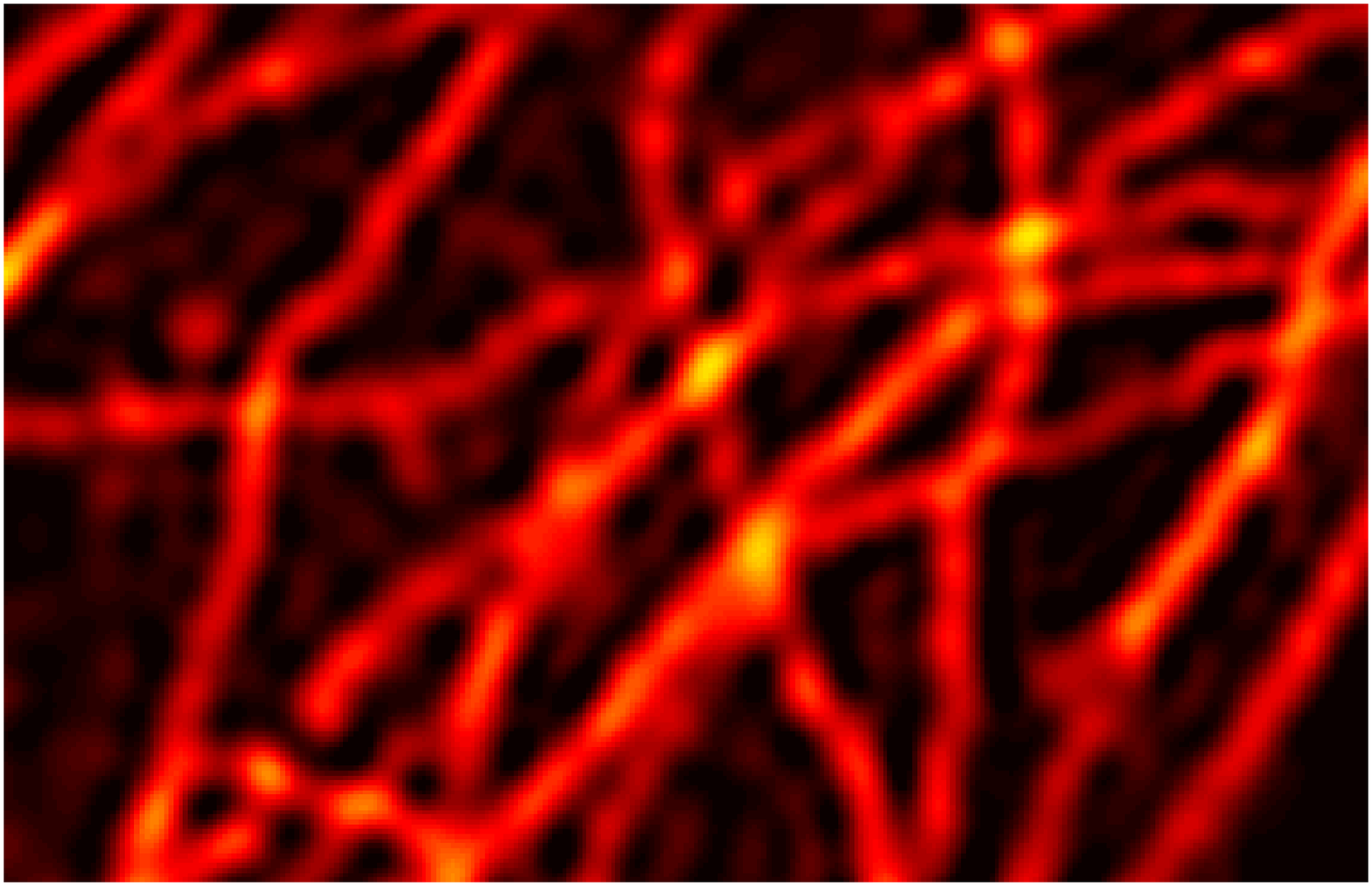}
\end{center}
\caption{Confocal deconvolution (closer look): data $Y$ (left) and SMRE
(right).}\label{results:confocal_detail}
\end{figure}
 
We finally remark, that we proposed a different multi-scale deconvolution method
for confocal microscopy in \cite{FriMarMun12}. In contrast to this work,
we there used a different MR-statistic than $T$ in \eqref{intro:mre} and we used
the standardization $(Y-\beta)\slash \sqrt{\beta}$ in order to transform the
Poisson data $Y$ to normality. The performance of the two approaches for
confocal recording is comparable since the image intensity (=photon
count rate) is relatively high throughout the data and hence standardization
yields a fair approximation to normality. However, for low-count Poisson
data, as we will investigate in the following section, our new approach is
clearly preferable, mostly due to the fact that Anscombe's transform also works
well for small intensities (cf. \cite{BroCaiZhaZhaZhou10}).

\subsubsection{STED microscopy}\label{results:fluor:sted}

The left panel of Figure \ref{results:sted_full} depicts a STED recording of a
PtK2 cell.  For the
convolution kernel a full width at half maximum of $70$nm is assumed, which
corresponds to a standard deviation of approximately $1.1327$ pixels. The right
image of Figure \ref{results:sted_full} depicts an approximate solution of
\eqref{poisson:optprob} that was computed by Algorithm \ref{poisson:uzawa}.
Again we use the system $\S$ of all squares in $\S_0$ up to a maximal side
length of $20$ pixels and the thresholds $c_S$ are chosen according to Proposition
\eqref{smre:conf} with $\alpha = 0.9$. Due to the relatively small convolution
kernel, the impact of the deconvolution is less striking as e.g. for the
confocal recording in Paragraph \ref{results:fluor:confocal}.

 \begin{figure}[h!]
\begin{center}
\includegraphics[width = 0.48\textwidth]{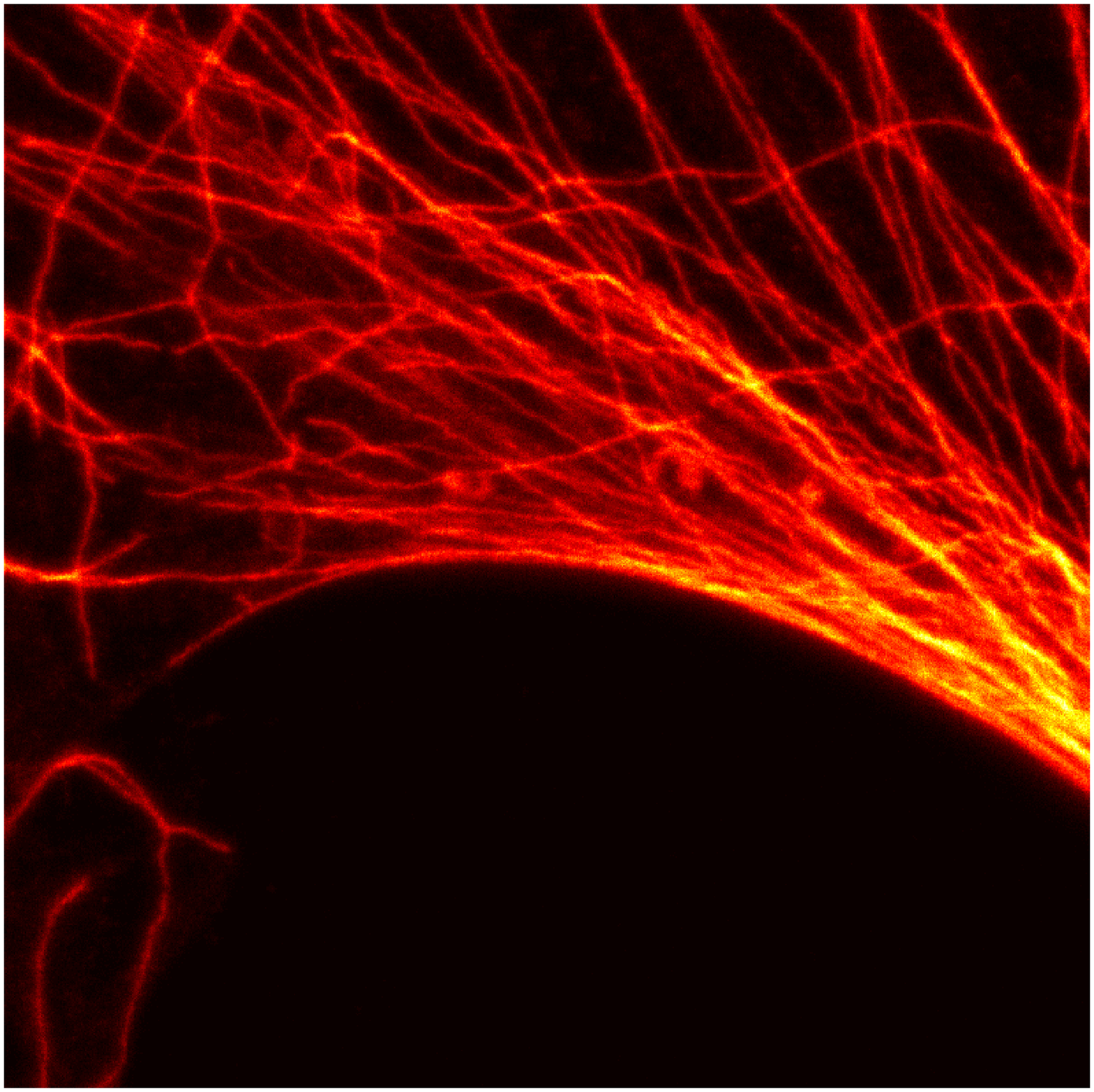}
\hspace{0.004\textwidth}
\includegraphics[width =0.48\textwidth]{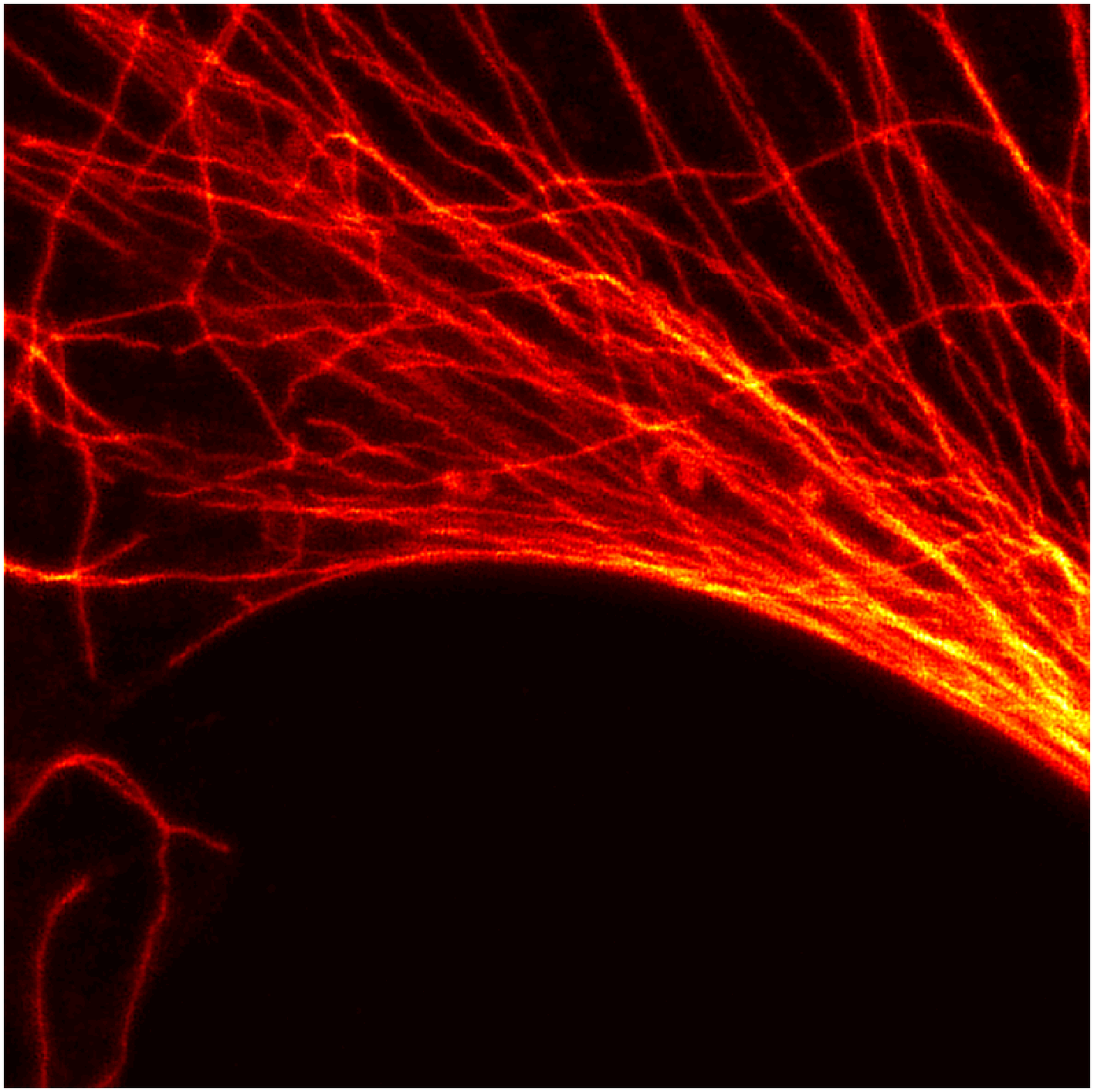}
\end{center}
\caption{STED deconvolution: data $Y$ (left) and SMRE
(right).}\label{results:sted_full}
\end{figure}

 When zooming into the image, the effect becomes more obvious: The left and
 middle panel in Figure \ref{results:sted_details} depicts a detailed view on
 the STED data and the SMRE from Figure \ref{results:sted_full}. The right
 panel in Figure \ref{results:sted_details} shows the \emph{global}
 reconstruction $\hat u_{\text{g}}$, i.e. we computed a solution of
 \eqref{poisson:optprob} w.r.t. to the trivial system $\S = \set{G}$ and the parameter $c_G$ as in Proposition \ref{smre:conf} with
$\alpha = 0.9$. The global reconstruction exhibits typical concentration
phenomena (especially in the upper half of the image) that are due to the
 ill-posedness of the deconvolution. These artefacts are less prominent for the
 SMRE solution.
 
\begin{figure}[h!]
\begin{center}
\includegraphics[height = 0.5\textwidth]{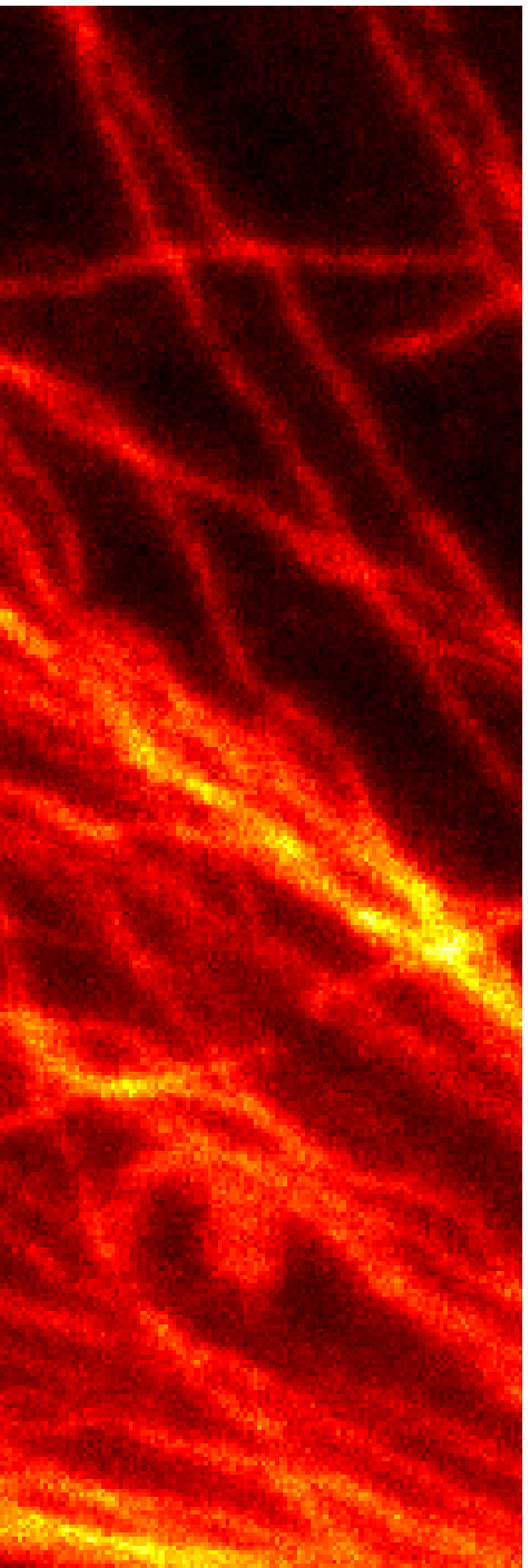}
\hspace{0.001\textwidth}
\includegraphics[height = 0.5\textwidth]{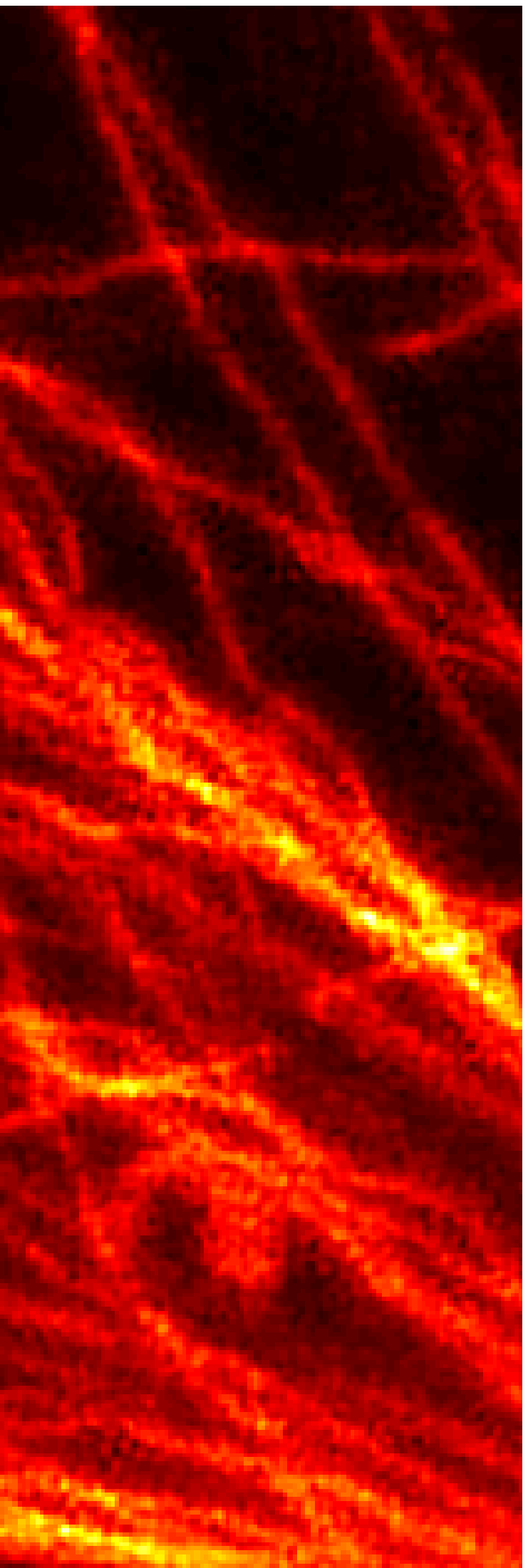}
\hspace{0.001\textwidth}
\includegraphics[height = 0.5\textwidth]{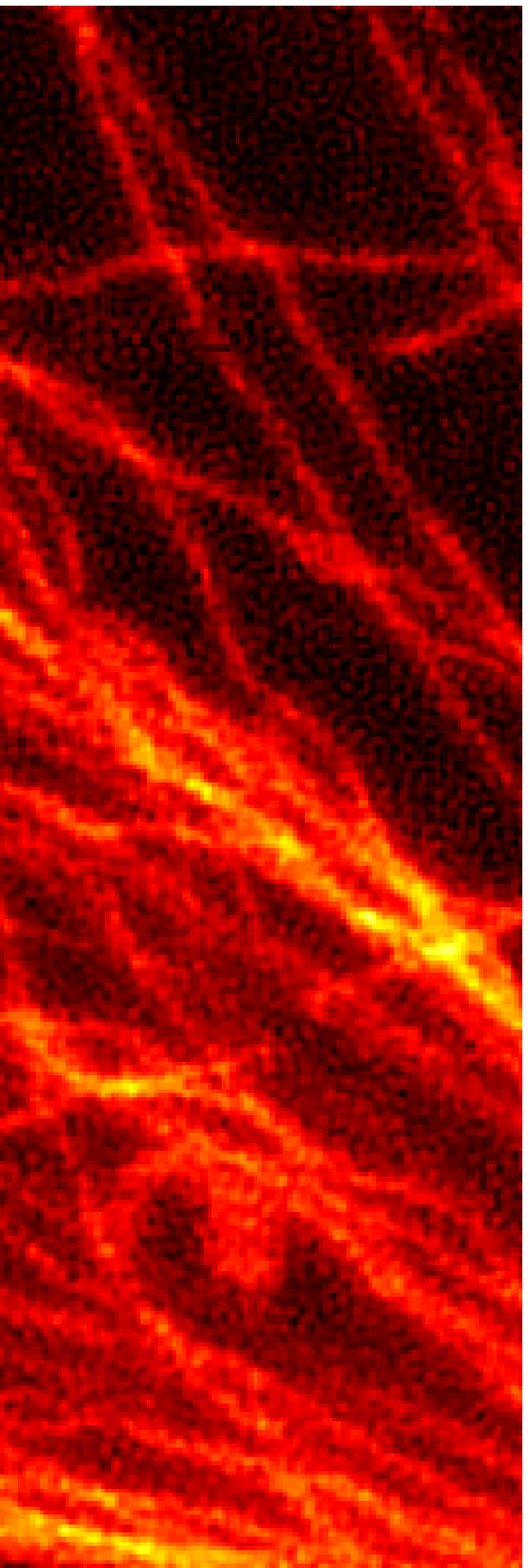}
\end{center}
\caption{STED deconvolution (detail): data $Y$ (left), SMRE (middle) global
reconstruction (right)}\label{results:sted_details}
\end{figure}

 At the same time, some parts of $\hat u_{\text{g}}$ are 
 oversmoothed. In Figure \ref{results:sted_compare} the union of the sets $S\in \S_0$ where
 \begin{equation*}
 c_S\sum_{(i,j)\in S} \abs{2\sqrt{Y_{ij} + 3\slash 8} -  2\sqrt{K\hat
 u_{\text{g}}}} > 1
 \end{equation*}
 are highlighted, where we restrict our consideration to the squares with a
 sidelength of $7$ pixel. A closer view on the highlighted subset confirms the
 oversmoothing (lower zoom-box). Again we note that \emph{at the same time} the
 global image reconstruction has artefacts due to the ill-posedness of the
 deconvolution (upper zoom-box). These can only be avoided by invoking stronger
 regularization. A comparison with the corresponding details of the SMRE (right) shows that our
 locally adaptive approach lacks this undesirable behaviour.  

\begin{figure}[h!]
\begin{center} 
\includegraphics[height = 0.5\textwidth]{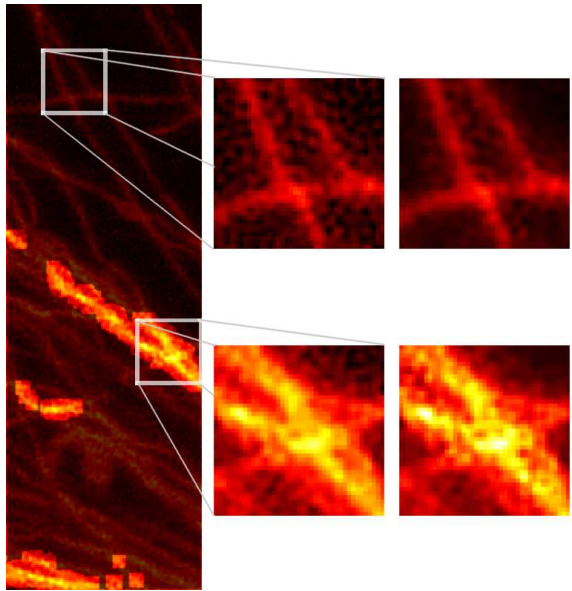}
\end{center}
\caption{Oversmoothed regions in the global estimator $\hat u_{\text{g}}$
identified on the scales $\abs{S} = 7$. Zoomed-in regions show over-
(lower zoom-box) and underregularized (upper zoom-box) parts in the image that
do not occur in the SMRE reconstruction (right boxes).}\label{results:sted_compare}
\end{figure}

\section{Conclusion} 

In this paper we show how statistical multiresolution estimators, that is
solutions of \eqref{intro:optprob}, can be employed for image reconstruction. We
stress that our method, combined with a new automatic selection rule, locally
adapts the amount of regularization according to the multi-scale nature of the
image features. For the solution of the optimization problem
\eqref{intro:optprob} we suggest an inexact alternating direction method of
multipliers combined with Dykstra's projection algorithm. We show how
this estimation paradigm can be extended to the Poisson model which opens up a vast field of
applications, such as Poisson nanoscale fluorescence microscopy. Aside to this
application, the performance of our method is illustrated for standard problems
in imaging such as denoising and inpainting.

\subsection*{Acknowledgments} K.F. and A.M. are supported by the DFG--SNF
Research Group FOR916 \emph{Statistical Regularization and Qualitative
Constraints} (Z-Project). P.M is supported by the BMBF project $03$MUPAH$6$
\emph{INVERS}. A.M.  and P.M. are supported by the SFB755 \emph{Nanoscale
Photonic Imaging} and the SFB803 \emph{Functionality Controlled by Organization
in and between Membranes}.  The authors are indebted to S.~Hell, A.~Egner and
A.~Schoenle  (Department of NanoBiophotonics, Max Planck Institute for
Biophysical Chemistry, G{\"o}ttingen and Laser Laboratorium G{\"o}ttingen) for
providing the microscopy data and for fruitful discussions. 

\bibliographystyle{abbrv}
\bibliography{lit}

\end{document}